\documentclass[12pt]{article}
\usepackage{multirow}
\usepackage[T1]{fontenc}
\usepackage{hyphenat}
\usepackage{graphicx}
\usepackage{color}
\usepackage{soul}
\usepackage{enumerate}
\usepackage{rotating}
\usepackage{url}
\usepackage{array}

\graphicspath{ {./graphs/} }

\usepackage{pgf,tikz}
\usetikzlibrary{graphs}
\usetikzlibrary{matrix}
\usetikzlibrary{backgrounds}
\usetikzlibrary{positioning}
\usetikzlibrary{shapes}
\usetikzlibrary{patterns}
\usetikzlibrary{fit}
\usepackage{fullpage}
\usepackage{amsmath, amsfonts, amssymb, amsthm, mathtools}

\newtheorem{theorem}{Theorem}[section]
\newtheorem{corollary}[theorem]{Corollary}

\newtheorem{lemma}[theorem]{Lemma}

\usepackage[authoryear]{natbib}

\usepackage[noend]{algorithmic}
\usepackage{algorithm, calc}
\usepackage[capitalise]{cleveref}
\usetikzlibrary{positioning}
\usepackage{dsfont}
\usepackage{multirow}
\usepackage{xspace}
\usepackage{hhline}
\usepackage{caption}
\usepackage{subcaption}

\newcommand{\N}{\mathbb{N}}

\newcommand{\Osymbol}{\textnormal{O}}

\newcommand{\EPTAS}{\textnormal{\sffamily EPTAS}\xspace}
\newcommand{\PTAS}{\textnormal{\sffamily PTAS}\xspace}
\newcommand{\FPTAS}{\textnormal{\sffamily FPTAS}\xspace}

\newcommand{\EPTASClass}{\EPTAS \xspace}
\newcommand{\PTASClass}{\PTAS \xspace}
\newcommand{\FPTASClass}{\FPTAS \xspace}

\newcommand{\PARTITION}{\textnormal{\sffamily PARTITION}\xspace}
\newcommand{\NUMERICALKDIEMNSIONAL}[1]{\textnormal{\sffamily NUMERICAL $#1$-DIMENSIONAL MATCHING WITH TARGET SUMS}\xspace}
\newcommand{\MUTUALEXCLUSIONSCHEDULING}{\textnormal{\sffamily MUTUAL EXCLUSION SCHEDULING}\xspace}

\newcommand{\NPClass}{\textnormal{\sffamily NP}\xspace}

\newcommand{\SNPHClass}{\textnormal{\sffamily Strongly NP-hard}\xspace}

\newcommand{\PClass}{\textnormal{\sffamily P}\xspace}

\DeclarePairedDelimiter\ceil{\lceil}{\rceil}
\DeclarePairedDelimiter\floor{\lfloor}{\rfloor}

\newcommand{\completekpartite}[1]{\textit{complete #1-partite}}
\newcommand{\blockgraph}{\textit{block graph}}
\newcommand{\kblockgraph}[1]{\textit{{#1}-block graph}}

\newcommand{\cliques}{\textit{cliques}}
\newcommand{\kcliques}[1]{\textit{{#1}-clique graph}}
\newcommand{\cmaxcost}{C_{\max}}
\newcommand{\optcmaxcost}{C^{OPT}_{\max}}

\newcommand{\machinejobrestriction}{\mathcal{M}_j}
\newcommand{\J}{J}
\newcommand{\M}{M}
\newcommand{\forest}{\textit{forest}}

\newcommand{\blocksymbol}{B}
\newcommand{\succesorssymbol}{D}

\newcommand{\cutvertices}{\textit{cut}}
\newcommand{\treewidth}{\textrm{tw}}

\newcommand{\UnmergedColorings}{RemainingColorSets}
\newcommand{\colorv}{col\_v}
\newcommand{\colorDescv}{col\_desc\_v}
\newcommand{\colorTargetDescv}{col\_target\_desc\_v}
\newcommand{\cardv}{card\_v}
\newcommand{\cardTargetDescv}{card\_target\_desc\_v}
\newcommand{\ColorsOfU}{col\_U}
\newcommand{\ColorsOther}{col\_desc\_U}
\newcommand{\ColorOfu}{col\_u}
\newcommand{\ColorsOtherPrim}{col\_desc\_u}
\newcommand{\ColorsUSrc}{col\_U'}
\newcommand{\ColorsOtherSrc}{col\_desc\_U'}
\newcommand{\ColorsUuTgt}{col\_Uu}
\newcommand{\ColorsOtherTgt}{col\_desc\_Uu}
\newcommand{\cardu}{card\_u}
\newcommand{\cardDescu}{card\_desc\_u}
\newcommand{\ColoringsObtainable}{MergingColorSets}
\newcommand{\ColorsInBlock}{col\_block}
\newcommand{\ColorsOtherInBlock}{col\_desc\_block}

\newcommand{\bettermaxclass}{\ceil{n/(m-1)}}

\newcommand{\NO}{\texttt{NO}}
\newcommand{\YES}{\texttt{YES}}

\begin{document}

\title{Approximation algorithms for job scheduling with block-type conflict graphs}

\author{Hanna Furma\' nczyk\footnote{Institute of Informatics,\ Faculty of Mathematics, Physics and Informatics,\ University of Gda\'nsk,\ 80-309 Gda\'nsk,\ Poland. \ e-mail: hanna.furmanczyk@ug.edu.pl},
Tytus Pikies\footnote{Department of Algorithms and System Modeling, Gda\'nsk University of Technology, Poland.
\ e-mail: tytpikie@pg.edu.pl},
Inka Soko\l{}owska\footnote{Theoretical Computer Science Department, Jagiellonian University, Poland.
\ email: inka.d.sokolowska@gmail.com; krzysztof.szymon.turowski@gmail.com},
Krzysztof Turowski\footnotemark[3] \footnote{Krzysztof Turowski's research was funded in whole by Polish National Science Center $2020$/$39$/D/ST$6$/$00419$ grant. For the purpose of Open Access, the author has applied a CC-BY public copyright license to any Author Accepted Manuscript (AAM) version arising from this submission.}}

\markboth{H. Furma\'nczyk, T. Pikies, I. Soko\l{}owska, K. Turowski}{Approximation algorithms for job scheduling with block-type conflict graphs}
\date{}

\maketitle

\begin{abstract}
The problem of scheduling jobs on parallel machines (identical, uniform, or unrelated), under incompatibility relation modeled as a block graph, under the makespan optimality criterion, is considered in this paper. 
No two jobs that are in the relation (equivalently in the same block) may be scheduled on the same machine in this model.

The presented model stems from a well-established line of research combining scheduling theory with methods relevant to graph coloring.
Recently, cluster graphs and their extensions like block graphs were given additional attention. 	
We complement hardness results provided by other researchers for block graphs by providing approximation algorithms.
In particular, we provide a $2$-approximation algorithm for $P|G = \blockgraph|\cmaxcost$ and a PTAS for the case when the jobs are unit time in addition. 
In the case of uniform machines, we analyze two cases.
The first one is when the number of blocks is bounded, i.e. $Q|G = \kblockgraph{k}|\cmaxcost$.
For this case, we provide a PTAS, improving upon results presented by D. Page and R. Solis-Oba.
The improvement is two-fold: we allow richer graph structure, and we allow the number of machine speeds to be part of the input.
Due to strong NP-hardness of $Q|G = \kcliques{2}|\cmaxcost$, the result establishes the approximation status of $Q|G = \kblockgraph{k}|\cmaxcost$.
The PTAS might be of independent interest because the problem is tightly related to the \NUMERICALKDIEMNSIONAL{k} problem.
The second case that we analyze is when the number of blocks is arbitrary, but the number of cut-vertices is bounded and jobs are of unit time.
In this case, we present an exact algorithm.
In addition, we present an FPTAS for graphs with bounded treewidth and a bounded number of unrelated machines.

The paper ends with extensive tests of the selected algorithms.
\end{abstract}


The sources of wealth in many current societies are skills and knowledge of the employees. 
To increase welfare, obstacles blocking utilization of the skills and potential should be removed. 
Imagine that there is some corporation lacking in the fields of inclusivity and diversity; hence, there is a huge risk that some talents are wasted.
To remedy the situation, the corporation decided to hire a consulting firm, at which we work.
At our disposal there is a group of inclusivity officers who can be assigned to perform group trainings.
The corporation has a hierarchical structure: a person can have some co-workers, subordinates, and a superior.
We like to assign the employees to the instructors in a way that no instructor will train persons that are in a direct business relationship.
This is due to the fact, that it will be good to avoid mixing the inclusivity problems with business hierarchy -- the hierarchy might interfere with learning. 
Also, providing a more confidential environment is valuable.
How to perform the training to minimize the time of solving the inclusivity problems, if:
\begin{enumerate}
	\item all the officers have equal skills and the persons require a similar amount of training,
	\item the employees require a different amount of training, 
	\item the skills of the officers are diverse, but can be grouped within teams,
	\item some employees have special needs and some officers can more efficiently train them?
\end{enumerate} 


We consider the problem of makespan minimization for job scheduling on parallel machines with a conflict graph.
Formally, an instance of the problem is characterized by $n$ jobs $\J=\{J_1,\ldots,J_n\}$, a set of $m$ machines $\M=\{M_1,\ldots,M_m\}$, and a \textit{conflict graph} $G=(\J, E)$, also known in the literature as the \textit{incompatibility graph}. 
The set of vertices of $G$ is exactly the set of jobs, and two vertices are adjacent if and only if the corresponding jobs cannot be processed on the same machine.
Due to the diverse machine environment, we define the \textit{processing time} of the job $J_j$ on the machine $M_i$, $i \in \{1, \ldots, m\}$.
\begin{itemize}
		\item   \textit{Identical machines}, denoted by $P$.
		        Here, the processing time of a job $J_j \in \J$ on every machine $M_i$ is identical and equal to $p_j \in \N$.
		\item   \textit{Uniform machines}, denoted by $Q$.
		        In this variant, $M_i$ runs with speed $s_i \in \N^+$ and each job $J_j \in \J$ has a processing requirement, denoted by $p_j \in \N^+$.
		        The processing time of $J_j$ on $M_i$ is equal to $p_j/s_i$.
			    For brevity, we assume that $s_1 \geq \cdots \geq s_m$.
		\item   \textit{Unrelated machines}, denoted by $R$. Here, the processing time of a job depends on a machine in an arbitrary way.
				In this variant, there are given $mn$ values $p_{i,j} \in \N^+$, defining the processing time of $J_j$ on $M_i$.
	\end{itemize}
A schedule is a function $\sigma\colon \J \rightarrow \M$ -- we assume that the machines start at once and process the jobs, without interruption, in any order.
For a given $\sigma$, a \textit{processing time} of a machine is the total processing time of jobs assigned to the machine.
The makespan of $\sigma$, denoted by $\cmaxcost(\sigma)$, is the maximum value among all processing times of the machines.
For a given instance of a scheduling problem, the smallest value of the makespan among all schedules is denoted by $\optcmaxcost$. 
For consistency with the scheduling literature, throughout the paper we mostly use $p_j$, $s_i$, and $p_{i, j}$. However in places where it is more convenient we also refer to $p$ and $s$ as functions from the sets of tasks, machines, or their Cartesian products, respectively, into the set of natural numbers.

We use the three-field Lawler notation $\alpha|\beta|\gamma$ \cite{lawler1982recent} with $\alpha\in \{P,Q,R\}$ describing types of machines, $\beta$ representing properties of jobs such as their processing times or their conflict graph, and
$\gamma$ denoting the objective, in our case $\gamma =\cmaxcost$.

Task scheduling problems can often be expressed in the language of graph coloring.
A $k$-\textit{coloring} of a graph $G=(V,E)$ is a function $c\colon V(G) \to \{1, \ldots, k\}$. 
A \textit{proper} coloring is one where 
no two adjacent vertices have been assigned the same color.
An \textit{equitable} coloring is a proper one such that the cardinalities of any two color classes differ by at most one. 
There is a natural relation between a schedule in the considered model and a coloring of a conflict graph.
In particular, if there exists an equitable $m$-coloring $c$, then $c$ determines an optimal schedule for the scheduling problem.

In the paper, we focus on block graphs since this class of graphs best reflects the hierarchical structure in corporations which is of great importance in our application.
A graph is a \textit{block graph} (also called  \textit{clique tree}) if every maximal $2$-connected component is a clique \cite{harary1963characterization}. 
Any maximal clique is called a \textit{block}.
By an abuse of the notation, we use the same term (block) for the set of the vertices inducing a block in a graph.
Under such a definition every vertex belongs to at least one block. 
If a vertex is a member of exactly one block, we call it \textit{simplicial vertex}, otherwise it is a \textit{cut-vertex}.
Additionally, by $\kblockgraph{k}$ we mean a block graph with at most $k$ blocks. 
Note that if $G$ is $\kblockgraph{k}$, then $\cutvertices(G) \le k - 1$, where $\cutvertices(G)$ denotes the number of cut-vertices in $G$. Note that there is no converse relation of this sort: for example, stars have exactly one cut-vertex (their root), but they have exactly $n - 1$ blocks.
Block graphs have a very handy representation called a \textit{block-cut tree} $T_G = (V_B \cup V_{cut}, E_G)$ (cf. \Cref{fig:block-cut-tree}).
Here, there are vertices of two types: $V_B$ is a set of vertices representing all blocks in $G$, and $V_{cut}$ is a set of vertices representing all cut-vertices.
An edge $\{u, v\}$ belongs to $E_G$ if and only if $u \in V_B$, $v \in V_{cut}$ and the cut-vertex represented by $v$ is contained in the block represented by $u$.
For simplicity, if $B \in V_B$ in $T_G$ then we also use $B$ to describe the relevant set of the vertices in $G$.
Similarly, we directly identify vertices in $V_{cut}$ with the respective cut-vertices in $G$.
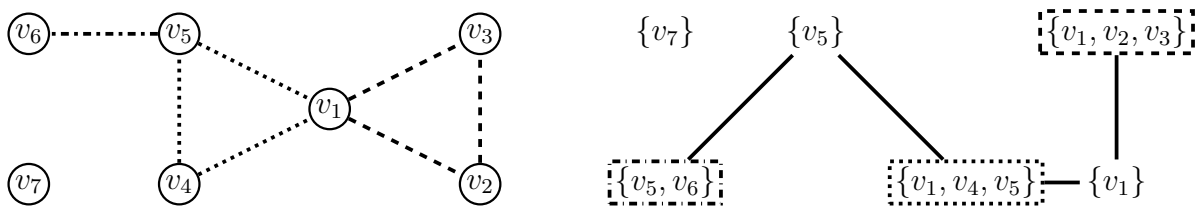
\begin{figure}[htpb]
	\centering
	\begin{subfigure}{0.45\textwidth}
		\begin{tikzpicture}
		[place/.style={circle,draw=black!100,line width=0.3mm,inner sep=1pt,minimum size=4mm},
		myline/.style={line width=0.5mm},
		matching/.style={line width=0.5mm},scale=2.0]
		\node at ( 0,1.5)   (v1) [place] {$v_1$};
		\node at ( 1, 1)   (v2) [place] {$v_2$};		
		\node at ( 1 ,2)   (v3) [place] {$v_3$};	
		
		\node at ( -1, 1)   (v4) [place] {$v_4$};
		\node at ( -1, 2)   (v5) [place] {$v_5$};
		
		\node at ( -2, 2)   (v6) [place] {$v_6$};
		\node at ( -2, 1)   (v7) [place] {$v_7$};
		
		\draw [myline, dashed]  (v1) -- (v3);
		\draw [myline, dashed]  (v1) -- (v2) ;
		\draw [myline, dashed]  (v2) -- (v3);
		
		\draw [myline, dotted]  (v4) -- (v5);
		\draw [myline, dotted]  (v4) -- (v1);
		\draw [myline, dotted]  (v5) -- (v1);
		
		\draw [myline, dash dot]  (v5) -- (v6);	
		\end{tikzpicture}
	\end{subfigure}
        \quad
	\begin{subfigure}{0.45\textwidth}
		\begin{tikzpicture}
		[place/.style={inner sep=2pt,minimum size=5mm,line width=0.5mm},
		myline/.style={line width=0.5mm},
		matching/.style={line width=0.5mm},scale=2.0]
		\node at ( -1, 2)   (v56) [rectangle,draw=black,dash dot,place] {$\{v_5,v_6\}$};
		\node at ( -1, 3)   (v7) [place] {$\{v_7\}$};
		
		\node at ( 0, 3)   (v5) [place] {$\{v_5\}$};
		
		\node at ( 1, 2)   (v145) [rectangle,draw=black,dotted,place] {$\{v_1, v_4, v_5\}$};
		
		\node at ( 2, 2)   (v1) [place] {$\{v_1\}$};
		
		\node at ( 2, 3)   (v123) [rectangle,draw=black,dashed,place] {$\{v_1, v_2, v_3\}$};
		
		\draw [myline]  (v56) -- (v5);
		\draw [myline]  (v5) -- (v145);
		\draw [myline]  (v145) -- (v1);
		\draw [myline]  (v1) -- (v123);
		\end{tikzpicture}
	\end{subfigure}
 \caption{A block graph $G$ (left) and its block-cut forest $T_G$ (right).  The respective blocks are marked with the same line type.}
 \label{fig:block-cut-tree}
\end{figure}
It is often convenient to assume that we work on rooted representations of $T_G$. 
Note that if block graph $G$ is not connected, then this representation is a forest, with one tree per every connected component of $G$.
Finally, observe that using a block-cut tree a block graph can be encoded in space $\Osymbol(n)$. 
A vertex representing a block corresponds to a list of vertices in the block. 
The total number of vertices that are not cut-vertex is $\Osymbol(n)$ and they contain $\Osymbol(n)$ vertices in total.
The total number of cut-vertices is $\Osymbol(n)$ due to the fact that each is associated with at least one distinct block. 

Let us return to the inclusivity problem presented earlier, and observe that it can be well modeled in the presented framework as a scheduling problem on block graphs:
\begin{enumerate}
    \item a training of the person can be modeled as a job,
    \item an instructor can be modeled as a machine,
    \item a hierarchy of employees can be modeled as a block graph:
a team has usually one leader, but each person in a team can be a leader for another team, lower in the hierarchy, hence a team can be represented as a block,
    \item the criterion of solving the problems as soon as possible can be modeled by $\cmaxcost$.
\end{enumerate}
This is only one sample application of scheduling in the presented model.
The model with other classes of graphs has many applications like scheduling jobs on a system with unstable power supply \cite{jansen2021total}, providing medical services during emergency \cite{pikies2022scheduling}, producing software under tight quality requirements \cite{pikies2022schedulingbiparite}.
Also, for an earlier review of possible use cases see \cite{KowalczykLAnExact2017}.

\section{Previous work and our results}

We refer the reader to standard textbooks for other notions of graph theory, scheduling theory, and approximation theory e.g. \cite{brucker2006scheduling,diestel2005graph,AusielloSCGPK99,cygan2015parameterized} for an overview of the respective topics. For more information about approximation schemes classes (\PTAS, \EPTAS, \FPTAS) see also \cite{EpsteinS2004ApproximationSchemes,jansen2019eptas,kones2019unified}.

Graham provided one of the first analyses of a scheduling problem, proving that the list scheduling is a constant approximation ratio algorithm for $P||\cmaxcost$ \cite{graham1966bounds}.
It is well-known that $P2||\cmaxcost$ problem is equivalent to the NP-hard  \PARTITION problem \cite{gareyJ1979computers}. 
However, it was also proved that $Q||\cmaxcost$ (and therefore $P||\cmaxcost$ as well) admits a polynomial time approximation scheme (PTAS) \cite{hochbaum1988} and for $Rm||\cmaxcost$ there exists a fully polynomial time approximation scheme (FPTAS) \cite{horowitz1976}.
Moreover, $Q|p_j = 1|\cmaxcost$ can be solved optimally in $\Osymbol(\min\{n + m \log{m}, n \log{m}\})$ time \cite{dessouky1990scheduling}.
For the most general case, namely $R||\cmaxcost$, a $(2 - \frac{1}{m})$-approximation algorithm was provided in \cite{shchepin2005optimal}. 
On the other hand, it was proved that there is no polynomial algorithm with an approximation ratio better than $\frac{3}{2}$, unless $\PClass = \NPClass$ \cite{lenstra1990approximation}.

The idea of incompatibility relation and conflict graph as understood in this paper was introduced by Bodlaender and Jansen in \cite{bodlaender1993complexity}, and first significantly developed by Bodlaender et al. in \cite{bodlaender1994scheduling}.
These authors developed a $2$-approximation algorithm for $P|G=bipartite|\cmaxcost$ running in $\Osymbol(n \log{m}$) time, and showed that there cannot exist a polynomial-time approximation algorithm for approximating the optimum makespan with worst-case ratio $2-\varepsilon$ for any fixed $\varepsilon>0$ unless $\PClass = \NPClass$, while for the special case of $P|G=tree|\cmaxcost$ there exists an algorithm in said running time finding a solution with length at most $5/3$ times the length of the optimum makespan.
They also showed that PTAS for the problems $P2|G|\cmaxcost$ and $Pm|\treewidth(G)\le k|\cmaxcost$ -- and we note that the latter result is especially relevant for this paper since $Pm|G=\blockgraph|\cmaxcost$ is a subproblem of $Pm|\treewidth(G)\le m|\cmaxcost$.

Clearly, each forest is also a block graph with every block corresponding to an edge. As \cite{pikies2022scheduling} noted, the problem $P|G, p_j = 1|\cmaxcost$ is closely tied to \MUTUALEXCLUSIONSCHEDULING problem, in which we are looking for a schedule so that no two jobs are connected by an edge in $G$ are executed at the same time.
Thus, from a polynomial time optimal algorithm for \MUTUALEXCLUSIONSCHEDULING problem, presented in \cite{bakerCMutualExclusion1996}, they inferred that $P|G = \forest, p_j=1|\cmaxcost$ can be solved optimality in polynomial time.

For the above results and throughout this paper there is assumed a customary notion that the number of jobs and machines are encoded in unary form, and that a schedule is a function from tasks to machines. Additionally, our encoding of choice for a block graph as a part of an input instance is a block-cut tree, defined above (see \Cref{fig:block-cut-tree}).
However, let us point out that there exists also another approach to encoding input instance and solution which results in a different boundary between easy and hard problems. In \cite{mallek2022scheduling} the authors assumed that the number of jobs and machines is encoded in binary (and the information about the conflict graph is provided in some concise form), and thus by reduction from \PARTITION problem they proved that $P2|G = \forest, p_j=1|\cmaxcost$ is NP-hard, but still if we applied a pseudopolynomial algorithm for \PARTITION we would get an algorithm with running time $\Osymbol(n m)$. We also note that the same difference occurred before for $Q|G = \completekpartite{2}, p_j = 1|\cmaxcost$ problem: on the one hand in \cite{pikies2022scheduling} it was proved that it admits a $\Osymbol(m n^3 \log(mn))$ algorithm when number of jobs and machines is encoded in unary form, but on the other hand \cite{mallek2019scheduling} showed that this problem is NP-hard, when $n$ and $m$ are encoded in binary, and the conflict graph is encoded in $\Osymbol(poly(\log{n}))$ bits.

Let us denote by $\cliques$ a graph consisting of disjoint cliques and by $\kcliques{k}$ a $\cliques$ with at most $k$ connected components.
It is natural to view $\blockgraph$ and $\kblockgraph{k}$ as a generalization of $\cliques$ and $\kcliques{k}$, respectively.
An EPTAS for $P|G=\cliques|\cmaxcost$ was provided by \cite{grage2019eptas}, improving upon the PTAS given by Das and Wiese \cite{das2017minimizing}.
In \cite{page2020makespan} it was proved that $Q|G=\cliques,p_j=1,\machinejobrestriction|\cmaxcost$ (where $\machinejobrestriction$ denotes that we additionally restrict the subset of machines available for each job) can be solved in polynomial time, while $Q|\kcliques{2}|\cmaxcost$  is strongly NP-hard; the hardness result can be easily extended to any fixed number of cliques.
The case of $R|G=\cliques|\cmaxcost$ has also been studied in  \cite{das2017minimizing}.
In that paper, it was shown that there is no constant approximation algorithm in the general case and even in some more restricted cases. 

Finally, let us turn our attention to the equitable coloring problem.
It was proved that this problem is strongly NP-hard for block graphs \cite{santos2019parameterized}, and as a direct consequence the problem $P|G=\blockgraph, p_j=1|\cmaxcost$ is strongly NP-hard as well.

Against this background, we focus our attention on the conflict graph being a block graph and we consider the problem of task scheduling with such a conflict graph in different variants. 
As it was mentioned above, already the problem of unit-time task scheduling on identical machines with the conflict block graph is strongly NP-hard, thus we mostly focus on finding good approximation algorithms for hard problems.

The organization of the paper is as follows:
we start in \Cref{sec:identical} with scheduling on identical machines. 
First, we give a $2$-approximation algorithm for the general case with $\Osymbol(n \log{m})$ running time.
Then, we consider the narrow case $P|G = \blockgraph, p_j = 1|C_{\max} \le k$ and develop a polynomial time algorithm for it.
These two algorithms, in turn, are important subprocedures in a PTAS for $P|G = \blockgraph, p_j = 1|\cmaxcost$.
In \Cref{sec:uniform} we turn our attention to the uniform machines. We analyze two cases: when the number of blocks is bounded and when this number is arbitrary, but the number of cut-vertices is bounded and jobs are of unit time. In the former case we provide a PTAS and in the latter one we devise an exact algorithm running in polynomial time. Note that the first mentioned result in a sense is the best possible, due to the strong NP-hardness of the problem \cite{page2020makespan}.
In \Cref{sec:unrelated} we consider unrelated machines and we provide an FPTAS for graphs with bounded treewidth and a bounded number of machines, which also captures a class of block graphs with bounded clique number.
Finally, in \Cref{sec:experiments} we present the results of extensive tests of our algorithms. 

Thus, our approximation algorithms and schemes extend the knowledge in the area of job scheduling with conflict graphs.
To help a reader, we provide a summary of our results in \Cref{tab:results}, together with previously known results for the relevant classes of conflict graphs.

\begin{table*}[!htb]
	\renewcommand{\arraystretch}{1.15}
	\setlength{\tabcolsep}{2pt}
	\center
	\footnotesize
	\begin{tabular}{|l l|c|c|c|c}
		\cline{1-5}
		\multicolumn{2}{|c|}{ \multirow{2}{*}{\bfseries Problem}} &  \multicolumn{2}{c|}{\bfseries Status} & \multirow{2}{*}{\bfseries Reference} \\
		\cline{3-4}
		&                           & \bfseries Approximation Ratio & \bfseries Time Complexity &                                       \\
		\cline{1-5}
  $P|G = \forest, p_j=1$ & $|\cmaxcost$ & optimal & polynomial & Baker, Coffman Jr. \cite{bakerCMutualExclusion1996} \\
  \cline{1-5}
  $P|G=\cliques$ & $|\cmaxcost$ & $1+\varepsilon$ & \EPTASClass & Grage et al. \cite{grage2019eptas}\\
  \cline{1-5}
		\multirow{2}{*}{$P|G = \blockgraph, p_j=1$} & \multirow{2}{*}{$|\cmaxcost$} & \multicolumn{2}{c|}{\SNPHClass} & Gomes et al. \cite{santos2019parameterized} \\ 
		\cline{3-5}
		&                              & $1+\varepsilon$ & \PTASClass & \bfseries \Cref{thm:ptas_identical_unit_block} \\ 
		\cline{1-5}
  $P|G = \blockgraph, p_j = 1$ & $|C_{\max} \le k$ & optimal & $\Osymbol(n^2 k^3 m^{7k + 8})$& \bfseries \Cref{theorem:final-k-coloring-block-graphs}
  \\
    \cline{1-5}
    $P|G = \blockgraph$ & $|\cmaxcost$ & $2$ & $\Osymbol(n \log{m})$ & \bfseries \Cref{thm:2apx_identical_block} \\ 
  \cline{1-5}
  $Pm|\treewidth(G)\le k$ & $|\cmaxcost$ & $1+\varepsilon$ & \FPTASClass & Bodlaender et al. \cite{bodlaender1994scheduling} \\
		\hhline{|==|==|=|}
		$Q|G = \blockgraph, cut(G) \le k, p_j = 1$ & $|\cmaxcost$ & optimal & $\Osymbol(m^{k + 2} n^2 \log (m n) )$ & \bfseries \Cref{thm:uniform_block_unit_cut} & \\ 
		\cline{1-5}
  $Q|G=\kcliques{2}$ & $|\cmaxcost$ & \multicolumn{2}{c|}{\SNPHClass} & Page, Solis-Oba \cite{page2020makespan} \\
  \cline{1-5}
		$Q|G = \kblockgraph{k}$ & $|\cmaxcost$ & $1+\varepsilon$ & \PTASClass & \bfseries \Cref{thm:uniform_kblock_cut} & \\ 
		\cline{1-5}
		\hhline{|==|==|=|}
  $R|G=\cliques$ & $|\cmaxcost$ & \multicolumn{2}{c|}{no constant approximation algorithm} & Das, Wiese \cite{das2017minimizing}\\
  \cline{1-5}
		$Rm|\treewidth(G)\le k$ & $|\cmaxcost$ & $1+\varepsilon$ & \FPTASClass & \bfseries \Cref{thm:arbitrary_tw} & \\
		\cline{1-5}
	\end{tabular}
	\caption
	{
		Summary of results relevant for scheduling with block conflict graphs. The results proved in this paper are highlighted in bold.
	}
	\label{tab:results}
\end{table*}

\section{Identical machines}
\label{sec:identical}

\subsection{General case}\label{general:2approx}
First, we present a $2$-approximation algorithm for $P|G = \blockgraph|\cmaxcost$ problem.
In fact, we prove that the following greedy algorithm works: let us perform a pre-order traversal of the block-cut tree $T_G$.
For each vertex representing a block of $G$ in $T_G$, we sort the machines non-decreasingly according to their current loads, and we assign to them the jobs from the current block, sorted by their non-increasing processing times. Additionally, we skip the machine to which the job represented by the cut-vertex associated with the parent of the current node in $T_G$ is assigned.
To achieve better running time in \Cref{alg:2apx_identical_block}, instead of sorting the machines at each step we keep a min-heap $H$ of machines ordered by their current loads and update it for every processed block.

\begin{algorithm}[htpb]
\begin{algorithmic}
		\REQUIRE A set of jobs $\J$, a set of machines $\M$, a block graph $G$, a function $p\colon \J \to \N$.
		\STATE Find the block-cut forest $T_G = (V_B \cup V_{cut}, E_G)$ of $G$.
		\STATE Initialize a min-heap $H$ with $(M_i, 0)$ for all $M_i \in \M$.\hfill\COMMENT{All machines are initially empty.}
		\STATE $V_B^O \gets$ an ordering of $V_B$ given by the pre-order traversal of all components of $T_G$
		\FORALL{$B \in V_B^O$}
		\STATE Sort jobs $J_j \in B$ in $G$ by their $p_j$ non-increasingly as $L_J$
		\STATE Retrieve $|B|$ machines with the smallest current loads from $H$ as $L_M$
		\IF{$B$ has a parent $u$ in $T_G$ already scheduled to machine $M'$}
		\IF {$M'\in L_M$}
		\STATE Remove $M'$ from $L_M$ and add it with its unmodified load to $H$
		\ELSE
		\STATE Remove the last machine from $L_M$ and add it with its unmodified load to $H$
		\ENDIF
		\STATE Remove $u$ from $L_J$
		\ENDIF
		\FOR{$i = 1, 2, \ldots, |L_M|$}
		\STATE Assign $i$-th job from $L_J$ to $i$-th machine from $L_M$ (updating its load)
		\STATE Add the $i$-th machine from $L_M$ with its current load to $H$
		\ENDFOR
		\ENDFOR
	\end{algorithmic}
	\caption{A greedy algorithm for $P|G=\blockgraph|\cmaxcost$}
	\label{alg:2apx_identical_block}
\end{algorithm}

Let $C_j(M_i)$ be the load of $M_i$ after processing $j$-th block in the sequence $V_B^O$, i.e. in an ordering of $V_B$ given by the pre-order traversal of all components of $T_G$. Let also $C_j = \frac{1}{m} \sum_i C_j(M_i)$.
To arrive at the approximation ratio we need the following lemma.
\begin{lemma}
	For any $i \in \{1, \ldots, m\}$ and for any $j \in \{0, \ldots, k\}$, the strategy preserves the invariant $C_j(M_i) \le C_j + \max\{C_j, p_{max}\}$. 
	\label{lem:strategy}
\end{lemma}
\begin{proof}
	We proceed by induction on $j$.
	Clearly, for $j = 0$, the claim holds trivially, since $C_0(M_i) = 0$ for every machine $M_i$.
	
	Let us assume that it holds for all $j' = 0, \ldots, j - 1$.
	Without loss of generality let us assume that the machines were numbered before processing $j$-th block according to the non-decreasing sum of loads (with ties broken arbitrarily). In particular, we assume that $C_{j - 1}(M_i)$ was the $i$-th smallest number in the multiset $\{C_{j - 1}(M_l)\}_{l = 1}^m$.
	We only need to prove the induction step for the single $i$ such that $C_j(M_i)$ is the largest among all $i \in \{1, 2, \ldots, m\}$ -- as all the other machines have loads $C_j(M_l) \le C_j(M_i)$.
	
	Clearly if we do not assign any job from $j$-th block to $M_i$, then $C_j(M_i) = C_{j - 1}(M_i)$, $C_j \ge C_{j - 1}$, and the invariant holds since
	\begin{align*}
		C_j(M_i) - C_j & \le C_{j - 1}(M_i) - C_{j - 1} 
		 \le \max\{C_{j - 1}, p_{max}\} \le \max\{C_j, p_{max}\}.
	\end{align*}
	
	Consider now a job from $j$-th block with processing time $p$ at the moment when we assign it to $M_i$. Clearly, $C_j(M_i) = C_{j - 1}(M_i) + p$.
	
	Moreover, by the ordering of machines we have $C_{j - 1}(M_l) \ge C_{j - 1}(M_i)$ for all $l \in \{i, \ldots, m\}$.
	Thus,
	\begin{align*}
		C_{j - 1}(M_i) \le \frac{m}{m - i + 1} C_{j - 1}.
	\end{align*}
	
	On the other hand, clearly $C_j \ge C_{j - 1} + (i - 1) p/m$, since if we assigned another job to $M_i$, it has to be the case that all machines in set $\{M_1, \ldots, M_i\}$ but at most one (removed from $L_M$) got assigned new jobs with processing times at least $p$. Thus,
	\begin{align*}
		C_j(M_i) & = C_{j - 1}(M_i) + p \le \frac{m}{m - i + 1} C_{j - 1} + p 
		 \le \frac{m}{m - i + 1} \left(C_j - \frac{(i - 1) p}{m}\right) + p 
	   \le \frac{m}{m - i + 1} C_j + \left(1 - \frac{i - 1}{m - i + 1}\right) p \\
		& \le C_j + \frac{i - 1}{m - i + 1} C_j + \left(1 - \frac{i - 1}{m - i + 1}\right) p_{max}
		\le C_j + \max\{C_j, p_{max}\}.
	\end{align*}
\end{proof}

\begin{theorem}
	There exists a $2$-approximation algorithm for $P|G=\blockgraph|\cmaxcost$.
	\label{thm:2apx_identical_block}
\end{theorem}
\begin{proof}
	The claim follows from the lemma above since for a block graph with $k$ blocks it holds both that $C_k \le \optcmaxcost$ and $p_{max} \le \optcmaxcost$ -- and for the constructed schedule $S$, $\cmaxcost(S) = \max_i C_k(M_i)$.
\end{proof}

The hard example which achieves the approximation ratio equal to $2 - \frac{1}{m}$ is as follows: let us have a conflict graph with $m$ cliques with $m - 1$ jobs with processing times equal to $\frac{p}{m}$ each, and one clique with one job with processing time $p$.
Clearly, in the optimum solution we assign the first cliques to $m - 1$ machines and the last one to the last machine, thus $\optcmaxcost = p$.

On the other hand, if we process the blocks in this order, then \Cref{alg:2apx_identical_block} after processing the first $m$ cliques assigns to each machine an equal load of $p \left(1 - \frac{1}{m}\right)$.
Thus, regardless of where the last job is assigned, the total makespan is equal to $p \left(2 - \frac{1}{m}\right)$.

Observe that the algorithm can be implemented in $\Osymbol(n \log{m})$ time. 
At each $B \in V_B$ with $|B|$ jobs we sort them in $\Osymbol(|B| \log{|B|})$ time. 
Then, we retrieve the first $|B|$ machines from the heap (with checking and omitting the one associated with a job represented by cut-vertex above) in $\Osymbol(|B| \log{m})$ time. 
Finally, we assign the jobs to the machines in the proper order and re-add the machines with updated loads to the heap in $\Osymbol(|B| \log{m})$ time.
Since $\forall_{B \in V_B} |B| \le m$ and $\sum_{B \in V_B} |B| = n$, the total running time follows directly.

\subsection{Unit time jobs, bounded makespan} \label{unit:bounded}
Let us first consider a version of the problem with a restricted makespan, i.e. $P|G = \blockgraph, \allowbreak p_j = 1|\cmaxcost \le k$. 
Note that it is equivalent to deciding whether there exists a coloring of $G$ where there are at most $k$ vertices in any color.
This problem may look very narrow, but it is an essential stepping stone toward considerations of the general problem for unit-time jobs.
In fact, we solve it by determining all feasible distinct colorings, that is, the colorings with different multisets of cardinalities of color classes.

First, let us introduce some auxiliary notions.
For a given block graph $G = (\J, E)$, we use the notion of block-cut tree $T_G$ to define the so-called \textit{sets of descendants} in the following way:
\begin{itemize}
	\item $u \in \J$ is a \textit{descendant} of a given cut-vertex $v \in V(T_G)$ if there exists a block $B \in V_B$ such that $u \in B$ and $B$ is in the subtree of $T_G$ rooted in $v$,
	\item $u \in \J$ is a \textit{descendant} of a given subset $U \subseteq B \setminus \{v\}$ for a block $B \in V_B$ and its parent cut-vertex $v \in B$ if either $u \in U$ or $u$ is a \textit{descendant} of some cut-vertex $u' \in U$.
\end{itemize}
We denote the sets of all descendants of $v$ and $U$ as $\succesorssymbol(v)$ and $\succesorssymbol(U)$, respectively. Observe that $v \in \succesorssymbol(v)$ for any $v \in V_{cut}$.

Additionally, we can assume that the block-cut tree $T_G$ is a \textit{plane tree}, i.e. there is an order of children for its every vertex. We call $u \in \J$ a \textit{$d$-th descendant} of a given $v \in V_{cut}$ and $d \in \N^+$ if there exists a block $B \in V_B$ such that $u \in B$ and $B$ is in the subtree of $T_G$ rooted in $d$-th child block of $v$.
We define $\succesorssymbol_d(v)$ to be the set of $d$-th descendants of $v$.
Similar as before, it is the case that $v \in \succesorssymbol_d(v)$.

As w stated above, ultimately we want to find all feasible distinct colorings for $G = G[\succesorssymbol(r)]$.
In order to achieve this, we will also compute all feasible distinct colorings for:
\begin{itemize}
    \item $G[\succesorssymbol(v)]$ for all $v \in V_{cut}$,
    \item $G[\succesorssymbol_d(v)]$ for all $v \in V_{cut}$ and all possible $d \in \N^+$,
    \item $G[\succesorssymbol(B \setminus \{v\})]$ for all blocks $B \in V_B$ and their parent vertices $v \in V_{cut}$ in $T_G$.
\end{itemize}
Our main algorithm (see \Cref{algorithm:main-loop}) does a post-order traversal of $T_G$ and eventually it constructs sets of all feasible colorings for $G[\succesorssymbol(v)]$, $G[\succesorssymbol_d(v)]$, and $G[\succesorssymbol(B \setminus \{v\})]$ in a recursive fashion.

The crucial idea in the algorithm is a compact way of representing colorings for subgraphs of $G$ mentioned above and identification (i.e. treatment as being the same) of some colorings. 
Note that we can only keep the set of cardinalities of all color classes that are assigned to a given object (be it a cut-vertex or a subset of vertices in a block) and that are assigned to the further descendants.
In addition, we have to keep track of the cardinalities of the color classes which are in use by vertices from the currently processed object.
Formally, we can define a \textit{pattern} for a coloring $c$ as a pair of vectors $(a, b)$ with $a, b \in \{0, \ldots, m\}^{k + 1}$ such that if $c$ is a coloring of $G[\succesorssymbol(U)]$ (respectively, $G[\succesorssymbol(v)]$ or $G[\succesorssymbol_d(v)]$), then $a_i$ denotes the number of colors of cardinality $i$ which are used by vertices in $U$ (respectively, by a cut-vertex $v$) and $b_i$ denotes the number of all other colors of cardinality $i$ in $c$ (i.e. the ones used only in a set of descendants of $U$).
Of course, in the case $G[\succesorssymbol(v)]$ or $G[\succesorssymbol_d(v)]$ there is only a single color assigned to $v$ -- but for consistency of the notation we do not differentiate this case.

We denote by $P(U)$ (respectively, $P(v)$ and $P_d(v)$) the set of all different patterns for $G[\succesorssymbol(U)]$ (respectively, $G[\succesorssymbol(v)]$ or $G[\succesorssymbol_d(v)]$). 
Such set contains all distinct colorings of $G[\succesorssymbol(U)]$ (respectively, $G[\succesorssymbol(v)]$ or $G[\succesorssymbol_d(v)]$).
Observe that there are $\Osymbol(m^{2 k + 2})$ different patterns for any subset $U$ of any block and at most $\Osymbol(k m^{k + 1})$ patterns for any $v \in V_{cut}$ since its $a$ has to contain a single one at the position $l$ and exactly $k$ zeroes (denoted by $\mathbf{1}^l$).

We can also talk of a pattern for simplicial vertices in an analogous manner. There would be only a single one for each such vertex, described by $a = \mathbf{1}^1$ (i.e. exactly one color of cardinality $1$ used for this vertex) and $b = (m - 1) \cdot \mathbf{1}^0$ (i.e. $m - 1$ colors of cardinality $0$, as they are not used anywhere).

In essence, our algorithm consists of three procedures that run recursively.
\begin{itemize}
    \item The first procedure (\Cref{algorithm:composing-subcolorings-for-subblockgraphs}) is run for each cut-vertex $v \in V(T_G)$ to merge all the patterns of all its descendant blocks.
    At step $d$-th the set $P_d(v)$ is merged with the set of all patterns for previous descendants $MP(v, d-1)$ ($MP(v, 0) = (\mathbf{1}^0, (m-1) \mathbf{1}^0)$).
    After considering all descendants the set $P(v)$ is finally produced.
	\item The second procedure (\Cref{algorithm:merger-of-earlier-cliques-and-this-clique}) is run for a block vertex $B \in V(T_G)$ (with its implicit parent cut-vertex $v$ in $T_G$).
	At each call it merges $P(u)$ for a next child cut-vertex $u \in T_G$ (or a single pattern for a simplicial vertex) into a set of all patterns for $P(B \setminus \{v\})$.
	\item The third procedure (\Cref{algorithm:colorings-of-block-without-parent-to-colorings-with-parent}) is run for a cut-vertex $v \in V(T_G)$ and its $d$-th child block $B \in V(T_G)$ to obtain $P_d(v)$ from $P(B \setminus \{v\})$.
\end{itemize}
We present them all in detail below. For brevity, we will use in pseudocodes zero vector ($\mathbf{0}$), unit vectors ($\mathbf{1}^l$ for some $l = 0, \ldots, k$), and vector operations such as addition, subtraction or scalar multiplication.

Moreover, it is tacitly assumed that with each pattern we always store a sample coloring respective to this pattern, i.e. a precise assignment from vertices to colors.

\subsubsection{Merging the patterns of children of a given cut-vertex}
\label{sec:merging-cut}

\begin{algorithm}[htpb]
	\begin{algorithmic}[1]
		\REQUIRE {A set $MP(v, d - 1)$ of merged patterns $P_i(v)$ for $i = 1, 2, \ldots, d - 1$ and a set of patterns $P_d(v)$}
		\ENSURE {A set $MP(v, d)$ of merged patterns $P_i(v)$ for $i = 1, 2, \ldots, d$}
		\STATE \textbf{if} $d = 1$ \textbf{then return} $P_d(v)$
		\STATE $MP(v, d) \gets \emptyset$
		\FORALL{$(\colorv, \colorDescv) \in MP(v, d - 1)$}
		    \FORALL{$(\colorv', \colorDescv') \in P_d(v)$}
		        \STATE $\cardv, \cardv' \gets$ the only non-zero coordinates of $\colorv$ and $\colorv'$, respectively
		        \STATE \textbf{if} $\cardv + \cardv' - 1 > k$ \textbf{then continue}
		        \FORALL{$\colorTargetDescv \in \{0, \ldots, m\}^{k + 1}$}
                    \STATE \textbf{if} $\sum_i \colorTargetDescv(i) \neq m - 1$ \textbf{then continue}
                        \STATE $j \to 0$
                        \FOR{$i = 0, \ldots, k + 1$}
                            \STATE $\cardTargetDescv(j), \ldots, \cardTargetDescv(j + \colorTargetDescv(i)) \gets i$
                            \STATE $j \gets j + \colorTargetDescv(i)$
                        \ENDFOR
                        \hfill\COMMENT{$\cardTargetDescv$ is a sequence of $m - 1$ color cardinalities consistent with $\colorTargetDescv$}
		            \STATE $\UnmergedColorings(0) \gets \{\colorDescv\}$, $\UnmergedColorings'(0) \gets \{\colorDescv'\}$ \\
                        \hfill\COMMENT{initially only the unique colors from $\colorv$ and $\colorv'$ are identified}
		            \FOR{$i = 1, \ldots, m - 1$}
		                \STATE $\UnmergedColorings(i) \gets \emptyset$, $\UnmergedColorings'(i) \gets \emptyset$
		                \FORALL{$(coloring, coloring') \in \UnmergedColorings(i - 1) \times \UnmergedColorings'(i - 1)$}
		                    \FOR{$j = 1, \ldots, k-\cardTargetDescv(i)$}
		                        \IF{$coloring(j) > 0$ \textbf{and} $coloring'(\cardTargetDescv(i) - j) > 0$}
		                            \STATE Add $coloring - \mathbf{1}^{j}$ to $\UnmergedColorings(i)$
		                            \STATE Add $coloring' - \mathbf{1}^{\cardTargetDescv(i) - j}$ to $\UnmergedColorings'(i)$
		                        \ENDIF
		                    \ENDFOR
		                \ENDFOR
		            \ENDFOR
		            \IF{$\mathbf{0} \in \UnmergedColorings(m - 1)$ \textbf{and} $\mathbf{0} \in \UnmergedColorings'(m - 1)$}
		                \STATE Add $(\mathbf{1}^{\cardv + \cardv' - 1}, \colorTargetDescv)$ to $MP(v, d)$
		            \ENDIF
		        \ENDFOR
    	    \ENDFOR
	    \ENDFOR
		\RETURN $MP(v, d)$
	\end{algorithmic}
	\caption
	{
		An algorithm for merging a set of already merged patterns $P_i(v)$ for $i = 1, 2, \ldots, d - 1$ and a set of patterns $P_d(v)$.
	}
	\label{algorithm:composing-subcolorings-for-subblockgraphs}
\end{algorithm}

First, we present a lemma about combining colorings of consecutive children of a cut-vertex $v$, to form colorings of all descendants of $v$.
Keep in mind that, for any $i$-th descendants of given a cut-vertex  pattern $P_i(a, b)$  is still such that $a = \mathbf{1}^i$, simply only one color is needed for a cut-vertex.

For $v \in V_{cut}$ and first $d$-th descendants, let the set of all distinct patterns be denoted by $MP(v, d - 1)$.
\begin{lemma}
    For any $v \in V_{cut}$ and $d \in N_{+}$, given $MP(v, d - 1)$ and $P_d(v)$ \Cref{algorithm:composing-subcolorings-for-subblockgraphs} calculates $MP(v, d)$.
	\label{lem:composing-subcolorings-for-subblockgraphs}
\end{lemma}

\begin{proof}
	In the algorithm we determine if using some pattern $(a, b)$ from a set of the patterns for first $d - 1$ children $MP(v, d - 1)$ and some pattern $(a', b') \in P_d(v)$, a given pattern $(a^*, b^*)$ exists in $MP(v, d)$.
	First, it has to hold that $v$ has to be colored with a unique color given by $a$ and $a'$.
	We have to add the respective cardinalities and subtract one -- because this vertex is common in both colors, thus $a^*$ is determined uniquely.

	Other color classes can be composed by merging pairs of colors.
	For simplicity, we check exhaustively the cardinalities of target color classes and we check if a sequence of such cardinalities can be achieved.
	We may do this in a dynamic programming fashion, one color at a time.	
	For $i \in 0, \ldots, m-1$, the sets $\UnmergedColorings(i)$ and $\UnmergedColorings'(i)$ represents colors still to be merged after using some colors to construct colors of cardinalities $s_1, \ldots, s_i$.
	For each tuple representing colors to be merged, we can choose in up to $k$ ways colors that can be combined to obtain a merged color of the desired size.
	Clearly, it is sufficient for correctness (and necessary for time bounds) that each set $MP(v, d)$ preserves only unique patterns.
\end{proof}

\begin{corollary}
    By applying \Cref{algorithm:composing-subcolorings-for-subblockgraphs} in a loop to a sequence of children of $v \in V_{cut}$ in $T_G$ in any order, we can obtain a set of all patterns for a graph induced by $G[\succesorssymbol(v)]$, that is, $P(v)$.
	\label{cor:composing-subcolorings-for-subblockgraphs}
\end{corollary}

\begin{lemma}
    \Cref{algorithm:composing-subcolorings-for-subblockgraphs} runs in $\Osymbol(n k^3 m^{5k+6})$ time. 
	\label{lem:composing-subcolorings-for-subblockgraphs-complexity}
\end{lemma}

\begin{proof}
    We proceed by estimating the number of iterations in each \textbf{for} loop:
    \begin{itemize}
        \item the first loop runs for $\Osymbol(k m^{k + 1})$ iterations since $|MP(v, d)| = \Osymbol(k m^{k + 1})$ for any $v \in V_{cut}$ and $d$ by the fact that it consists patterns from the set $\{\mathbf{1}^l\colon l = 0, \ldots, k\} \times \{0, \ldots, m\}^{k + 1}$,
        \item the second loop runs for $\Osymbol(k m^{k + 1})$ iterations since $|P_d(v)| = \Osymbol(k m^{k + 1})$ by the same argument as above,
        \item the third and fourth loop run for $\Osymbol(m^{k + 1})$ and $\Osymbol(m)$ iterations, respectively,
        \item the fifth loop runs for $\Osymbol(m^{2 k + 2})$ iterations, since $|\UnmergedColorings(i)| = \Osymbol(m^{k + 1})$, $|\UnmergedColorings'(i)| = \Osymbol(m^{k + 1})$ due to the fact that they are always subsets of the set $\{0, \ldots, m\}^{k + 1}$,
        \item the sixth loop obviously runs for $\Osymbol(k)$ iterations.
    \end{itemize}
    Thus, all loops in total amount to $\Osymbol(k^3 m^{5k+6})$ iterations.

    The elementary operation consists of a construction of a new vector from the set $\{0, \ldots, m\}^{k + 1}$ (which requires $\Osymbol(k)$ time), and an implicit construction of a respective coloring (which requires $\Osymbol(m + n)$ time).
    Therefore, by the natural assumption that $k, m \le n$ we arrive at the desired complexity, which concludes the proof.
\end{proof}

\subsubsection{Merging the colorings of children for a given block}
\label{sec:merging-block}

\begin{algorithm}[htpb]
	\begin{algorithmic}[1]
	\REQUIRE A block $B \in V_B$, its parent cut-vertex $v \in V_{cut}$ in $T_G$.
	 Sets of all patterns $P(U)$ and $P(u)$ for some $U \subseteq B \setminus \{v\}$ and a vertex $u \in B \setminus (U \cup \{v\})$
	\ENSURE A set of all patterns $P(U \cup \{u\})$
	\STATE \textbf{if} $U = \emptyset$ \textbf{then return} $P(u)$\hfill\COMMENT{Only for conformity with \Cref{algorithm:main-loop}}
	\STATE $P(U \cup \{u\}) \gets \emptyset$
	\FORALL{$(\ColorsOfU, \ColorsOther) \in P(U)$}
		\FORALL{$(\ColorOfu, \ColorsOtherPrim) \in P(u)$}
	        \STATE $\cardu \gets$ the only non-zero coordinate in $\ColorOfu$
	        \FOR{$i = 0, \ldots, k - \cardu$}
	            \IF{$\ColorsOther(i) > 0$}
	                \STATE $\ColoringsObtainable(0) \gets \{(\ColorsOfU, \ColorsOther - \mathbf{1}^i, \mathbf{1}^{\cardu + i}, \mathbf{0})\}$
                        \\\hfill\COMMENT{initially only the unique color from $\ColorOfu$ is identified with some color outside of $\ColorsOfU$}
                        \STATE $j \to 0$
                        \FOR{$i = 0, \ldots, k + 1$}
                            \STATE $\cardDescu(j), \ldots, \cardDescu(j + \ColorsOtherPrim(i)) \gets i$
                            \STATE $j \gets j + \ColorsOtherPrim(i)$
                        \ENDFOR
                        \hfill\COMMENT{$\cardDescu$ is a sequence of $m - 1$ color cardinalities consistent with $\ColorsOtherPrim$}
                    \FOR{$j = 1, \ldots, m - 1$}
                        \STATE $\ColoringsObtainable(j) \gets \emptyset$
                        \FORALL{$(\ColorsUSrc, \ColorsOtherSrc, \ColorsUuTgt, \ColorsOtherTgt) \in \ColoringsObtainable(j - 1)$}
	                    	\FOR{$l = 0, \ldots, k - \cardDescu(j)$}
		                        \IF{$\ColorsUSrc(l) > 0$}
		                            \STATE Add $(\ColorsUSrc - \mathbf{1}^l, \ColorsOtherSrc, \ColorsUuTgt + \mathbf{1}^{l + \cardDescu(j)}, \ColorsOtherTgt)$ to $\ColoringsObtainable(j)$
		                            \\\hfill\COMMENT {Merge $j$-th color from $\ColorsOtherPrim$ with some color used for $U$}
	                    	    \ENDIF
		                        \IF{$\ColorsOtherSrc(l) > 0$}
		                            \STATE Add $(\ColorsUSrc, \ColorsOtherSrc - \mathbf{1}^l, \ColorsUuTgt, \ColorsOtherTgt + \mathbf{1}^{l + \cardDescu(j)})$ to $\ColoringsObtainable(j)$
		                            \\\hfill\COMMENT {Merge $j$-th color from $\ColorsOtherPrim$ with some color not used for $U$}
	                    	    \ENDIF
	                    	\ENDFOR
                        \ENDFOR
                    \ENDFOR
                    \FORALL{$(\mathbf{0}, \mathbf{0}, \ColorsUuTgt, \ColorsOtherTgt) \in \ColoringsObtainable(m - 1)$}
                    		\STATE Add $(\ColorsUuTgt, \ColorsOtherTgt)$ to $P(U \cup \{u\})$
                    \ENDFOR
	            \ENDIF
	        \ENDFOR
        \ENDFOR
    \ENDFOR
	\RETURN $P(U \cup \{u\})$
	\end{algorithmic}
	\caption
	{
		For a given block $B$ together with its parent cut-vertex $v$, the procedure merges sets of patterns $U \subseteq B \setminus \{v\}$, and sets of patterns of another vertex from $B \setminus (U \cup \{v\})$.
	}
	\label{algorithm:merger-of-earlier-cliques-and-this-clique}
\end{algorithm}

Now we present a lemma about combining sets of patterns for cut-vertices and simplicial vertices for a given block $\blocksymbol$.
\begin{lemma}
    For any $B \in V_B$, its parent cut-vertex $v \in V_{cut}$ in $T_G$, and any $U \subseteq B \setminus \{v\}$ and $u \in B \setminus (U \cup \{v\})$ with their given sets of patters $P(U)$ and $P(u)$ \Cref{algorithm:merger-of-earlier-cliques-and-this-clique} computes $P(U \cup \{u\})$.
    \label{lem:merger-of-earlier-cliques-and-this-clique}
\end{lemma}

\begin{proof}
    First, we have to merge the color assigned to $u$ with some other color.
	Of course, it cannot be used to any other vertex in $\blocksymbol$, hence we can merge only with the colors that are not assigned to vertices from $U$.
	The third {\bfseries for} loop proceeds over all such possibilities.
	
	Second, we have to merge other colors assigned to $\succesorssymbol(u)$ with the ones assigned to $\succesorssymbol(U)$.
	Observe that it can be done in any way -- either with colors assigned to the vertices in $\blocksymbol$ or with other colors.
	To find all distinct ways of merging the colors we use a dynamic programming approach over a sequence of cardinalities of colors used to color $\succesorssymbol(u)$.
	The next loop does so: at each step of merging a single color we preserve the sets of colors that are yet to be merged and that were already merged.
	The state of the dynamic program is denoted by $(\ColorsUSrc, \ColorsOtherSrc, \ColorsUuTgt, \ColorsOtherTgt)$.
	The first two elements of the tuple represent the colors (i.e. their cardinalities for $\succesorssymbol(U)$) that are not yet merged.
	The last two elements of the tuple represent the colors that are already merged.
	
	For any given color used in $\succesorssymbol(u)$ we can merge it with $O(k)$ colors with different cardinalities used to color $\succesorssymbol(U)$. And for each fixed cardinality we have two possibilities: either we can merge with it a color assigned to some vertex of $U$ or not -- the {\bfseries ifs} are taking care of both cases.
	
	Finally, if in $\ColoringsObtainable(m - 1)$ there is some tuple, then it means that the corresponding pattern can be obtained.
	Observe that for $i \in \{0, \ldots, m-1\}$ the set $\ColoringsObtainable(i)$ contain all obtainable patterns after merging $i$ colors from $P(u)$.
	Finally, the returned set represents all patterns for all possible colorings of $G[\succesorssymbol(U \cup \{u\})]$, thus it is exactly $P(U \cup \{u\})$.
\end{proof}

\begin{corollary}
    For any block $B \in V_B$ with its parent cut-vertex $v \in V_{cut}$ in $T_G$, a set of all patterns $P(B \setminus \{v\})$ can be obtained by calling \Cref{algorithm:merger-of-earlier-cliques-and-this-clique} for the consecutive vertices of $B \setminus \{v\}$.
    \label{cor:merger-of-earlier-cliques-and-this-clique}
\end{corollary}

\begin{lemma}
    \Cref{algorithm:merger-of-earlier-cliques-and-this-clique} runs in $\Osymbol(n k^3 m^{7k + 8})$ time. 
	\label{lem:merger-of-earlier-cliques-and-this-clique-complexity}
\end{lemma}

\begin{proof}
    Again we proceed by bounding the number of iterations in each \textbf{for} loop:
    \begin{itemize}
        \item the first loop runs for $\Osymbol(m^{2 k + 2})$ iterations since $|P(U)| = \Osymbol(m^{2 k + 2})$ for any $U$ by the fact that $P(U) \subseteq \{0, \ldots, m\}^{k + 1} \times \{0, \ldots, m\}^{k + 1}$,
        \item the second loop runs for $\Osymbol(k m^{k + 1})$ iterations since $|P(u)| = \Osymbol(k m^{k + 1})$ since $P(u) \subseteq \{\mathbf{1}^l\colon l = 0, \ldots, k\} \times \{0, \ldots, m\}^{k + 1}$,
        \item the third and fourth loop run for $\Osymbol(k)$ and $\Osymbol(m)$ iterations, respectively,
        \item the fifth loop runs for $\Osymbol(m^{4 k + 4})$ iterations, since every element is a quadruple of vectors from $\{0, \ldots, m\}^{k + 1}$,
        \item the sixth loop obviously runs for $\Osymbol(k)$ iterations.
    \end{itemize}
    Note that the running time of the last loop is dominated by the running times of the fourth, fifth, and sixth loops, so we do not need to consider it separately. 
    
    The total number of iterations is bounded by $\Osymbol(k^3 m^{7k + 8})$.
    Taking into account the time necessary for an elementary operation (which, as for the previous algorithm, can be bounded by $\Osymbol(n)$) results in the total complexity $\Osymbol(n k^3 m^{7k + 8})$.
\end{proof}

\subsubsection{Converting the patterns of a given block to the patterns of all descendants of given parent cut-vertex}
\label{sec:block-to-cut}

\begin{algorithm}[htpb]
	\begin{algorithmic}[1]
		\REQUIRE A cut-vertex $v \in V(T_G)$, $d \in \N^+$, and $P(B \setminus \{v\})$ for a $d$-th child block $B$ of $v$ in $T_G$.
		
		\STATE $P_d(v) \gets \emptyset$
		\FORALL{$(\ColorsInBlock, \ColorsOtherInBlock) \in P(B \setminus \{v\})$}
		    \FOR{$l = 0, \ldots, k - 1$}
			    \IF {$\ColorsOtherInBlock(l) > 0$}
				    \STATE Add $(\mathbf{1}^{l+1}, \ColorsInBlock + \ColorsOtherInBlock - \mathbf{1}^l)$ to $P_d(v)$
			    \ENDIF
		    \ENDFOR
		\ENDFOR
    \RETURN $P_d(v)$
	\end{algorithmic}
	\caption{Compute $P_d(v)$ from $P(B \setminus \{v\})$ for the respective $d$-th child block $B$ of $v$ in $T_G$.}
	\label{algorithm:colorings-of-block-without-parent-to-colorings-with-parent}
\end{algorithm}

\begin{lemma}
    Given $v \in V(T_G)$, $d \in \N^+$, and a set of patterns $P(B \setminus \{v\})$  for a $d$-th child block $B$ of $v$ in $T_G$ \Cref{algorithm:colorings-of-block-without-parent-to-colorings-with-parent} computes the set $P_d(v)$.
    \label{lem:coloring-cardinalities-for-a-block-and-fixed-vertex}
\end{lemma}

\begin{proof}
    The algorithm chooses one of the colors that are not assigned to $B \setminus \{v\}$ and assigns it to $v$. 
    Other colors, both assigned to $B \setminus \{v\}$ and not, can be lumped together (without merging) into one set for colors used only in descendants of $v$.
    We can do so, because by construction we ensured that they cannot interfere with the colorings of other blocks. 
\end{proof}

\begin{lemma}
    \Cref{algorithm:colorings-of-block-without-parent-to-colorings-with-parent} runs in time $\Osymbol(n k m^{2 k + 2})$.
\end{lemma}

\begin{proof}
    The proof follows directly from the fact that $P(B \setminus \{v\})$ contains patterns (i.e. pairs of vectors from $\{0, \ldots, m\}^{k + 1}$) and therefore $|P(B \setminus \{v\})| = \Osymbol(m^{2 k + 2})$.
\end{proof}

\subsubsection{Computing the total running time}

Finally, we outline a recursive algorithm for constructing all feasible patterns $P(v)$ for a given cut-vertex that gathers all the subroutines listed above.

\begin{algorithm}[htpb]
	\begin{algorithmic}[1]
	\REQUIRE {$v \in V_{cut}$}
	\ENSURE {The set of patterns representing all distinct colorings $P(v)$}
	\STATE Let $ColoringsAfterVisitingBlocks(v) \gets (\mathbf{1}^{1}, (m - 1) \cdot \mathbf{1}^{0})$ \hfill\COMMENT{Initially only $v$ is colored}
	\STATE $d_{max} \gets $ number of children of $v$ in $T_G$
	\FOR{$d = 1, 2, \ldots, d_{max}$}
	    \STATE $B \gets$ $d$-th child block of $v$ in $T_G$
	    \STATE $U \gets \emptyset$
	    \STATE $ColoringsOfSubblock(U) \gets (\mathbf{0}, m \cdot \mathbf{1}^0)$ \hfill\COMMENT{$P(\emptyset)$ slightly abuses notation, but allows for a cleaner loop}
	    \FORALL{$v' \in B \setminus \{v\}$}
	        \IF{$v' \in V_{cut}$}
	            \STATE $ColoringsRootedIn(v') \gets$ \Cref{algorithm:main-loop} $(v')$
	        \ELSE
	            \STATE $ColoringsRootedIn(v') \gets (\mathbf{1}^1, (m - 1) \cdot \mathbf{1}^0)$ \hfill\COMMENT{$v'$ is a simplicial vertex}
	        \ENDIF
	        \STATE $ColoringsOfSubblock(U \cup \{v'\}) \gets$ \Cref{algorithm:merger-of-earlier-cliques-and-this-clique} $(ColoringsOfSubblock(U), ColoringsRootedIn(v'))$
	        \STATE $U \gets U \cup \{v'\}$
	    \ENDFOR
	    \STATE $ColoringsOfChild(d) \gets$ \Cref{algorithm:colorings-of-block-without-parent-to-colorings-with-parent} $(v, d, ColoringsOfSubBlock(B \setminus \{v\}))$
	    \STATE $ColoringsAfterVisitingBlocks(v) \gets$ \Cref{algorithm:composing-subcolorings-for-subblockgraphs} $(ColoringsAfterVisitingBlocks(v), ColoringsOfChild(d))$
    \ENDFOR
	\RETURN $ColoringsAfterVisitingBlocks(v)$
	\end{algorithmic}
	\caption{The main recursive function}
	\label{algorithm:main-loop}
\end{algorithm}

\begin{theorem}
	\Cref{algorithm:main-loop} returns a set of all patterns $P(v)$ 
	for any $v \in V_{cut}$.
	\label{lem:MergeColoringsOfVertices}
\end{theorem}
\begin{proof}
By \Cref{lem:merger-of-earlier-cliques-and-this-clique} and \Cref{cor:merger-of-earlier-cliques-and-this-clique}, the innermost \textbf{for} loop computes for every child block $B_i \in V_B$ of $v$ in $T_G$ a set $P(B_i \setminus \{v\})$. 
By \Cref{lem:coloring-cardinalities-for-a-block-and-fixed-vertex}, in line 13 of \Cref{algorithm:main-loop} we build from them the set $P_d(v)$.
By \Cref{lem:composing-subcolorings-for-subblockgraphs} we merge it (together with patterns for the previous block) into $MP(v, d)$ for $d = 1, 2, \ldots$ in the consecutive iterations.
Observe, that \Cref{algorithm:colorings-of-block-without-parent-to-colorings-with-parent} colors $v$.
Finally, \Cref{cor:composing-subcolorings-for-subblockgraphs} proves that $MP(v, d_{max}) = P(v)$, which finishes the proof.
This is by the assumption that initially $v$ is colored with one color with cardinality $1$ and by an observation that \Cref{algorithm:composing-subcolorings-for-subblockgraphs} takes into account that $v$ is already colored.
\end{proof}

\begin{theorem}
    Applying \Cref{algorithm:main-loop} for a root $r$ of $T_G$ solves the problem $P|G = \blockgraph, p_j = 1|C_{\max} \le k$ in $\Osymbol(n^2 k^3 m^{7k+8})$ time.
    \label{theorem:final-k-coloring-block-graphs}
\end{theorem}

\begin{proof}
From \Cref{lem:MergeColoringsOfVertices} it follows directly that for a root $r$ of $T_G$ \Cref{algorithm:main-loop} finds $P(r)$, the set of all patterns for a block graph $G$ such that where cardinalities of color classes are bounded by a constant $k$.
Thus, any element in the set $P(r)$ implies a positive answer to the problem $P|G = \blockgraph, p_j = 1|C_{\max} \le k$ as it can be translated to an appropriate schedule. And when the set is empty it follows that no such schedule exists.

Let us calculate the total running time.
Observe that \Cref{algorithm:composing-subcolorings-for-subblockgraphs,algorithm:merger-of-earlier-cliques-and-this-clique,algorithm:colorings-of-block-without-parent-to-colorings-with-parent,algorithm:main-loop} will be called $\Osymbol(n)$ times in total.
\begin{itemize}
    \item \Cref{algorithm:composing-subcolorings-for-subblockgraphs} is called exactly once for each block and its $d$-th child (so at most once per each edge of $T_G$),
    \item \Cref{algorithm:merger-of-earlier-cliques-and-this-clique} is called at most once for each $u \in V(G)$,
    \item \Cref{algorithm:colorings-of-block-without-parent-to-colorings-with-parent} is called once for each block,
    \item \Cref{algorithm:main-loop} is called once for each vertex from $V_{cut}$
\end{itemize}
Note that each call of \Cref{algorithm:main-loop} runs in $\Osymbol(n)$ time in total if we discard calls to all other functions (including recursion).

Thus, the total time complexity adds to $\Osymbol(n^2 k^3 m^{5k+6} + n^2 k^3 m^{7k+8} + n^2 k m^{2k+2} + n^2) = \Osymbol(n^2 k^3 m^{7k+8})$.
\end{proof}

Note that if $G$ is not connected, instead of constructing an additional algorithm for merging patterns for connected components, we can simply add a dummy cut-vertex $r$ which connects the graph (and which is a root of $T_G$), and additional color.
This way, we have to check in the final $P(r)$, whether there exists a coloring of the block graph rooted in $r$ with pattern $(a, b)$, where $a = \mathbf{1}^{1}$ which means that the dummy vertex has a unique color.


\subsubsection{An example}

To give an intuitive overlook of the algorithm, let us consider the following example: a set of $m = 3$ identical machines, and block graph as depicted in \Cref{fig:blockgraph}, and limit on $\cmaxcost = 3$.
Note that it implies that all vectors used in the patterns in the algorithms given above are of length $4$.
Consider \Cref{algorithm:main-loop} applied to the graph presented in \Cref{fig:blockgraph}.
In the consecutive pictures the dashed line illustrates how the knowledge about the possible colorings of subblock graphs is constructed.
Of course, the algorithm is recursive, hence the procedures are called in a top-bottom manner.
However, also of course, the knowledge is built in bottom-top manner.

\newcommand{\nshift}{(0.0cm,0.5cm)}
\newcommand{\wshift}{(-0.5cm, -0.25cm)}
\newcommand{\eshift}{(0.5cm, -0.25cm)}
\newcommand{\sshift}{(0cm, -0.5cm)}

\newcommand{\basic}
{
\coordinate (J1pos) at (0,0);
\coordinate (J1posn) at ([shift=\nshift]J1pos);
\coordinate (J1posw) at ([shift=\wshift]J1pos);
\coordinate (J1poss) at ([shift=\sshift]J1pos);
\coordinate (J1pose) at ([shift=\eshift]J1pos);

\coordinate (J2pos) at (0,-1);
\coordinate (J2posn) at ([shift=\nshift]J2pos);
\coordinate (J2posw) at ([shift=\wshift]J2pos);
\coordinate (J2poss) at ([shift=\sshift]J2pos);
\coordinate (J2pose) at ([shift=\eshift]J2pos);

\coordinate (J3pos) at (-1,-2);
\coordinate (J3posn) at ([shift=\nshift]J3pos);
\coordinate (J3posw) at ([shift=\wshift]J3pos);
\coordinate (J3poss) at ([shift=\sshift]J3pos);
\coordinate (J3pose) at ([shift=\eshift]J3pos);

\coordinate (J4pos) at (1,-2);
\coordinate (J4posn) at ([shift=\nshift]J4pos);
\coordinate (J4posw) at ([shift=\wshift]J4pos);
\coordinate (J4poss) at ([shift=\sshift]J4pos);
\coordinate (J4pose) at ([shift=\eshift]J4pos);

\coordinate (J5pos) at (-2,-3);
\coordinate (J5posn) at ([shift=\nshift]J5pos);
\coordinate (J5posw) at ([shift=\wshift]J5pos);
\coordinate (J5poss) at ([shift=\sshift]J5pos);
\coordinate (J5pose) at ([shift=\eshift]J5pos);

\coordinate (J6pos) at (-1,-3);
\coordinate (J6posn) at ([shift=\nshift]J6pos);
\coordinate (J6posw) at ([shift=\wshift]J6pos);
\coordinate (J6poss) at ([shift=\sshift]J6pos);
\coordinate (J6pose) at ([shift=\eshift]J6pos);

\coordinate (J7pos) at (-0,-3);
\coordinate (J7posn) at ([shift=\nshift]J7pos);
\coordinate (J7posw) at ([shift=\wshift]J7pos);
\coordinate (J7poss) at ([shift=\sshift]J7pos);
\coordinate (J7pose) at ([shift=\eshift]J7pos);

\coordinate (J8pos) at (1,-3);
\coordinate (J8posn) at ([shift=\nshift]J8pos);
\coordinate (J8posw) at ([shift=\wshift]J8pos);
\coordinate (J8poss) at ([shift=\sshift]J8pos);
\coordinate (J8pose) at ([shift=\eshift]J8pos);

\coordinate (J9pos) at (2,-3);
\coordinate (J9posn) at ([shift=\nshift]J9pos);
\coordinate (J9posw) at ([shift=\wshift]J9pos);
\coordinate (J9poss) at ([shift=\sshift]J9pos);
\coordinate (J9pose) at ([shift=\eshift]J9pos);

\node[main] at (J1pos) (J1) {$J_1$};
\node[main] at (J2pos) (J2) {$J_2$};
\node[main] at (J3pos) (J3) {$J_3$};
\node[main] at (J4pos) (J4) {$J_4$};

\node[main] at (J5pos) (J5) {$J_5$};
\node[main] at (J6pos) (J6) {$J_6$};
\node[main] at (J7pos) (J7) {$J_7$};
\node[main] at (J8pos) (J8) {$J_8$};
\node[main] at (J9pos) (J9) {$J_9$};

\draw (J1) -- (J2) -- (J3) -- (J4) -- (J2);
\draw (J3) -- (J5) -- (J6) -- (J3);
\draw (J3) -- (J7);
\draw (J4) -- (J8);
\draw (J9) -- (J4);
}

\begin{figure}[htpb]
		\centering
	\begin{subfigure}[t]{.3\textwidth}
	\begin{tikzpicture}[align=center, main/.style = {draw, circle},  inner sep=1pt]
	\basic
	\end{tikzpicture}
	\subcaption{An example block graph $G$ with $T_G$ rooted in $J_1$}
	\label{fig:blockgraph}
	\end{subfigure}%
	\hfill
	\begin{subfigure}[t]{.3\textwidth}
	\centering
	\begin{tikzpicture}[align=center, main/.style = {draw, circle},  inner sep=1pt]
	\basic
	\draw[dashed] plot [smooth, tension=2] coordinates{(J5posw) (J5posn) (J5pose)};
	\draw[dotted] plot [smooth, tension=2] coordinates{(J6posw) (J6posn) (J6pose)};
	\end{tikzpicture}
	\subcaption
	{
		Sets $P(J_5)$ (dashed) and $P(J_6)$ (dotted).
 	}
	\label{fig:StartingColorings}
	\end{subfigure}%
	\hfill
	\begin{subfigure}[t]{.3\textwidth}
	\begin{tikzpicture}[align=center, main/.style = {draw, circle},  inner sep=1pt]
	\basic
	\draw[dashed] plot [smooth, tension=2] coordinates{(J5posw) (J5posn) (J6posn) (J6pose)};
	\end{tikzpicture}
	\subcaption
	{
		The set $P(\{J_5, J_6\})$ (dashed) obtained by \Cref{algorithm:merger-of-earlier-cliques-and-this-clique} from $P(J_5)$ and $P(J_6)$.
	}
	\label{fig:ColoringsCombined}
	\end{subfigure}%

	\begin{subfigure}[t]{.3\textwidth}
	\begin{tikzpicture}[align=center, main/.style = {draw, circle},  inner sep=1pt]
	\basic
	\draw[dashed] plot [smooth, tension=1] coordinates{(J5posw) (J5posn) (J3posw) (J3posn) (J3pose) (J6pose)};
	\end{tikzpicture}
	\subcaption{
		The set $P_1(J_3)$ (dashed) obtained by using \Cref{algorithm:colorings-of-block-without-parent-to-colorings-with-parent} from $P(\{J_5, J_6\})$.
 	}
	\label{fig:SecondBranch}
	\end{subfigure}%
	\hfill
	\begin{subfigure}[t]{.3\textwidth}
	\begin{tikzpicture}[align=center, main/.style = {draw, circle},  inner sep=1pt]
		\basic
		\draw plot [smooth, tension=1] coordinates{(J5posw) (J5posn) (J3posw) (J3posn) (J3pose) (J6pose)};
		\draw[dashed] plot [smooth, tension=2] coordinates{(J7posw) (J7posn) (J7pose)};
	\end{tikzpicture}
	\subcaption{
		The set $P(J_7)$ (dashed) is constructed.
	}
	\label{fig:SeconBrnach}
	\end{subfigure}%
	\hfill
	\begin{subfigure}[t]{.3\textwidth}
	\begin{tikzpicture}[align=center, main/.style = {draw, circle},  inner sep=1pt]
		\basic
		\draw plot [smooth, tension=1] coordinates{(J5posw) (J5posn) (J3posw) (J3posn) (J3pose) (J6pose)};
		\draw[dashed] plot [smooth, tension=1] coordinates{(J7posw) (J3poss)  (J3posw) (J3posn) (J3pose) (J7posn) (J7pose)};
	\end{tikzpicture}
	\subcaption{
		And transformed into $P_2(J_3)$ (dashed) by \Cref{algorithm:merger-of-earlier-cliques-and-this-clique}.
	}
	\label{fig:SecondBranchWithParent}
	\end{subfigure}%

	\begin{subfigure}[t]{.3\textwidth}
	\begin{tikzpicture}[align=center, main/.style = {draw, circle},  inner sep=1pt]
		\basic
		\draw[dashed] plot [smooth, tension=1] coordinates{(J5posw) (J5posn) (J3posw) (J3posn) (J3pose) (J7posn) (J7pose)};
	\end{tikzpicture}
	\subcaption{
	 	\Cref{algorithm:composing-subcolorings-for-subblockgraphs} merges $P_1(J_3)$ with $P_2(J_3)$ into $P(J_3)$ (dashed).
	}
	\label{fig:LeftBranchesMerged}
	\end{subfigure}%
	\hfill
	\begin{subfigure}[t]{.3\textwidth}
	\begin{tikzpicture}[align=center, main/.style = {draw, circle},  inner sep=1pt]
		\basic
		\draw plot [smooth, tension=1] coordinates{(J5posw) (J5posn) (J3posw) (J3posn) (J3pose) (J7posn) (J7pose)};
		\draw[dashed] plot [smooth, tension=1] coordinates{(J8posw) (J4posw) (J4posn) (J4pose) (J8pose)};
		\draw[dotted] plot [smooth, tension=1] coordinates{(J9posw) (J8posn) (J4posw) (J4posn) (J4pose) (J9posn) (J9pose)};
	\end{tikzpicture}
	\subcaption{
		In similar fashion (by first calculating $P(J_8)$ and calling \Cref{algorithm:colorings-of-block-without-parent-to-colorings-with-parent}) we obtain $P_1(J_4)$ (dashed) and $P_2(J_4)$ (dotted).
	}
	\label{fig:ThirdAndFourthBranch}
	\end{subfigure}%
	\hfill
	\begin{subfigure}{.333\textwidth}
	\begin{tikzpicture}[align=center, main/.style = {draw, circle},  inner sep=1pt]
		\basic
		\draw plot [smooth, tension=1] coordinates{(J5posw) (J5posn) (J3posw) (J3posn) (J3pose) (J7posn) (J7pose)};
		\draw[dashed] plot [smooth, tension=1] coordinates{(J8posw) (J4posw) (J4posn) (J4pose) (J9posn) (J9pose)};
	\end{tikzpicture}
	\subcaption{		  
		We call \Cref{algorithm:composing-subcolorings-for-subblockgraphs} to obtain $P(J_4)$ (dashed).
	}
	\label{fig:RightBranchesMerged}
	\end{subfigure}%
	\label{fig:exampleForPTAS}
	\begin{subfigure}[t]{.3\textwidth}
	\begin{tikzpicture}[align=center, main/.style = {draw, circle},  inner sep=1pt]
		\basic
		\draw[dashed] plot [smooth, tension=1] coordinates{(J5posw) (J5posn) (J3posw) (J3posn) (J4posn) (J4pose) (J9posn) (J9pose)};
	\end{tikzpicture}
	\subcaption{
		On the obtained colorings, i.e. $P(J_3)$ and $P(J_4)$, we call \Cref{algorithm:merger-of-earlier-cliques-and-this-clique} to obtain $P(\{J_3, J_4\})$ (dashed).
	}
	\label{fig:AllBranchesMerged}
	\end{subfigure}%
	\hfill
	\begin{subfigure}[t]{.3\textwidth}
	\begin{tikzpicture}[align=center, main/.style = {draw, circle},  inner sep=1pt]
		\basic
		\draw[dashed] plot [smooth, tension=1] coordinates{(J5posw) (J5posn) (J3posw) (J3posn) (J2posw) (J2posn) (J2pose) (J4posn) (J4pose) (J9posn) (J9pose)};
	\end{tikzpicture}
	\subcaption{
		And call \cref{algorithm:colorings-of-block-without-parent-to-colorings-with-parent} to obtain $P_1(J_2)$ (dashed), which is also equal to $P(J_2)$ by a dummy call of \cref{algorithm:composing-subcolorings-for-subblockgraphs}).
	}
	\label{fig:J2Subtree}
	\end{subfigure}%
	\hfill
	\begin{subfigure}[t]{.3\textwidth}
	\begin{tikzpicture}[align=center, main/.style = {draw, circle},  inner sep=1pt]
		\basic
		\draw[dashed] plot [smooth, tension=0.5] coordinates{(J5posw) (J5posn) (J3posw) (J3posn) (J1posw) (J1posn) (J1pose) (J4posn) (J4pose) (J9posn) (J9pose)};
	\end{tikzpicture}
	\subcaption{ 
		Finally, we call \cref{algorithm:colorings-of-block-without-parent-to-colorings-with-parent} to obtain $P_1(J_1)$ (dashed), which is equal to $P(J_1)$ by a dummy call of \cref{algorithm:composing-subcolorings-for-subblockgraphs}.
	}
	\label{fig:J1Subtree}
	\end{subfigure}%
	\caption{
		The figures present the order in which the knowledge about the colorings of subgraphs is constructed.
		In consecutive figures the lines represent subgraphs corresponding to the \emph{sets} of distinct colorings that are consecutively constructed.
		Observe also that the calls of the algorithms give the intuitive meaning of the presented algorithms.
	}
\end{figure}

\begin{figure}[htpb]
	\centering
	\begin{subfigure}[t]{.3\textwidth}
	\begin{tikzpicture}[align=center, main/.style = {draw, circle},  inner sep=1pt]
	\basic
	\node[] at (J1posn) {$M_1$};
	\node[] at (J2poss) {$M_3$};
	\node[] at (J3posn) {$M_1$};
	\node[] at (J4posn) {$M_2$};
	\node[] at (J5poss) {$M_2$};
	\node[] at (J6poss) {$M_3$};
	\node[] at (J7poss) {$M_2$};	
	\node[] at (J8poss) {$M_3$};
	\node[] at (J9poss) {$M_1$};	
	\end{tikzpicture}
	\end{subfigure}
	\hfill
	\begin{subfigure}[t]{.3\textwidth}
	\begin{tikzpicture}[align=center, main/.style = {draw, circle},  inner sep=1pt]
	\basic
	\node[] at (J1posn) {$M_1$};
	\node[] at (J2poss) {$M_3$};
	\node[] at (J3posn) {$M_2$};
	\node[] at (J4posn) {$M_1$};
	\node[] at (J5poss) {$M_3$};
	\node[] at (J6poss) {$M_1$};
	\node[] at (J7poss) {$M_3$};	
	\node[] at (J8poss) {$M_2$};
	\node[] at (J9poss) {$M_2$};	
	\end{tikzpicture}
	\end{subfigure}
	\hfill
	\begin{subfigure}[t]{.3\textwidth}
	\begin{tikzpicture}[align=center, main/.style = {draw, circle},  inner sep=1pt]
	\basic
	\node[] at (J1posn) {$M_1$};
	\node[] at (J2poss) {$M_3$};
	\node[] at (J3posn) {$M_2$};
	\node[] at (J4posn) {$M_1$};
	\node[] at (J5poss) {$M_1$};
	\node[] at (J6poss) {$M_3$};
	\node[] at (J7poss) {$M_3$};	
	\node[] at (J8poss) {$M_2$};
	\node[] at (J9poss) {$M_2$};	
	\end{tikzpicture}
	\end{subfigure}
	\caption{
		The sample block graph with sample colorings represented by $(0, 0, 0, 1; 0, 0, 0, 2)$.
		The meaning is that the color assigned to $J_1$ is of cardinality $3$ and there are $2$ other colors of cardinality $3$ each.
	}
	\label{fig:blockgraph_colord}
\end{figure}
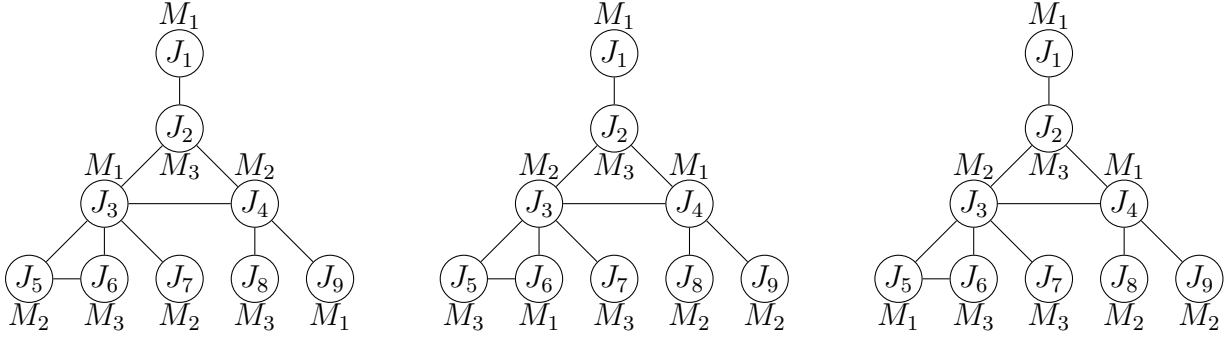

\begin{itemize}
	\item Without loss of generality we can assume that the first non-recursive application of the algorithms will be for block $B_1 = \{J_3, J_5, J_6\}$, and from this block we build the knowledge about the possible colorings of the graph.
	The \emph{sets} of the colorings are represented in \cref{fig:StartingColorings}.
	Due to the fact that both $J_5$ and $J_6$ are simplicial each of them will be represented by a tuple $\{(0,1,0,0;0,0,0,0)\}$ (obtained by lines 6--12 of \Cref{algorithm:main-loop}).
	\item Then they are merged into $P(\{J_5, J_6\}) = \{(0, 2, 0, 0; 1, 0, 0, 0)\}$\footnote{Here we write a pair of vectors $(a, b)$ as a single vector with semicolon separating both vectors.} by \Cref{algorithm:merger-of-earlier-cliques-and-this-clique} as there is only one possible coloring of $J_5$ and $J_6$. 
	The first part $(0, 2, 0, 0)$ denotes that all vertices in the argument (i.e. here $J_5$ and $J_6$) use exactly two different colors of cardinality $1$, and the second part $(1, 0, 0, 0)$ denotes that all other vertices in the partial coloring (here there are none) use one color of cardinality $0$.
	\item Using \Cref{algorithm:colorings-of-block-without-parent-to-colorings-with-parent} we transform $P(\{J_5, J_6\})$ into $P_1(J_3) = \{(0, 1, 0, 0; 0, 2, 0, 0)\}$: the only possible coloring of $J_3$, $B_1$, and its descendants in $T_G$ uses one color of cardinality $1$ for the argument (i.e. $J_3$), and two colors of cardinality $1$ for the remaining colored part.
	\item Next, for $B_2 = \{J_3, J_7\}$ the inner loop (lines 6--12) returns a set $P(\{J_7\}) = \{(0, 1, 0, 0; 2, 0, 0, 0)\}$.
	\item \Cref{algorithm:merger-of-earlier-cliques-and-this-clique} transforms it into $P_2(J_3) = \{(0, 1, 0, 0; 1, 1, 0, 0)\}$. 
	\item Finally, by \Cref{algorithm:composing-subcolorings-for-subblockgraphs}, we merge the obtained $P_1(J_3)$ with $P_2(J_3)$ and return $P(J_3) = \{(0, 1, 0, 0; 0, 1, 1, 0)\}$:
	the first vector has to be $(0, 1, 0, 0)$ because both first vertices of the only element for the previous $P(J_3)$ and $P_2(J_3)$ imply that $J_3$ is the only vertex colored with its color;
	the second vector has to be $(0, 1, 1, 0)$ because if we merge the only coloring of $\{J_5, J_6\}$ using two colors with cardinality $1$ with the only coloring of $\{J_7\}$ using one color with cardinality $1$ we have to get a coloring with one color with cardinality $2$ and one with cardinality $1$, because the third color is already used by $J_3$.
	\item The colorings for subgraph rooted in $J_4$ are obtained next: first, by obtaining coloring of $J_8$, as a simplical vertex, and then transforming it by \cref{algorithm:colorings-of-block-without-parent-to-colorings-with-parent} into a coloring of $\{J_4, J_8\}$.
	Similarly for first $J_9$ and then $\{J_4, J_9\}$.
	\item Finally, by \cref{algorithm:composing-subcolorings-for-subblockgraphs} we obtain the colorings of subgraph rooted in $J_4$.
	\item Of course (due to the recursive nature), both the sets of the colorings (the ones for subgraph rooted in $J_3$ and the ones rooted in $J_4$) are obtained due to the processing of block $\{J_2, J_3, J_4\}$.
	I.e., having a single block $B_2 = \{J_2, J_3, J_4\}$ we iteratively obtain in line 11 first $P(\{J_3\}) = \{(0, 1, 0, 0; \allowbreak 0, 1, 1, 0)\}$ and then $P(\{J_3, J_4\}) = \{(0, 0, 1, 1; 0, 0, 1, 0), (0, 0, 2, 0; 0, 0, 0, 1), (0, 0, 0, 2; 0, 1, 0, 0), (0, 1, 0, 1; 0, 0, 0, 1)\}$.
	In the last set, the first two elements come from combining $(0, 1, 0, 0; 0, 1, 1, 0)$ with $(0, 1, 0, 0; 0, 2, 0, 0)$, as we have to:
	identify color of $J_3$ (of cardinality $1$) with one of the two colors not used by $J_4$ in its partial coloring (also of cardinality $1$),
	identify color of $J_4$ (of cardinality $1$) with one of the two colors not used by $J_3$ in its partial coloring (of cardinality $2$ or $1$, hence two elements), and
	identify the remaining third color from both partial colorings.
	Similarly, we get the latter two elements, as well as the duplicate of the first element, from combining $(0, 1, 0, 0; 0, 1, 1, 0)$ with $(0, 1, 0, 0; 1, 0, 1, 0)$. 
	Note that in this case we cannot identify two colors of cardinality $2$ in both partial colorings, since it would violate the $\cmaxcost = 3$ condition.
	The combination of the colorings is done by \Cref{algorithm:merger-of-earlier-cliques-and-this-clique}.
	\item Now, \Cref{algorithm:colorings-of-block-without-parent-to-colorings-with-parent} transforms $P(\{J_3, J_4\})$ into $P_1(J_2) = \{(0, 0, 0, 1; 0, 0, 1, 1), (0, 0, 1, 0; 0, 0, 0, 2)\}$. 
	At this stage, we discarded several solutions which did not allow for a legal coloring of $J_2$ with a color not used by either $J_3$ or $J_4$. Since there is only one block, we get $P(J_2) = P_1(J_2)$ and return it.
	\item Of course it was done as a processing of the block $\{J_1, J_2\}$ when $J_2$ was visited.
	Hence finally by \Cref{algorithm:colorings-of-block-without-parent-to-colorings-with-parent} for $v = J_1$ we get $P(J_1) = P_1(J_1) = \{(0, 0, 0, 1; 0, 0, 0, 2)\}$ because the second vector from $P(J_2)$ cannot correspond to any feasible coloring of $J_1$. 
	Anyway $P(J_1)$ contains all feasible colorings of $G$, and since it is nonempty then we can conclude that there exists a coloring. In particular, as presented in \Cref{fig:blockgraph_colord}, it is a coloring where the cardinalities of all colors are equal to $3$.
	
\end{itemize}

\subsection{Unit time jobs} \label{unit:ptas}

In this section we formulate a PTAS for $P|G = \blockgraph, p_j = 1|\cmaxcost$ problem.
Note that for an instance of $Pm|G = \blockgraph, p_j=1|\cmaxcost$ either (1) there is no schedule if $\omega(G) > m$ or (2) the optimal schedule can be found in polynomial time because one can easily find a tree decomposition with width $\omega(G) \le m$. 
In the second case we can immediately use the FPTAS discussed in \cite{bodlaender1994scheduling}, 

After these preliminaries let us again consider \Cref{alg:2apx_identical_block}, but this time with respect to unit time jobs, i.e. when applied to $P|G = \blockgraph, p_j = 1|\cmaxcost$ problem.

\begin{lemma}
\label{lem:equal_coloring_block}
	Let an instance of $P|G = \blockgraph, p_j = 1|\cmaxcost$ with $n \ge 1$ jobs, where $n = d(m-1) + r$ for some $d \ge 0$ and $r \in \{0, \ldots, m-2\}$, and $m \ge 2$ machines be given.
    The greedy algorithm returns a schedule, in which there are assigned at most $\bettermaxclass$ jobs to every machine.
	Moreover, if $r = 0$, then the constructed schedule has at least one machine with a load strictly less than $\bettermaxclass$; otherwise, it has at least $m - r$ such machines.
\end{lemma}

\begin{proof}
The proof is by induction on the number of blocks.
For a graph having exactly one block the theorem holds: we assign at most one job to every machine, and clearly $\bettermaxclass \ge 1$.
It is also easy to verify that if $r = 0$ then it has to be the case $n = m - 1$ for the problem to be feasible, thus there is one machine with a load less than $1$.
Similarly, the cases for $n \le m - 2$ and $n = m$ hold.

Now, assume that the theorem holds for all graphs composed of $k$ blocks, and consider any block graph $G$ composed of $k+1$ blocks.
Let us fix a subgraph $G'$ of $G$ on $n'$ vertices including only the first $k$ blocks according to its block-cut tree pre-order traversal.
Let $n' = d' (m - 1) + r'$ be the number of the vertices of $G'$.
By induction assumption the greedy algorithm for $G'$ returns a schedule whose total makespan does not exceed $\ceil{n' / (m - 1)}$.

We schedule the remaining jobs from the last block according to the greedy algorithm, assigning at most one job to each machine.

If $r' = 0$ then it is clear that the limit on the number of jobs increased by at least one since $\bettermaxclass \ge d' + 1 = \ceil{n' / (m - 1)} + 1$. Thus the extended schedule satisfies the total makespan condition.
By considering separately cases $n - n' \le m - 2$, $n-n' = m - 1$, and $n-n' = m$ we can verify directly that we have the required number of machines that are processing less than $\bettermaxclass$ jobs.
\begin{itemize}
	\item In the first case we have at most $r = n-n'$ machines where the new jobs are assigned.
	Hence only these machines can have $\bettermaxclass$ jobs in the extended schedule. Thus $m - r$ machines have to be loaded with less than  $\bettermaxclass$ jobs.
	\item In the second case $r = r' = 0$ holds and clearly there is at least one machine with a load less than $\bettermaxclass$.
	\item The last case means that a clique on $m$ vertices is added -- the $(k+1)$-th block has no common vertex with other $k$ blocks.
	However, this means that the new upper limit is $\bettermaxclass = \ceil{n' / (m - 1)} + 2$ and hence there are $m$ machines with a load less than $\bettermaxclass$, in fact.
\end{itemize}

For $r' > 0$ we can split the proof into four cases.
In each case, by induction, there are at least $m - r'$ machines processing less than $\ceil{n' / (m - 1)}$ jobs.
However, due to the fact that the blocks may have a common cut-vertex, perhaps one of them is not available for scheduling.
\begin{itemize}
	\item If $r' + n-n' < m - 1$, or equivalently $n-n' < m - r - 1$, then all jobs of $G - G'$ can be scheduled on machines processing less than $\ceil{n' / (m - 1)}$ jobs.
	In the end $m - r' - (n-n') = m - r$ machines process less than $\bettermaxclass$ jobs.
	\item If $r' + n-n' = m - 1$ then we still schedule the jobs of $G - G'$ on the machines processing less than $\ceil{n' / (m - 1)}$ jobs.
	After extending the schedule there remains at least one machine with a load less than $\ceil{n' / (m - 1)}$.
	\item If $m \le r' + n-n' \le 2 m - 3$ then $r = r' + (n-n')-(m-1)$ and, by induction, there are at most $r'$ machines processing $\ceil{n' / (m - 1)}$ jobs.
	Consider the case that one of the machines processing less than $\ceil{n' / (m - 1)}$ cannot be scheduled upon.
	Moreover, let $n-n' > m - r' - 1$
	(we skip other subcases because they are similar and easier to prove).
	Observe that no more than $r' - (n - n' - (m-r'-1))$ machines were processing $\ceil{n' / (m - 1)}$ jobs and are processing $\bettermaxclass$ jobs now.
	Hence, there will be at least $r' - (n - n' - (m-r'-1)) + (m-r'-1) + 1$ machines that are processing less than $\bettermaxclass$.
	By a simple transformation we obtain the induction thesis.
	\item if $r' + n-n' \ge 2 m - 2$ then it has to be the case $r = 0$, $n-n' = m$, and $r' = m - 2$.
	Hence, again the added block is a clique on $m$ vertices.
	By induction, there were at least $1$ machine processing less than $\ceil{n' / (m - 1)}$ jobs.
	Clearly, at least this machine is processing less than $\bettermaxclass$ in the extended schedule.
\end{itemize}
In the first two cases $\bettermaxclass = \ceil{n' / (m - 1)}$, but the jobs are assigned only to machines that are processing less than $\ceil{n' / (m - 1)}$ jobs.
The machines are guaranteed to exist, by induction.
In the other cases $\bettermaxclass > \ceil{n' / (m - 1)}$.
Thus by assigning at most one additional job to every machine we cannot violate the total makespan condition.

We verified that both the total makespan condition and the condition on the number of machines processing less than $\bettermaxclass$ jobs are fulfilled.
Therefore we see that extending the scheduling for the last block is always possible which concludes the proof of the induction step.
\end{proof} 

Interestingly, \Cref{theorem:final-k-coloring-block-graphs} and \Cref{alg:2apx_identical_block} together with an already known PTAS for the fixed number of machines and conflict graphs with bounded treewidth are sufficient to construct a PTAS for identical machines, unit time jobs, and block conflict graphs which is the main result of this section.
\begin{theorem}
\label{thm:ptas_identical_unit_block}
There exists a PTAS for $P|G = \blockgraph, p_j = 1|\cmaxcost$.
\end{theorem}

\begin{proof}
    First let us assume that there exists a schedule with $\cmaxcost \le \frac{2}{\varepsilon}$.
    If this is the case, from \Cref{theorem:final-k-coloring-block-graphs}, it follows that we can find an optimal solution using \Cref{algorithm:main-loop} in time $\Osymbol(n^2 (2/\varepsilon)^3m^{14/\varepsilon + 8})$.

    Next let us suppose that $m \le \frac{2}{\varepsilon} + 1$.
    In this case we can use PTAS for $Pm|\treewidth(G) \le m|\cmaxcost$ from \cite{bodlaender1994scheduling}, i.e. $(1 + \delta)$-approximate algorithm running in time $\Osymbol(\delta^{-m} n^{m+1} m^{2 m})$. 
    For $\delta = \frac{1}{n+1}$ it is guaranteed to return the optimal solution in time $\Osymbol(n^{2/\/\varepsilon + 1}n^{2/\/\varepsilon + 2}(2/\/\varepsilon + 1)^{{4/\/\varepsilon} + 2})$.

    Finally, if $m > \frac{2}{\varepsilon} + 1$ and we did not find any optimal solution with $\optcmaxcost \le \frac{2}{\varepsilon}$ then we apply \Cref{alg:2apx_identical_block}.
    From \Cref{lem:equal_coloring_block} we infer that
    \begin{align*}
        \cmaxcost & < \frac{n}{m - 1} + 1
              < \frac{n}{m} \left(1 + \frac{1}{m - 1}\right) + \optcmaxcost \cdot \frac{\varepsilon}{2}
              < \frac{n}{m} \cdot(1 + \frac{\varepsilon}{2}) + \optcmaxcost \cdot \frac{\varepsilon}{2}
              \le \optcmaxcost \left(1 + \varepsilon\right).
    \end{align*}
    In the last inequality we used a simple lower bound $\optcmaxcost \ge \left\lceil\frac{n}{m}\right\rceil$.
\end{proof}

\section{Uniform machines}\label{sec:uniform}

\subsection{Bounded number of cut-vertices, unit time jobs}

To find an optimal solution for the problem $Q|G = \blockgraph, \cutvertices(G) \le k, p_j = 1|\cmaxcost$, we combine three ideas: binary search over all possible objective values, exhaustive search over all possible valid assignment of all cut-vertices to machines (i.e. assigning adjacent cut-vertices to different machines), and solving an auxiliary flow problem.

To simplify the problem we concentrate on the subproblem in which we guessed the value $C$ and we have to determine if there exists a feasible schedule with a makespan at most $C$.
Observe that in any schedule the optimal $\cmaxcost$ is determined by the number of jobs assigned to some machine.
Hence there can $\Osymbol(mn)$ candidates for $\cmaxcost$.
Each of the candidate makespans determines the capacities of the machines.


Let us also introduce the auxiliary flow problem. For some fixed value of the objective $C$ and a fixed valid assignment of all cut-vertices jobs $f\colon V_{cut} \to \M$, we build the flow network $F(G, C, f)$ with five layers (see \Cref{fig:flow_network} for an example):
\begin{itemize}
    \item first layer contains only the source node $s$ and last layer contains only the sink node $t$,
    \item second layer contains nodes for all jobs respective to the simplicial vertices in $G$ i.e. $J_i \in \J \setminus V_{cut}$,
    \item third layer contains nodes for all pairs of blocks (i.e. connected components after removal of cut-vertices) $B_k \in V_B$ and machines $M_j \in \M$,
    \item fourth layer contains nodes for all machines $M_j \in \M$,
\end{itemize}
We add all edges between the first two layers and between the last two layers.
Moreover, we add edges between nodes in the second (for some $J_i \in \J \setminus V_{cut}$) and third layer (for some $(B_k, M_j) \in V_B \times \M$) if and only if $J_i \in B$ and for every $v \in V_{cut}$ adjacent to $J_i$ we have $f(v) \neq M_j$.
Finally, we add edges between nodes in the third and fourth layer whenever they refer to the same machine $M_j$.

All edges but the ones ending in the sink node $t$ have capacity $1$.
An edge between $M_j$ and $t$ has weight $w_j = \floor{C s_j / l} - |f^{-1}(M_j)|$, where $f^{-1}(M_j)$ is the set of cut-vertices of $G$ that are assigned to the machine $M_j$ according to $f$.

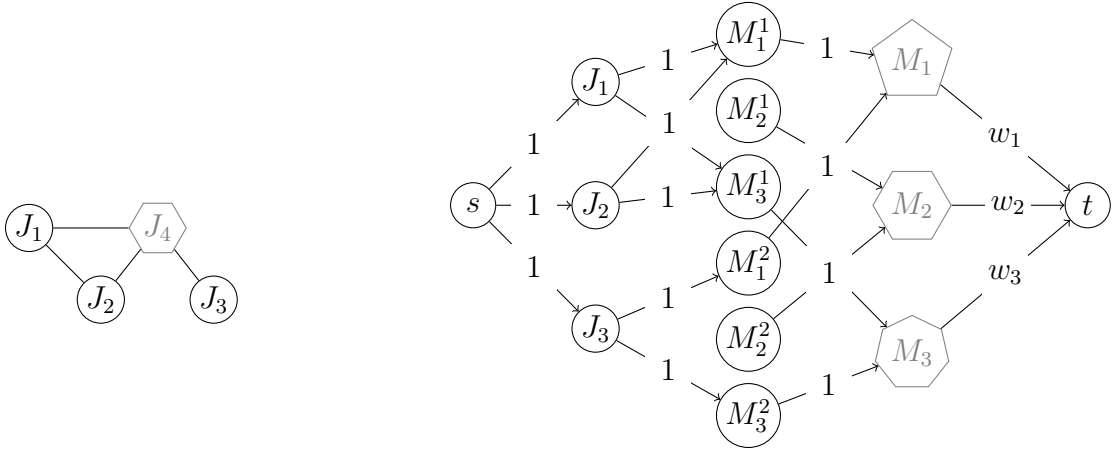
\begin{figure}[htpb]
\centering
\begin{subfigure}{0.2\textwidth}
\begin{tikzpicture}[align=center,node distance=2cm, main/.style = {draw, circle}, minimum size=6mm, inner sep=1pt]
	\node[main] (J1) {$J_1$};
	\node[main] (J2) [below right =0.5cm and 0.5 of J1]{$J_2$};
	\node[main] (J4) [right =1cm of J1, color=gray,regular polygon, regular polygon sides=6]{$J_4$};
	\node[main] (J3) [below right =0.5cm and 2 of J1]{$J_3$};
        \node[] (dummy) [below=2cm of J1]{};
	
	\draw (J4) -- (J1) -- (J2) -- (J4) -- (J3);
\end{tikzpicture}
\end{subfigure}
\qquad\qquad\qquad
\begin{subfigure}{0.5\textwidth}
\begin{tikzpicture}[align=center, main/.style = {draw, circle}, minimum size=6mm, inner sep=1pt]
	\node[main, inner sep=1pt] (s){$s$};
 
	\node[main] (J2) [right=1cm of s]{$J_2$};
	\node[main] (J1) [above=1cm of J2]{$J_1$};
	\node[main] (J3) [below=1cm of J2]{$J_3$};

	\node[main] (M11) [above right=0.1cm and 1.5cm of J1]{$M_1^1$};
	\node[main] (M21) [below=0.15cm of M11]{$M_2^1$};
	\node[main] (M31) [below=0.15cm of M21]{$M_3^1$};

	\node[main] (M12) [below=0.15cm of M31]{$M_1^2$};
	\node[main] (M22) [below=0.15cm of M12]{$M_2^2$};
	\node[main] (M32) [below=0.15cm of M22]{$M_3^2$};
	
	\node[main,regular polygon, regular polygon sides=6] (M2) [right=5cm of s,color=gray]{$M_2$};
	\node[main,regular polygon, regular polygon sides=5] (M1) [above=1cm of M2,color=gray]{$M_1$};
	\node[main,regular polygon, regular polygon sides=7] (M3) [below=1cm of M2,color=gray]{$M_3$};

	\node[main, inner sep=1pt] (t) [right=1.5cm of M2]{$t$};

	\draw[->] (s) -- (J1) node [midway, fill=white] {$1$};
	\draw[->] (s) -- (J2) node [midway, fill=white] {$1$};
	\draw[->] (s) -- (J3) node [midway, fill=white] {$1$};

	\draw[->] (J1) -- (M11) node [midway, fill=white] {$1$};
	\draw[->] (J1) -- (M31) node [midway, fill=white] {$1$};
	\draw[->] (J2) -- (M11) node [midway, fill=white] {$1$};
	\draw[->] (J2) -- (M31) node [midway, fill=white] {$1$};

	\draw[->] (M11) -- (M1) node [midway, fill=white] {$1$};
	\draw[->] (M21) -- (M2) node [midway, fill=white] {$1$};
	\draw[->] (M31) -- (M3) node [midway, fill=white] {$1$};

	\draw[->] (M12) -- (M1) node [midway, fill=white] {$1$};
	\draw[->] (M22) -- (M2) node [midway, fill=white] {$1$};
	\draw[->] (M32) -- (M3) node [midway, fill=white] {$1$};

	\draw[->] (J3) -- (M12) node [midway, fill=white] {$1$};
	\draw[->] (J3) -- (M32) node [midway, fill=white] {$1$};

	\draw[->] (M1) -- (t) node [midway, fill=white] {$w_1$};
	\draw[->] (M2) -- (t) node [midway, fill=white] {$w_2$};
	\draw[->] (M3) -- (t) node [midway, fill=white] {$w_3$};
\end{tikzpicture}
\end{subfigure}
\caption{An example graph $G$ with $2$ blocks and $f(J_4) = M_2$, and its flow network $F(G, C, f)$.}
\label{fig:flow_network}
\end{figure}

Now we proceed with a lemma describing the relation between the flow network and the original task scheduling problem:
\begin{lemma}
     For an instance of a problem $Q|G = \blockgraph, \cutvertices(G) \le k, p_j = 1|\cmaxcost$ there exists scheduling with total makespan at most $C$ if and only if for some valid function $f\colon V_{cut} \to \M$ there exists a flow network $F(G, C, f)$ with a maximum flow equal to $n - \cutvertices(G)$.
     \label{alg:flow_network}
\end{lemma}

\begin{proof}
    First note that the edges between the first two layers ensure that the maximum flow in $F(G, C, f)$ cannot exceed $n - \cutvertices(G)$.

    Assume that $f(v) = \sigma(v)$ for all $v \in V_{cut}$ for some optimal schedule $\sigma$.
    If there exists a maximum flow equal to $n - \cutvertices(G)$, then all the edges outgoing from $s$ are saturated.
    Moreover, since all capacities are integer values, there exists a maximum flow such that its value for every edge is also some integer value.
    Therefore, we assign $J_i$ to $M_j$ if there is a non-zero flow between nodes corresponding to $J_i$ and its adjacent $(B_k, M_j)$.

    Each node in the third layer has exactly one outgoing edge with a capacity $1$. This ensures that no two simplicial vertices from the same block get assigned to the same machine.
    Moreover, if some cut-vertex $v$ is assigned to the machine $M_j$, then there are no edges between $J_i$ and $(B_k, M_j)$ for any job $J_i$ adjacent to $v$ in $G$ and any block $B_k$ containing $v$.
    Finally, the edges between the last two layers ensure that the obtained scheduling has $\cmaxcost$ at most equal to $C$.

    The converse proof goes along the same lines as it is sufficient to use the converse mapping to the one presented above to pick saturated edges between the second and fourth layer: if a simplicial vertex $J_i$ from block $B_k$ is scheduled on machine $M_j$, then saturate edges $(J_i, (B_k, M_j))$ and $((B_k, M_j), M_j)$. In addition, we saturate all edges outgoing from $s$ and we compute the flows through the vertices in the fourth layer to get the values of flow incoming to $t$.
\end{proof}

Now we are ready to present the complete algorithm for $Q|G = \blockgraph, \cutvertices(G) \le k, p_j = 1|\cmaxcost$:
\begin{algorithm}[htpb]
\begin{algorithmic}
\REQUIRE A set of jobs $\J$, a set of machines $\M$, a block graph $G$, function $s\colon \M \to \N$, and a~guessed makespan $C$.
\FORALL{assignments $f\colon V_{cut} \to \M$}
    \FORALL{$u, v \in V_{cut}$}
        \STATE \textbf{if} $\{u, v\} \in E(G)$ \textbf{and} $f(u) = f(v)$ \textbf{then continue}
    \ENDFOR
    \FOR{$i = 1, \ldots, m$}
        \STATE \textbf{if} $|\{u \in V_{cut}\colon f(u) = i\}| > C$ \textbf{then continue}
        \hfill\COMMENT{guessed capacity exceeded}
    \ENDFOR
    \STATE Build the flow network $F(G, C, f)$ and solve the maximum flow problem on it
    \IF{maximum flow for $F(G, C, f)$ is equal to $n - \cutvertices(G)$}
        \RETURN a scheduling retrieved from the maximum flow for $F(G, C, f)$
    \ENDIF
\ENDFOR
\RETURN \NO
\end{algorithmic}
\caption{The core part of our algorithm for $Q|G = \blockgraph, \cutvertices(G) \le k, p_j = 1|\cmaxcost$}
\label{alg:uniform_flow_network}
\end{algorithm}

\begin{theorem}
    There exists a polynomial time algorithm for $Q|G = \blockgraph, \cutvertices(G) \le k, p_j = 1|\cmaxcost$ in time $\Osymbol(m^{k + 2} n^2 \log(mn))$.
    \label{thm:uniform_block_unit_cut}
\end{theorem}
\begin{proof}
	Consider \Cref{alg:uniform_flow_network}.
    From \Cref{alg:flow_network} one directly obtains the equivalence between the existence of a flow in the constructed network and the existence of a corresponding schedule: it is sufficient to generate all valid assignments $f$, build the respective flow networks $F(G, C, f)$ and check whether for any of them there exists a maximum flow with a value equal to $n - \cutvertices(G)$ to verify the existence of a schedule with makespan at most $C$.

    For every $C \ge \optcmaxcost$, we will find at least one such solution (i.e. with $f$ exactly the same as the assignment of jobs from $V_{cut}$ in some optimal schedule $\sigma$) -- and for $C < \optcmaxcost$ it is impossible to find any feasible solution.

    The number of iterations for the binary search is at most $\Osymbol(\log(nm))$.
    Clearly, for each fixed possible makespan value $C$, there are at most $m^k$ different assignments $f$. Each one can be checked whether it is valid (i.e. it does not assign adjacent jobs in $G$ to the same machine and a load of each machine does not exceed $C$) in $\Osymbol(n)$ time.

    Finally, any flow network $F(G, C, f)$ has $\Osymbol(m n)$ vertices and $\Osymbol(m n)$ edges, thus the maximum flow problem can be solved in $\Osymbol(m^2 n^2)$ time e.g. using Orlin algorithm \cite{orlin2013max}.
    
\end{proof}

Note that Furma\'nczyk and Mkrtchyan in \cite{furmanczyk2020graph} proved that the problem of \textsc{Equitable Coloring} for block graphs is FPT with respect to the number of cut-vertices. Their result can be expressed in the language of unit-time job scheduling on identical machines with the relevant conflict graph. So, our polynomial-time algorithm for uniform machines can be seen as a partial generalization of their result.

\subsection{Bounded number of blocks}

Let us now consider a more restricted case, when the number of blocks is bounded by a constant, but with arbitrary job processing requirements. We begin with a simple $k$-approximation algorithm.

\begin{algorithm}[htpb]
\begin{algorithmic}
\REQUIRE A set of jobs $\J$, a set of machines $\M$, a block graph $G$, functions $p\colon \J \to \N$ and $s\colon \M \to \N$.
\STATE $\sigma_{best} \gets$ a dummy schedule with $C_{max} = \infty$ \COMMENT{to simplify the pseudocode}
\STATE Sort machines by their non-increasing speeds $s_j$
\FORALL{assignments $f\colon V_{cut} \to \M$}
    \FORALL{$u, v \in V_{cut}$}
        \STATE \textbf{if} $\{u, v\} \in E(G)$ \textbf{and} $f(u) = f(v)$ \textbf{then continue}
    \ENDFOR
    \STATE $\sigma \gets f$
    \FORALL{$B \in V_B$}
        \STATE $L \gets$ unassigned jobs from $B$ sorted by their non-increasing processing requirements $p_j$
        \STATE $j \gets 0$
        \FOR{$i = 1, 2, \ldots, |L|$}
            \STATE \textbf{while} $\sigma(v) = M_{i + j}$ for some $v \in V_{cut} \cap B$ \textbf{do} $j \gets j + 1$
            \STATE Assign $L_i$ to $M_{i + j}$ in $\sigma$
        \ENDFOR
    \ENDFOR
    \IF{$C_{max}(\sigma) < C_{max}(\sigma_{best})$}
        \STATE $\sigma_{best} \gets \sigma$
    \ENDIF
\ENDFOR
\RETURN $\sigma_{best}$
\end{algorithmic}
\caption{A $k$-approximation algorithm for $Q|G = \kblockgraph{k}|\cmaxcost$.}
\label{alg:kaproximateuniformmachines}
\end{algorithm}

\begin{theorem}
\Cref{alg:kaproximateuniformmachines} is a $k$-approximation strongly polynomial time algorithm for $Q|G = \kblockgraph{k}|\cmaxcost$.
\end{theorem}

\begin{proof}
    Since we iterate over all possible assignments of cut-vertices to machines we only need to prove that a schedule $\sigma$ with the same assignment of these jobs as in an optimal schedule has a makespan at most $k$ times greater than the optimal makespan.

    Observe that if $G$ has at most $k$ blocks, then there are at most $k$ jobs assigned to every machine.
    Let us denote by $M_i$ any machine with processing time equal to the makespan $C_{max}(\sigma)$, and by $J_j$ a job assigned to it with a maximum processing time. We have two possibilities:
    \begin{enumerate}
        \item $J_j \in V_{cut}$: since it is assigned to the same machine as in some optimal schedule, then $\optcmaxcost \ge \frac{p_j}{s_i}$; moreover, any other job on $M_i$ contributes at most $\frac{p_j}{s_i}$ to the makespan, thus $C_{max}(\sigma) \le k \cdot \frac{p_j}{s_i}$,
        \item $J_j \notin V_{cut}$: now $J_j \in B$, for some block $B$, could be scheduled on a different machine in the optimal schedule.
        Due to the greedy assignment, we know that there are at least $i - |V_{cut} \cap B|$ jobs of not smaller size that are simplicial.
        This means that some machine $M_{i' \ge i}$ has to be assigned a job of size at least $p_j$ in the optimal schedule.
        Hence as well $\optcmaxcost \ge \frac{p_j}{s_i}$.
    \end{enumerate}
    In either case we know that at least one schedule under consideration has to be $k$-approximate.

    Note that the running time of this algorithm is bounded by the number of all assignments from $V_{cut}$ to $M$ and by sorting of jobs for each block.
    Since in total they are bounded by $O(m^k)$ and $O(n \log{n})$, respectively, the running time of this algorithm is $O(n m^k \log{n})$.
\end{proof}

The above algorithm will come handy later but now we would like to point out that we can find a considerable improvement to it, i.e. a PTAS for the same problem.

\begin{algorithm}[htpb]
\begin{algorithmic}
\REQUIRE A set of jobs $\J$, a set of machines $\M$, a block graph $G$, functions $p\colon \J \to \N$ and $s\colon \M \to \N$, and a~guessed value $C$.
\FOR{$i = 1, \ldots, m$}
  \STATE Calculate capacity $c_i$ as $C \cdot s(M_i)$ rounded up to the nearest power of $1 + \varepsilon$
\ENDFOR
\STATE Sort machines in $\M$ in a non-decreasing order of $c_i$
\FOR{$j = 1, \ldots, n$}
    \STATE Round $p_j$ down to the nearest power of $1 + \varepsilon$
\ENDFOR
\IF{$\max_j p_j > \max_i c_i$}
   \RETURN \NO
\hfill   \COMMENT{Necessary condition, guarantees visiting all jobs}
\ENDIF
\STATE $\tau \gets \ceil{\log_{1+\varepsilon}(k/\varepsilon)}$
\STATE $l_0 \gets 0$
\STATE Compute $U_0$ -- a set with a single configuration for basis $l_0$
\FOR{$i = 1, \ldots, m$}
    \STATE $U_i \gets \emptyset$, $l_i \gets \max\{\log_{1 + \varepsilon}(c_i) - \tau, 0\}$
    \FOR{$u \in U_{i - 1}$}
        \STATE $shift \gets \min\{l_{i} - l_{i-1}, \tau\}$
        \hfill\COMMENT {The jobs that are not scheduled yet and become very small}
        \FOR{$s = 1, \ldots, k$}
        	\STATE $n_{0,s} \gets \sum_{i = 0}^{shift} n_{i,s} + |\{J_j \in B_s \setminus V_{cut}\colon (1+\varepsilon)^{l_i} >  p_j > (1+\varepsilon)^{l_{i - 1} + \tau} = c_{i-1}\}|$
            \FOR{$i = 1, \ldots, \tau - shift$}
                \STATE $n_{i, s} \gets n_{i + shift, s}$ \hfill\COMMENT{Convert $u$ from basis $l_{i - 1}$ to basis $l_i$}
            \ENDFOR
            \FOR{$i = \tau - shift + 1, \ldots, \tau$}
            	\STATE $n_{i, s} \gets |\{J_j \in B_s \setminus V_{cut}\colon p_j = (1+\varepsilon)^{l_i + i}\}|$
            	\hfill\COMMENT{The jobs not yet considered}
            \ENDFOR
            \FORALL{$J_j \in B_s \setminus V_{cut}$}
                \IF{$p_j < (1 + \varepsilon)^{l_i}$}
                    \STATE $p_j \gets 0$ \hfill\COMMENT{Discard sizes of the small jobs}
                \ENDIF
            \ENDFOR
        \ENDFOR
        \FORALL{$u' \in \{0, 1\}^{|u|}$}
            \IF{$u' \le u$}
                \STATE \textbf{if} $u'$ represents a set of jobs with at least two jobs from a single block \textbf{then continue}
                \STATE \textbf{if} $u'$ represents a set of jobs with the total processing time greater than $c_i$ \textbf{then continue}
                \STATE Add $u - u'$ to $U_i$ \hfill\COMMENT{Schedule jobs from $u'$ on $M_i$ and store the configuration of remaining jobs}
            \ENDIF
        \ENDFOR
    \ENDFOR
\ENDFOR
\IF{$\mathbf{0} \in U_m$}
    \RETURN any schedule corresponding to $\mathbf{0} \in U_m$
\ENDIF
\RETURN \NO
\end{algorithmic}
\caption{The core part of our algorithm for $Q|G = \kblockgraph{k}|\cmaxcost$}
\label{alg:ptas_uniform_block}
\end{algorithm}

\begin{theorem}
	Binary search used in conjunction with \Cref{alg:ptas_uniform_block} is a PTAS for $Q|G = \kblockgraph{k}|\cmaxcost$.
	\label{thm:uniform_kblock_cut}
\end{theorem}

\begin{proof}
    To prove the theorem it is sufficient to consider \Cref{alg:ptas_uniform_block} for a fixed guessed value $C$.
    Throughout the algorithm instead of speeds $s(M_i)$ we deal rather with capacities $c_i$ for all machines, which are directly determined by their speeds and the guessed value.

    We iterate over all machines in the order of their non-decreasing capacities and we try every feasible combination of jobs for each machine.
    Clearly, if $M_i$ has capacity $c_i$, then we only need to take into account jobs with $p_j \le c_i \le (1+\varepsilon)^{l_i + \tau}$. Additionally, we need to consider cut-vertices separately, since they belong to multiple blocks at the same time.
    Note that we do not need to differentiate jobs with the same rounded processing requirements coming from the same block, as long as they all correspond to simplicial vertices in $G$.

    The crucial idea is to use \emph{configurations} of unscheduled jobs grouped by their rounded size. Formally, a configuration for a given $\tau \in \mathbb{N}$ (fixed throughout the algorithm, depending only on $\varepsilon$ and $k$) and value $l \ge 0$ (called the \emph{basis}) is defined as a tuple $(n_{0, 1}, \ldots, n_{0, k}; n_{l, 1}, \ldots, n_{l, k}; \ldots; n_{l + \tau, 1}, \ldots, n_{l + \tau, k}; a_{1}, \ldots, a_{cut(G)})$, where:
    \begin{enumerate}
        \item $n_{0, s}$ is the number of unscheduled jobs from block $B_s \setminus V_{cut}$ with $p_j < (1 + \varepsilon)^{l}$ for $s = 1, \ldots, k$;
        \item $n_{i, s}$ is the number of unscheduled jobs from block $B_s \setminus V_{cut}$ with $p_j = (1+\varepsilon)^{i}$ for $s = 1, \ldots, k$, $i = l, \ldots, l + \tau$;
        \item $a_i \in \{0, 1\}$ denotes whether the $i$-th cut-vertex is not yet scheduled for $i = 1, \ldots |V_{cut}|$.
    \end{enumerate}

    Additionally, we need an operation of conversion of a configuration from a basis $l$ to some other basis $l' \ge l$.
    This is done just by:
    \begin{itemize}
        \item taking the original configuration $u$ with basis $l$ and recomputing all $n_{i, s}$ values with respect to the new basis (i.e. shifting from $n_{i, s}$ to $n_{\max\{i + (l - l'), 0\}, s}$ for all $s = 1, \ldots, k$,
        \item iterating over all jobs from $V(G) \setminus V_{cut}$ with $(1 + \varepsilon)^l < p_j \le (1 + \varepsilon)^{l'}$ and counting them in the new configuration at their proper place.
    \end{itemize}
    We refer to \Cref{fig:conversion} for a high-level outline of a conversion operation.

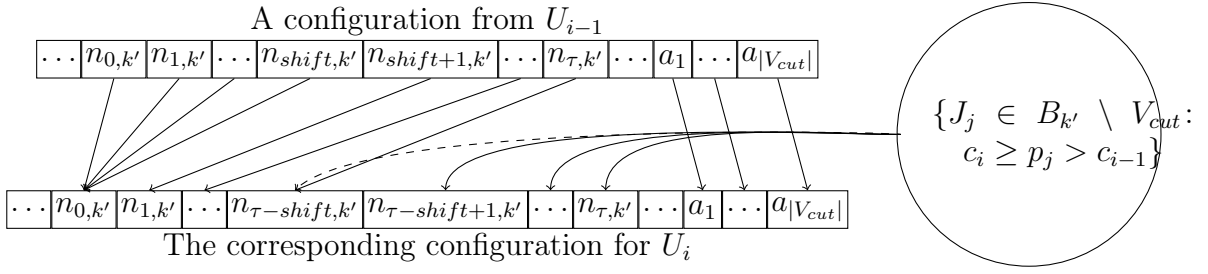
\begin{figure}[H]
	\centering
	\begin{tikzpicture}
		\coordinate (v1pos) at (0,0);
		\coordinate (v1primpos) at (0,-2);
		\coordinate (jobspos) at (8, -1);
		\matrix [align=center,  inner sep=0.05cm, nodes={draw, anchor=center, minimum size=.5cm}] at (v1pos)
		{
			\node {$\ldots$}; & \node (n0org) {$n_{0,k'}$}; & \node (n1org) {$n_{1,k'}$}; & \node (n2org) {$\ldots$}; & \node (nshiftorg) {$n_{shift,k'}$}; & \node (nshift1org) {$n_{shift+1,k'}$}; & \node (nshift2org) {$\ldots$}; & \node (ntauorg) {$n_{\tau,k'}$}; & \node {$\ldots$}; & \node (a1org) {$a_1$}; & \node (a2org) {$\ldots$};  & \node (alastorg) {$a_{|V_{cut}|}$};   \\
		};
		\matrix [align=center,  inner sep=0.05cm, nodes={draw, anchor=center, minimum size=.5cm}] at (v1primpos)
		{
			\node {$\ldots$}; & \node (n0mod) {$n_{0,k'}$}; & \node (n1mod) {$n_{1,k'}$}; & \node (n2mod) {$\ldots$}; & \node (ntauminushiftmod) {$n_{\tau-shift, k'}$}; & \node (nnextmod) {$n_{\tau-shift+1, k'}$}; & \node (nnext2mod) {$\ldots$};  & \node (ntaumod) {$n_{\tau,k'}$}; & \node {$\ldots$}; & \node (a1mod) {$a_1$}; & \node (a2mod) {$\ldots$};  & \node (alastmod) {$a_{|V_{cut}|}$};   \\
		};
		\node at ([shift={(0,0.5)}]v1pos) {A configuration from $U_{i-1}$};
		\node at ([shift={(0,-0.5)}]v1primpos) {The corresponding configuration for $U_{i}$};
		\node [draw, circle, text width =3cm] at (jobspos) (jobs) {
                \begin{tabular}{l}
                    $\{J_j~\in~B_{k'}~\setminus~V_{cut}\colon$ \\
                    \quad $c_{i} \ge p_j > c_{i-1} \}$
                \end{tabular}};
		\draw[->] (n0org.south) -- (n0mod.north);
		\draw[->] (n1org.south) -- (n0mod.north);
		\draw[->] (n2org.south) -- (n0mod.north);
		\draw[->] (a1org.south) -- (a1mod.north);
		\draw[->] (a2org.south) -- (a2mod.north);
		\draw[->] (alastorg.south) -- (alastmod.north);
		\draw[->] (nshiftorg.south) -- (n0mod.north);
		\draw[->] (nshift1org.south) -- (n1mod.north);
		\draw[->] (nshift2org.south) -- (n2mod.north);
		\draw[->] (ntauorg.south) -- (ntauminushiftmod.north);
		\draw[->] (jobs.west) .. controls +(-1,0) and +(0,1) .. node[above] {} (ntaumod.north);
		\draw[->] (jobs.west) .. controls +(-1,0) and +(0,1) .. node[above] {} (nnextmod.north);
		\draw[->] (jobs.west) .. controls +(-1,0) and +(0,1) .. node[above] {} (nnext2mod.north);
		\draw[dashed, ->] (jobs) .. controls +(-1,0) and +(1,1) .. node[above] {} (ntauminushiftmod.north);
	\end{tikzpicture}
	\caption{
		A high-level overview of the conversion of configuration sets $U_{i-1}$ to $U_i$ in \cref{alg:ptas_uniform_block} for some block $V_{k'}$ is illustrated.
		Keep in mind that when transforming the configurations, the jobs that have processing requirement $c_i \ge p_j > c_{i-1}$ and are not cut-vertices are added.
		Finally, if the difference between the capacities of the machines is big (and $shift \ge \tau$) then it might be the case that some new jobs will be added to $n_{0,k'}$; the dashed line is added to emphasize that it is the only case when the same field is filled with values from both: the old vector and the set of jobs.
	}
        \label{fig:conversion}
\end{figure}    

    Given such notation, at each iteration we build set $U_i$ containing configurations remaining after considering the $i$ slowest machines (i.e. the machines with the smallest capacities).
    Clearly, $U_0$ contains exactly one configuration for basis $l_0$, which can be obtained by iterating over all jobs and incrementing respective variables $n_{i, s}$. Of course $a_i = 1$ for all $i = 1, \ldots, |V_{cut}|$ since at the beginning all cut-vertices are not assigned to any machine.
    And if $U_m$ contains a configuration $\mathbf{0}$ and every job was considered at some point (guaranteed by necessary condition $\max_j p_j \le \max_i c_i$), then all jobs are scheduled, thus the algorithm found some feasible schedule.
  
    Note that the algorithm ensures that the sum of processing times (after rounding down and zeroing out the smallest jobs) for jobs scheduled on $M_i$ does not exceed $c_i$.
    Thus, if there were no rounding down and zeroing out, we would get a schedule with makespan $C$.
    However, both of these operations do not take us far away from our guess.
    Observe that for any graph with $k$ blocks, there can be only up to $k$ jobs assigned to a single machine $M_i$. Therefore, the total processing time of all jobs denoted by $n_{0, s}$ variables assigned to any single machine cannot exceed $k(1+\varepsilon)^{l_i} \le \varepsilon (1+\varepsilon)^{l_i+\tau} = \varepsilon c_i$ -- thus the total processing time on $M_i$ is no greater than $(1 + \varepsilon) c_i$.
    Moreover, by de-rounding of all $p_j$ we increase the total processing times on all machines at most $1+\varepsilon$ times.
    Therefore, if we find a schedule for a guessed value $C$, then we are certain that its makespan does not exceed $(1 + \varepsilon)^2 C$.

    Now suppose there exists a schedule with makespan $C$. We claim that \Cref{alg:ptas_uniform_block} always finds a schedule for a guess $C$.
    First, let us consider any schedule $\sigma$ with makespan $C$. Let us also denote by $u_i$ ($i = 0, \ldots, m$) its respective configurations of unscheduled jobs after the first $i$ machines were scheduled.
    We can prove inductively that $u_i \in U_i$ since for any $u_{i - 1} \in U_{i - 1}$ we consider all possible schedulings up to the roundings of $p_j$ and not considering the processing times of jobs denoted by $n_{0, s}$ variables.
    If $u'$ denotes a scheduling of jobs for machine $M_i$ in $\sigma$, then it is also feasible in the algorithm since both rounding down and zeroing out some processing times cannot increase the total processing time of any set of jobs. 
    Thus, since $\sigma$ is a feasible schedule for a guess $C$, then there has to exist some schedule corresponding to $\mathbf{0}$ in $U_m$.
    This completes the correctness proof.

    Now let us estimate the total running time of the algorithm.
    Note that for any fixed $\tau$ and $l$ each variable $n_{i, s}$ ($i = 0, l, \ldots, l + \tau$, $s = 1, \ldots, k$) has only $\Osymbol(n)$ possible values.
    Thus, there are $\Osymbol(n^{k (\tau + 1)} 2^k)$ distinct configurations of unscheduled jobs in each $U_i$.
    
    The running time of a single iteration of the main loop is dominated by computing a set of feasible assignments to a machine. We have $\Osymbol(n^{k (\tau + 1)} 2^k)$ configurations from a previous iteration.
    Every configuration can be converted to basis $l_i$ in $\Osymbol(n)$ time.
    There are $\Osymbol(2^{k (\tau + 2)})$ possible vectors $u'$, and checking every single one whether it obeys the requirements can be done in $\Osymbol(n)$ time.
    Taking everything into account we get the running time bound of $\Osymbol(n^{k (\tau + 1)} 2^{k (\tau + 3)})$ i.e. clearly a bound of form $\Osymbol(n^{f(\varepsilon)})$.
    Combining this with a binary search over a set of possible makespans we get this running time with an additional logarithmic factor.
\end{proof}

There is a small caveat concerning the total running time.
In the simplest approach we can just ensure an integer value on $\optcmaxcost$ by multiplying the processing requirements of all jobs by $\prod_{i=1}^m s_i$.
Then $\optcmaxcost \in \{1, 2, \ldots, kp_{\max}s_{\max}^m\}$, hence the binary search runs $\Osymbol(m\log (ks_{\max}p_{\max}))$ times.

However, this guarantees only a weakly polynomial time bound on running time.
To overcome this problem we use the well-known trick by Hochbaum and Shmoys (see \cite{hochbaum1988}, Theorem 1).
We just need to run \Cref{alg:kaproximateuniformmachines} first, obtaining a schedule with $C \in [\optcmaxcost,  k \optcmaxcost]$.
Then, although the application of the usual binary search in $\log{k}$ iterations fails to get us to the optimal solution, it was shown that it is possible to modify this procedure to arrive in $\log{k}$ + $\log{\frac{3}{\varepsilon}}$ carefully designed iterations at the $(1 + \varepsilon)$-approximate solution -- which of course still preserves the properties of PTAS.

In a sense this result is the best possible, due to the strong NP-hardness of the problem \cite{page2020makespan}.

\section{Unrelated machines}\label{sec:unrelated}

\subsection{Fixed number of machines, graph with bounded treewidth}

Here we provide an algorithm for $Rm|\treewidth(G) \le k|\cmaxcost$ problem, capturing $Rm|G = \blockgraph|\cmaxcost$ for block graphs with bounded clique number as a special case.
As with most of the algorithms for graphs with bounded treewidth, we work on $T$, a tree decomposition of graph $G$ \cite{cygan2015parameterized}.
To simplify the procedure, we preprocess $T$ by duplicating its bags and introducing empty ones to obtain a tree where each internal node has exactly two children. Such \textit{extended} tree decomposition has overall $\Osymbol(n)$ vertices.

In our reasoning we were inspired by an algorithm for $Pm|\treewidth(G) \le k|\cmaxcost$ problem from \cite{bodlaender1994scheduling}, albeit we use their steps in a different order.

\begin{algorithm}[htpb]
\begin{algorithmic}
\REQUIRE A set of jobs $\J$, a set of machines $\M$, a block graph $G$, function $p\colon \J \times \M \to \N$, and a~guessed makespan $C$.
\STATE Round down all $p_{i,j}$ to the nearest multiple of $\frac{\varepsilon C}{n}$
\STATE $V_B^O \gets$ a sequence of bags given by a post-order traversal of an extended tree decomposition $T$ of $G$
\FORALL{$B \in V_B^O$}
    \IF{$B$ is a leaf in $T$}
        \STATE $S(B) \gets$ a set of all possible schedules of jobs from $B$ with $\cmaxcost \le C$
    \ELSE
        \STATE $B_1, B_2 \gets$ children of $B$ in $T$
        \STATE $S' \gets$ a set of all possible schedules of jobs from $B$ with $\cmaxcost \le C$
        \STATE $S(B_1), S(B_2) \gets$ sets of already computed possible schedules for $B_1$, $B_2$
        \STATE $S(B) \gets \emptyset$
        \FORALL{$(s', s_1, s_2) \in S' \times S(B_1) \times S(B_2)$}
            \STATE \textbf{if} jobs from $B \cap B_1$ are on different machines in $s'$ and $s_1$ \textbf{then continue}
            \STATE \textbf{if} jobs from $B \cap B_2$ are on different machines in $s'$ and $s_2$ \textbf{then continue}
            \STATE $s \gets$ merged partial schedules $s'$, $s_1$, and $s_2$
            \IF{$\cmaxcost(s) \le C$}
                \STATE $S(B) \gets S(B) \cup \{s\}$
            \ENDIF
        \ENDFOR
        \STATE Trim $S(B)$ to contain unique machine loads or assignment of jobs from $B$
    \ENDIF
\ENDFOR
\RETURN any schedule from $S(r)$ \textbf{if} $S(r) \neq \emptyset$ \textbf{otherwise} \texttt{NO}
\end{algorithmic}
\caption{The core part of our algorithm for $Rm|\treewidth(G) \le k|\cmaxcost$}
\label{alg:arbitrary_tw}
\end{algorithm}

\begin{theorem}
    \Cref{alg:arbitrary_tw} combined with binary search is an FPTAS for $Rm|\treewidth(G) \le k|\cmaxcost$ problem.
    It runs in total time equal to $\Osymbol(n \log(n p_{max}) \cdot m^{3 k + 3} \cdot \ceil{\frac{n}{\varepsilon}}^{2m})$.
    \label{thm:arbitrary_tw}
\end{theorem}

\begin{proof}
    The algorithm will perform a binary search on the value of $\cmaxcost$ and in each turn it will check for feasibility. For an initial upper bound of $\cmaxcost$ let us take $n p_{max}$.

    Note, that the maximum possible difference between the initial and rounded processing times is equal to $n \cdot \frac{\varepsilon \optcmaxcost}{n} = \varepsilon \optcmaxcost$.

    To prove the correctness of the algorithm, it is sufficient to note that on the one hand, if it does not return \texttt{NO}, then it always returns a valid schedule.

    On the other hand, it is easy to note that every partial schedule $s \in S(B)$:
    \begin{enumerate}
        \item assigns all jobs included in a subtree of $T$ rooted in $B$,
        \item satisfies all incompatibility conditions between jobs in the subtree of $T$ rooted in $B$,
        \item has a total makespan at most $C$.
    \end{enumerate}
    Moreover, by construction we ensure that we include schedules with all possible current loads of all machines and assignments of jobs to machines in the current bag $B$.

    Thus, if we consider an optimal schedule, then it appears ultimately in $S(r)$ as long as $C \ge (1 + \varepsilon) \optcmaxcost$ -- or it can be excluded at some point during trimming, but then there remains another, equivalent schedule in terms of current loads and assignments of jobs from the current bag, which eventually will provide some equivalent complete schedule in terms of loads with rounded processing times.
    
    To compute the time complexity, let us first consider how large any set $S(B)$ can be.
    We can have at most $\ceil{\frac{n}{\varepsilon}}$ different loads for each machine, thus $\ceil{\frac{n}{\varepsilon}}^m$ different loads in total.
    Additionally, there are at most $m^{k + 1}$ assignments of jobs from $B$ to the machines.
    Thus, after trimming we are left with at most one schedule per unique machine load or assignment of jobs, so $|S(B)| \le \ceil{n / \varepsilon}^m m^{k + 1}$.
    Moreover, $|S'| \le m^{k + 1}$,
    thus $|S' \times S(B_1) \times S(B_2)| \le \ceil{n / \varepsilon}^{2 m} m^{3 k + 3}$.
    For each triple it takes $\Osymbol(k^2)$ time to check if the jobs are scheduled consistently, and $\Osymbol(m)$ time to merge the schedules and check their new total makespan.

    Since there are $\Osymbol(n)$ bags to process and $\Osymbol(\log(n p_{max}))$ iterations of binary search, it completes the proof.
\end{proof}

\section{Experimental results}
\label{sec:experiments}

\subsection{Setup}
In this section, we present the findings from the computational experiments conducted.
We ran the experiments on: Intel(R) Core(TM)2 Duo CPU, T8300 @ 2.40GHz; 2x2GB of DDR RAM (667 MT/s); OS: NAME="Linux~Mint" VERSION="21.3 (Virginia)"; SSD: Model=KINGSTON SA400S37120G.
The algorithms were implemented in C++ using \texttt{Boost Library} for solving maximum flow instances.
The code with example scripts is available at \url{https://bitbucket.org/tpikies/block-scheduling/} \cite{blockrepo}.
The testing consists of the following steps:

\newcommand{\ntasks}{n_{tasks}\xspace}
\newcommand{\nparts}{n_{parts}\xspace}
\newcommand{\avg}{\textrm{avg}\xspace}
\newcommand{\GREEDY}{\texttt{Greedy}\xspace}
\newcommand{\TREE}{\texttt{Tree-Width}\xspace}
\newcommand{\EXACTCMAX}{\texttt{Exact-Cmax}\xspace}
\newcommand{\FLOW}{\texttt{Flow}\xspace}
\newcommand{\TIMEOFCOMPUTATIONLABEL}{Time of computations [$\mu$s]}

\paragraph{Partition Generation}
First, we provide two values: $\ntasks$ (number of tasks $n$) and $\nparts$ (number of blocks),   to obtain the sizes of the blocks in the desired block graphs.
Precisely, we generate a random partition of $\ntasks + \nparts - 1$ into $\nparts$ parts.
We apply a modified version of the well-known uniform random integer partitioning algorithm from \cite[p. 72--76]{nijenhuis2014combinatorial}.
We modified the algorithm, to generate a partition into given number of parts, and to ensure that all the values in the partition are within arbitrary predefined bounds $2$ and $m$.
The number to be partitioned and the bounds were chosen to generate a collection of $n$ jobs, represented by connected graph, which can be scheduled on $m$ machines.
 
\paragraph{Block Graph Generation}
Due to our best knowledge, there does not exist a predefined, established, algorithm for random block graph generation.
Due to this, we propose the following method, which consists of the 2 phases.
First, the sizes of the parts are randomly shuffled; and a clique of the size given by the first part size is created.
After this, iteratively until all other part sizes are considered, the following steps are performed:
\begin{enumerate}
	\item A next part size $s$ in the shuffled partition is picked.
	\item One of existing blocks, denoted as $B$, is chosen, uniformly at random and one of its vertices, denoted as $v$, is chosen, also uniformly at random.
	\item A new block $B'$ consisting of $v$ and $s-1$ new vertices is constructed. 
	That is, we connect a new clique consisting of $s - 1$ new vertices to $v$, thus forming a block of size $s$.
\end{enumerate}
By this construction we obtain a (rooted) cut-tree structure representing the generated graphs.

\paragraph{Algorithm Application}
Finally, the respective scheduling algorithms are applied. 
In particular, we implemented: \Cref{alg:2apx_identical_block} (\GREEDY),  \Cref{algorithm:main-loop} (\EXACTCMAX), a version of \Cref{alg:arbitrary_tw} suitable for the identical machines (\TREE); and \Cref{alg:uniform_flow_network} (\FLOW).
In addition, we implemented the PTAS for $P|G=\blockgraph|\cmaxcost$ from \Cref{thm:ptas_identical_unit_block}, which effectively consists of the first three algorithms (\PTAS).

\paragraph{A note on interpretation of data in tables}
The algorithms were applied for instances of the random graph generated according to the parameters specified for a particular test.
In particular we performed the tests for $25$ instances of a graph with $n \in \{5, 10, \ldots, 50\}$ vertices, and for $m \in \{4,6,8\}$ machines; unless stated otherwise.
In the following tables, the letter 'T' signifies that some timeouts were present -- tests crossed the $15$ minutes limit.
Similarly, the letter 'M' signifies that the program was terminated due to excessive memory usage.
We considered $3$ different functions providing the number of blocks for given $n$, namely:
$b_{\min} = \ceil{\frac{n-1}{m-1}}$, $b_{\max} = n-1$, $b_{\avg} = \floor{(b_{\min} + b_{\max})/2}$.
We use such a formula for $b_{\min}$, due to the fact that a first block in block graph consists of up to $m$ vertices, and any other block can consists of up to $m-1$ vertices not considered so far; obviously it holds when $n \ge m$.
For relevant algorithms, we consider different processing times for jobs.
Namely, for functions $p_0, p_1, p_2$, processing times assigned are random integers from the ranges $[1..5]$, $[1..10]$, $[1..20]$, respectively.
All computation times are given in $\mu$s.

\subsection{Empirical results for \Cref{alg:2apx_identical_block} -- \GREEDY}

In this subsection we assess the approximation ratio of \Cref{thm:2apx_identical_block} for an average instance of $P|G=\blockgraph|\cmaxcost$.
\GREEDY was applied for $m$ identical machines, where $m\in \{4,6,8\}$ and a block graph consisting of $n$ vertices, where $n \in \{50, \ldots, 300\}$. 

The result, $\cmaxcost$ of the constructed schedule, is then compared with a simple lower bound on the optimal $\cmaxcost$, i.e. with $(\sum p_i)/m$.
The averages of those ratios are presented in the table.
Due to the fact that the ratio tends to $1.0$, when $n$ increases, in addition we analyze the case when $n \le 40$ and the jobs are unit time in more detailed way.
Precisely, for each of the generated graphs we construct a schedule by \GREEDY as previously, but we compare the length of the schedule with optimum length obtained by an application of \TREE; the averages of those ratios are also presented in the tables.
As we can see, in some cases the difference between the lower bound and real value of optimal $\cmaxcost$ was noticeable, but not great.
\begin{table}[H]
        \newcolumntype{H}{>{\setbox0=\hbox\bgroup}c<{\egroup}@{}}
        \newcolumntype{F}{cHHHHH}
	\footnotesize
	\centering
	\newcommand{\APXGREEDY}{$\cmaxcost(S_G)/(\sum p_j/m)$}
	\begin{tabular}{|c|c|F|F|F|F|F|F|F|F|F|}
		\hline
		\multicolumn{56}{|c|}{\APXGREEDY} \\
		\hline
		\hline
		&   & \multicolumn{18}{c|}{$b_{\min}$}&\multicolumn{18}{c|}{$b_{\avg}$}&\multicolumn{18}{c|}{$b_{\max}$} \\
		$m$ & $n$  & \multicolumn{6}{c|}{$p_{0}$}&\multicolumn{6}{c|}{$p_{1}$}&\multicolumn{6}{c|}{$p_{2}$}&\multicolumn{6}{c|}{$p_{0}$}&\multicolumn{6}{c|}{$p_{1}$}&\multicolumn{6}{c|}{$p_{2}$}&\multicolumn{6}{c|}{$p_{0}$}&\multicolumn{6}{c|}{$p_{1}$}&\multicolumn{6}{c|}{$p_{2}$} \\

		\hline4 &   50& 1.06 &   41.1 &   38.9 &   14.7 & 100.00 & GREEDY& 1.05 &   40.6 &   38.8 &   19.7 & 100.00 & GREEDY& 1.06 &   41.3 &   39.0 &   22.5 & 100.00 & GREEDY& 1.06 &   75.7 &   71.6 &   15.4 & 100.00 & GREEDY& 1.05 &   74.8 &   70.8 &   19.8 & 100.00 & GREEDY& 1.06 &   75.8 &   71.6 &   22.2 & 100.00 & GREEDY& 1.06 &  145.1 &  136.4 &   15.2 & 100.00 & GREEDY& 1.07 &  144.3 &  135.2 &   19.9 & 100.00 & GREEDY& 1.07 &  145.5 &  136.2 &   22.5 & 100.00 & GREEDY\\
		&  100& 1.03 &   79.2 &   76.9 &   26.3 & 100.00 & GREEDY& 1.03 &   78.6 &   76.5 &   35.4 & 100.00 & GREEDY& 1.03 &   76.5 &   74.4 &   43.2 & 100.00 & GREEDY& 1.04 &  145.7 &  140.8 &   26.3 & 100.00 & GREEDY& 1.03 &  144.2 &  140.0 &   37.1 & 100.00 & GREEDY& 1.03 &  140.3 &  136.0 &   45.6 & 100.00 & GREEDY& 1.04 &  278.2 &  268.7 &   29.7 & 100.00 & GREEDY& 1.04 &  276.9 &  267.6 &   35.4 & 100.00 & GREEDY& 1.03 &  267.3 &  259.1 &   42.5 & 100.00 & GREEDY\\
		&  150& 1.02 &  116.5 &  114.0 &   61.4 & 100.00 & GREEDY& 1.02 &  117.0 &  114.9 &   37.0 & 100.00 & GREEDY& 1.02 &  115.3 &  112.9 &   54.5 & 100.00 & GREEDY& 1.02 &  213.5 &  208.8 &   61.8 & 100.00 & GREEDY& 1.02 &  216.0 &  211.2 &   37.2 & 100.00 & GREEDY& 1.02 &  211.6 &  206.8 &   51.6 & 100.00 & GREEDY& 1.03 &  408.2 &  398.1 &   61.0 & 100.00 & GREEDY& 1.02 &  413.2 &  403.2 &   37.9 & 100.00 & GREEDY& 1.03 &  404.4 &  394.2 &   51.8 & 100.00 & GREEDY\\
		&  200& 1.02 &  155.0 &  152.3 &   68.8 & 100.00 & GREEDY& 1.02 &  153.2 &  150.6 &   80.6 & 100.00 & GREEDY& 1.02 &  153.5 &  150.9 &   48.4 & 100.00 & GREEDY& 1.02 &  285.0 &  279.5 &   73.4 & 100.00 & GREEDY& 1.02 &  281.6 &  276.3 &   81.1 & 100.00 & GREEDY& 1.02 &  281.8 &  276.5 &   48.6 & 100.00 & GREEDY& 1.02 &  543.8 &  533.3 &   68.2 & 100.00 & GREEDY& 1.02 &  537.7 &  526.7 &   80.4 & 100.00 & GREEDY& 1.02 &  538.4 &  527.3 &   51.0 & 100.00 & GREEDY\\
		&  250& 1.01 &  189.9 &  187.7 &   85.6 & 100.00 & GREEDY& 1.01 &  191.2 &  188.9 &  100.2 & 100.00 & GREEDY& 1.01 &  191.4 &  189.1 &   60.8 & 100.00 & GREEDY& 1.01 &  348.4 &  343.7 &   85.8 & 100.00 & GREEDY& 1.01 &  351.4 &  346.3 &  102.5 & 100.00 & GREEDY& 1.01 &  351.3 &  346.5 &   60.0 & 100.00 & GREEDY& 1.01 &  664.1 &  654.8 &   87.4 & 100.00 & GREEDY& 1.02 &  671.6 &  661.0 &  101.5 & 100.00 & GREEDY& 1.01 &  669.0 &  660.9 &   60.8 & 100.00 & GREEDY\\
		&  300& 1.01 &  230.6 &  228.1 &   70.8 & 100.00 & GREEDY& 1.01 &  229.5 &  227.1 &  101.4 & 100.00 & GREEDY& 1.01 &  227.8 &  225.0 &  120.1 & 100.00 & GREEDY& 1.01 &  423.0 &  417.7 &   75.7 & 100.00 & GREEDY& 1.01 &  421.7 &  416.2 &  100.8 & 100.00 & GREEDY& 1.01 &  416.8 &  411.4 &  122.0 & 100.00 & GREEDY& 1.01 &  807.8 &  797.0 &   71.4 & 100.00 & GREEDY& 1.01 &  803.5 &  793.2 &  102.8 & 100.00 & GREEDY& 1.02 &  796.1 &  783.9 &  120.9 & 100.00 & GREEDY\\
		\hline6 &   50& 1.07 &   27.8 &   26.0 &   13.7 & 100.00 & GREEDY& 1.10 &   28.1 &   25.6 &   20.4 & 100.00 & GREEDY& 1.09 &   28.5 &   26.2 &   22.6 & 100.00 & GREEDY& 1.08 &   51.4 &   47.6 &   14.0 & 100.00 & GREEDY& 1.10 &   51.8 &   46.9 &   20.6 & 100.00 & GREEDY& 1.10 &   52.8 &   47.9 &   23.5 & 100.00 & GREEDY& 1.08 &   97.6 &   90.2 &   13.8 & 100.00 & GREEDY& 1.12 &  100.0 &   89.2 &   20.1 & 100.00 & GREEDY& 1.12 &  101.5 &   91.0 &   22.6 & 100.00 & GREEDY\\
		&  100& 1.04 &   53.3 &   51.5 &   23.6 & 100.00 & GREEDY& 1.05 &   53.1 &   50.7 &   35.7 & 100.00 & GREEDY& 1.04 &   51.9 &   49.7 &   42.3 & 100.00 & GREEDY& 1.04 &   98.0 &   94.5 &   24.2 & 100.00 & GREEDY& 1.05 &   98.0 &   93.0 &   36.0 & 100.00 & GREEDY& 1.05 &   95.4 &   90.9 &   42.4 & 100.00 & GREEDY& 1.04 &  187.2 &  179.8 &   24.3 & 100.00 & GREEDY& 1.06 &  187.0 &  177.1 &   35.8 & 100.00 & GREEDY& 1.05 &  182.0 &  173.0 &   42.3 & 100.00 & GREEDY\\
		&  150& 1.03 &   78.6 &   76.1 &   62.5 & 100.00 & GREEDY& 1.03 &   78.2 &   76.2 &   33.8 & 100.00 & GREEDY& 1.04 &   78.2 &   75.5 &   52.2 & 100.00 & GREEDY& 1.03 &  144.1 &  139.4 &   62.1 & 100.00 & GREEDY& 1.03 &  144.0 &  139.9 &   34.4 & 100.00 & GREEDY& 1.04 &  144.0 &  138.3 &   53.9 & 100.00 & GREEDY& 1.04 &  275.4 &  265.5 &   63.0 & 100.00 & GREEDY& 1.03 &  274.7 &  266.6 &   34.2 & 100.00 & GREEDY& 1.04 &  274.4 &  263.4 &   52.6 & 100.00 & GREEDY\\
		&  200& 1.02 &  102.4 &  100.1 &   68.1 & 100.00 & GREEDY& 1.02 &  102.7 &  100.7 &   82.2 & 100.00 & GREEDY& 1.02 &  102.9 &  101.3 &   43.8 & 100.00 & GREEDY& 1.03 &  188.3 &  182.9 &   68.8 & 100.00 & GREEDY& 1.03 &  189.4 &  184.4 &   82.8 & 100.00 & GREEDY& 1.02 &  189.8 &  186.0 &   44.6 & 100.00 & GREEDY& 1.03 &  358.7 &  348.5 &   69.9 & 100.00 & GREEDY& 1.03 &  361.3 &  351.4 &   85.3 & 100.00 & GREEDY& 1.02 &  361.4 &  354.2 &   45.2 & 100.00 & GREEDY\\
		&  250& 1.02 &  128.6 &  126.4 &   84.7 & 100.00 & GREEDY& 1.02 &  128.3 &  126.1 &  102.9 & 100.00 & GREEDY& 1.01 &  128.0 &  126.2 &   54.2 & 100.00 & GREEDY& 1.02 &  237.3 &  231.7 &   85.0 & 100.00 & GREEDY& 1.02 &  236.0 &  231.1 &  103.1 & 100.00 & GREEDY& 1.02 &  235.4 &  231.6 &   55.3 & 100.00 & GREEDY& 1.03 &  453.0 &  441.6 &   85.2 & 100.00 & GREEDY& 1.02 &  450.4 &  440.8 &  104.8 & 100.00 & GREEDY& 1.02 &  449.3 &  441.8 &   54.6 & 100.00 & GREEDY\\
		&  300& 1.02 &  154.1 &  151.7 &  102.3 & 100.00 & GREEDY& 1.02 &  152.4 &  150.1 &  122.8 & 100.00 & GREEDY& 1.01 &  152.4 &  150.7 &   64.1 & 100.00 & GREEDY& 1.02 &  282.6 &  278.0 &  101.8 & 100.00 & GREEDY& 1.02 &  279.5 &  274.4 &  123.6 & 100.00 & GREEDY& 1.02 &  280.5 &  276.1 &   68.2 & 100.00 & GREEDY& 1.02 &  540.4 &  530.2 &  102.7 & 100.00 & GREEDY& 1.02 &  533.4 &  522.9 &  125.4 & 100.00 & GREEDY& 1.01 &  533.9 &  526.2 &   67.6 & 100.00 & GREEDY\\
		\hline8 &   50& 1.14 &   22.2 &   19.4 &   20.3 & 100.00 & GREEDY& 1.14 &   22.6 &   19.8 &   23.7 & 100.00 & GREEDY& 1.07 &   21.1 &   19.7 &   13.4 & 100.00 & GREEDY& 1.16 &   41.0 &   35.3 &   20.6 & 100.00 & GREEDY& 1.17 &   42.0 &   36.0 &   24.2 & 100.00 & GREEDY& 1.09 &   39.3 &   36.0 &   13.4 & 100.00 & GREEDY& 1.18 &   79.1 &   67.0 &   20.7 & 100.00 & GREEDY& 1.16 &   79.5 &   68.4 &   23.6 & 100.00 & GREEDY& 1.12 &   76.1 &   68.2 &   13.5 & 100.00 & GREEDY\\
		&  100& 1.05 &   40.2 &   38.4 &   24.4 & 100.00 & GREEDY& 1.07 &   40.7 &   38.2 &   36.7 & 100.00 & GREEDY& 1.06 &   39.8 &   37.4 &   44.0 & 100.00 & GREEDY& 1.05 &   74.2 &   70.4 &   24.6 & 100.00 & GREEDY& 1.09 &   75.8 &   69.7 &   37.4 & 100.00 & GREEDY& 1.07 &   73.3 &   68.2 &   44.2 & 100.00 & GREEDY& 1.06 &  142.4 &  134.2 &   24.8 & 100.00 & GREEDY& 1.09 &  144.8 &  132.7 &   37.8 & 100.00 & GREEDY& 1.08 &  140.8 &  129.8 &   44.4 & 100.00 & GREEDY\\
		&  150& 1.04 &   59.7 &   57.2 &   71.1 & 100.00 & GREEDY& 1.03 &   58.8 &   57.3 &   35.8 & 100.00 & GREEDY& 1.05 &   59.6 &   56.9 &   53.6 & 100.00 & GREEDY& 1.06 &  110.9 &  104.7 &   64.9 & 100.00 & GREEDY& 1.04 &  109.1 &  105.0 &   34.6 & 100.00 & GREEDY& 1.05 &  108.9 &  104.1 &   53.4 & 100.00 & GREEDY& 1.06 &  211.6 &  199.3 &   64.6 & 100.00 & GREEDY& 1.04 &  207.8 &  199.9 &   35.6 & 100.00 & GREEDY& 1.04 &  206.9 &  198.2 &   53.7 & 100.00 & GREEDY\\
		&  200& 1.04 &   78.6 &   75.9 &   70.0 & 100.00 & GREEDY& 1.04 &   78.3 &   75.6 &   84.6 & 100.00 & GREEDY& 1.02 &   78.2 &   76.7 &   43.8 & 100.00 & GREEDY& 1.04 &  143.9 &  138.7 &   71.6 & 100.00 & GREEDY& 1.04 &  144.0 &  138.4 &   85.7 & 100.00 & GREEDY& 1.02 &  143.8 &  140.4 &   44.4 & 100.00 & GREEDY& 1.04 &  275.8 &  264.6 &   71.2 & 100.00 & GREEDY& 1.05 &  276.2 &  263.6 &   86.8 & 100.00 & GREEDY& 1.03 &  275.1 &  267.6 &   44.6 & 100.00 & GREEDY\\
		&  250& 1.03 &   96.5 &   94.1 &   86.9 & 100.00 & GREEDY& 1.02 &   96.8 &   94.7 &  107.4 & 100.00 & GREEDY& 1.02 &   96.8 &   95.0 &   54.6 & 100.00 & GREEDY& 1.03 &  177.6 &  171.8 &   87.7 & 100.00 & GREEDY& 1.03 &  178.5 &  173.4 &  108.0 & 100.00 & GREEDY& 1.02 &  177.4 &  174.0 &   56.4 & 100.00 & GREEDY& 1.03 &  338.4 &  327.4 &   89.6 & 100.00 & GREEDY& 1.03 &  342.0 &  330.8 &  106.6 & 100.00 & GREEDY& 1.03 &  340.1 &  331.6 &   55.1 & 100.00 & GREEDY\\
		&  300& 1.03 &  116.9 &  113.7 &  105.3 & 100.00 & GREEDY& 1.02 &  115.1 &  112.8 &  128.1 & 100.00 & GREEDY& 1.01 &  114.9 &  113.3 &   63.0 & 100.00 & GREEDY& 1.03 &  214.5 &  208.2 &  105.5 & 100.00 & GREEDY& 1.03 &  211.7 &  206.0 &  128.2 & 100.00 & GREEDY& 1.02 &  211.6 &  207.5 &   64.8 & 100.00 & GREEDY& 1.03 &  409.4 &  396.9 &  105.5 & 100.00 & GREEDY& 1.03 &  402.8 &  392.2 &  130.5 & 100.00 & GREEDY& 1.02 &  403.8 &  395.6 &   65.8 & 100.00 & GREEDY\\
		\hline
	\end{tabular}
	\caption{
		The ratio of $\cmaxcost$ of schedules produced by \GREEDY and $\frac{\sum p_j}{m}$ (a lower bound on the $\cmaxcost$ of an optimum schedule).
		As previously stated, for $p_0, p_1, p_2$, the processing times are drawn uniformly at random from the sets $[1,5]$, $[1,10]$, $[1,20]$, respectively.
	 }
	\label{tab:greedy}
\end{table}
\begin{table}[H]
        \newcolumntype{H}{>{\setbox0=\hbox\bgroup}c<{\egroup}@{}}
        \newcolumntype{F}{cHHHHH}
	\footnotesize
	\centering
	\begin{tabular}{|c|c|F|F|F|}
		\hline
		\multicolumn{20}{|c|}{$\cmaxcost(S_G) / \cmaxcost(S_T)$} \\
		\hline
		\hline
		$m$ & $n$  & \multicolumn{6}{c|}{$b_{\min}$}&\multicolumn{6}{c|}{$b_{\avg}$}&\multicolumn{6}{c|}{$b_{\max}$} \\
		\hline 
		4 &    5& 1.00 &   2.00 &   1.00 &   4.24 & 50.00 & GREEDYvsTREE& 1.00 &   2.00 &   1.00 &   4.20 & 50.00 & GREEDYvsTREE& 1.00 &   2.00 &   1.00 &   4.28 & 50.00 & GREEDYvsTREE\\
		&   10& 1.00 &   3.00 &   1.00 &   5.24 & 50.00 & GREEDYvsTREE& 1.00 &   3.00 &   1.00 &   6.16 & 50.00 & GREEDYvsTREE& 1.00 &   3.00 &   1.00 &   6.28 & 50.00 & GREEDYvsTREE\\
		&   15& 1.09 &   4.56 &   1.14 &   6.24 & 50.00 & GREEDYvsTREE& 1.03 &   4.12 &   1.03 &   8.04 & 50.00 & GREEDYvsTREE& 1.01 &   4.04 &   1.01 &   8.44 & 50.00 & GREEDYvsTREE\\
		&   20& 1.03 &   5.92 &   1.18 &   7.68 & 50.00 & GREEDYvsTREE& 1.10 &   5.52 &   1.10 &   9.72 & 50.00 & GREEDYvsTREE& 1.08 &   5.40 &   1.08 &  10.64 & 50.00 & GREEDYvsTREE\\
		&   25& 1.08 &   7.56 &   1.08 &   8.28 & 50.00 & GREEDYvsTREE& 1.00 &   7.04 &   1.01 &  10.80 & 50.00 & GREEDYvsTREE& 1.00 &   7.00 &   1.00 &  12.08 & 50.00 & GREEDYvsTREE\\
		&   30& 1.08 &   8.68 &   1.09 &   9.68 & 50.00 & GREEDYvsTREE& 1.01 &   8.08 &   1.01 &  12.68 & 50.00 & GREEDYvsTREE& 1.00 &   8.00 &   1.00 &  13.80 & 50.00 & GREEDYvsTREE\\
		&   35& 1.08 &  10.00 &   1.11 &  10.92 & 50.00 & GREEDYvsTREE& 1.05 &   9.44 &   1.05 &  14.16 & 50.00 & GREEDYvsTREE& 1.01 &   9.12 &   1.01 &  15.88 & 50.00 & GREEDYvsTREE\\
		&   40& 1.12 &  11.68 &   1.17 &  11.92 & 50.00 & GREEDYvsTREE& 1.05 &  10.52 &   1.05 &  15.68 & 50.00 & GREEDYvsTREE& 1.04 &  10.40 &   1.04 &  17.72 & 50.00 & GREEDYvsTREE\\
		\hline 
		6 &    5& 1.00 &   1.00 &   1.00 &   3.48 & 50.00 & GREEDYvsTREE& 1.00 &   1.00 &   1.00 &   4.08 & 50.00 & GREEDYvsTREE& 1.00 &   1.00 &   1.00 &   4.24 & 50.00 & GREEDYvsTREE\\
		&   10& 1.00 &   2.00 &   1.00 &   5.12 & 50.00 & GREEDYvsTREE& 1.00 &   2.00 &   1.00 &   6.32 & 50.00 & GREEDYvsTREE& 1.00 &   2.00 &   1.00 &   6.28 & 50.00 & GREEDYvsTREE\\
		&   15& 1.00 &   3.00 &   1.00 &   6.20 & 50.00 & GREEDYvsTREE& 1.00 &   3.00 &   1.00 &   8.28 & 50.00 & GREEDYvsTREE& 1.00 &   3.00 &   1.00 &   8.72 & 50.00 & GREEDYvsTREE\\
		&   20& 1.00 &   4.00 &   1.00 &   7.08 & 50.00 & GREEDYvsTREE& 1.00 &   4.00 &   1.00 &  13.52 & 50.00 & GREEDYvsTREE& 1.00 &   4.00 &   1.00 &  10.80 & 50.00 & GREEDYvsTREE\\
		&   25& 1.00 &   5.00 &   1.00 &   8.36 & 50.00 & GREEDYvsTREE& 1.00 &   5.00 &   1.00 &  10.76 & 50.00 & GREEDYvsTREE& 1.00 &   5.00 &   1.00 &  12.92 & 50.00 & GREEDYvsTREE\\
		&   30& 1.07 &   5.92 &   1.18 &   8.60 & 50.00 & GREEDYvsTREE& 1.10 &   5.52 &   1.10 &  12.56 & 50.00 & GREEDYvsTREE& 1.08 &   5.40 &   1.08 &  14.16 & 50.00 & GREEDYvsTREE\\
		&   35& 1.08 &   6.92 &   1.15 &   9.56 & 50.00 & GREEDYvsTREE&   T &   6.08 &   1.01 &  14.36 &  50.00 & GREEDYvsTREE& 1.00 &   6.00 &   1.00 &  15.92 & 50.00 & GREEDYvsTREE\\
		&   40& 1.10 &   7.80 &   1.11 &  10.60 & 50.00 & GREEDYvsTREE&   T &      T &   1.00 &   15.8 &  50.00 & GREEDYvsTREE&   T &      7 &   1.00 &  17.96 &  50.00 & GREEDYvsTREE\\
		\hline
	\end{tabular}
	\caption{
		The ratio of $\cmaxcost$ of a schedule produced by \GREEDY, denoted by $S_G$, and a schedule produced by \TREE, denoted by $S_T$ (an optimal schedule).
		Here we considered unit-time jobs.
		The letter 'T' denotes timeout, as stated earlier.
	} 
	\label{tab:greed_vs_treewidth}
\end{table}
We can infer from \Cref{tab:greedy} that the schedules produced by \GREEDY are not very far from optimal ones.
In many cases, the solution is much better than guaranteed by the worst-case ratio given in \Cref{thm:2apx_identical_block} or even in \Cref{lem:equal_coloring_block}.
That is, twice the value of the optimal makespan, or $\sim\frac{m}{m-1}$ in the latter case.
It seems that the probability of generating a hard instance tends to $0$ as $n$ increases, for all $p$ and all $b$ functions.

To have a better view for small instances, because it seems that the approximation ratio is worst in these cases, we ran the tests for small data and compared the $\cmaxcost$ of the schedules with the optimal ones, obtained using treewidth.
Here we considered unit-time jobs only, due to time limitations.
It seems that for small instances, the algorithm is also working well in fact, cf. \Cref{tab:greed_vs_treewidth}.
As we can see, it is the case of $n \le 50$, the approximation ratio is getting worse, when $n$ grows; however, at some moment it stops getting worse.

\subsection{Empirical results on \EXACTCMAX}
\Cref{tab:exact_cmax_increasing_cmax} and \Cref{fig:different_b_functions_for_cmax} illustrate $2$ things.
As expected, even small increase of the $\cmaxcost$ bound has a radical impact on the processing time -- the algorithm is exponential in $\cmaxcost$ bound.
The second thing is that the real processing time depends one the average size of a block, which is evident especially for $b=b_{\min}$.
If size of a block is equal to $m$, then, in fact, the problem becomes a bit easier: on average, a given subgraph has few distinct colorings possible.
This means that the exponential (with respect to $\cmaxcost$) number of distinct colorings may appear in the processing of given block-cut tree later (higher in the tree).
This is well illustrated in the left figure of \Cref{fig:different_b_functions_for_cmax}: when the average size of a block tends to $m$, for given number of blocks, the processing time decreases. 

\begin{table}[H]
	\setlength{\tabcolsep}{1pt}
	\newcolumntype{H}{>{\setbox0=\hbox\bgroup}c<{\egroup}@{}}
	\newcolumntype{F}{HHHHrHH}
	\footnotesize
	\centering
	\begin{tabular}{|c|c|c|F|F|F|}
		\hline
		\multicolumn{24}{|c|}{\TIMEOFCOMPUTATIONLABEL} \\
		\hline
		\hline
		$\cmaxcost$ & $m$   & $n$                    & \multicolumn{7}{c|}{$b_{\min}$}&\multicolumn{7}{c|}{$b_{\avg}$}&\multicolumn{7}{c|}{$b_{\max}$} \\
		\hline 
		2 & 4 &   25&   8 &  -0.14 &  -1.00 &   7.00 & 775.88 & 100.00  & EXACT-CMAX&  16 &  -0.14 &  -1.00 &   7.00 & 829.32 & 100.00  & EXACT-CMAX&  24 &  -0.14 &  -1.00 &   7.00 & 720.52 & 100.00  & EXACT-CMAX\\
		&   &   50&  17 &  -0.08 &  -1.00 &  13.00 & 1658.20 & 100.00  & EXACT-CMAX&  33 &  -0.08 &  -1.00 &  13.00 & 1693.44 & 100.00  & EXACT-CMAX&  49 &  -0.08 &  -1.00 &  13.00 & 1352.00 & 100.00  & EXACT-CMAX\\
		& 8 &   25&   4 &  -0.25 &  -1.00 &   4.00 & 6440.84 & 100.00  & EXACT-CMAX&  14 &  -0.25 &  -1.00 &   4.00 & 17078.72 & 100.00  & EXACT-CMAX&  24 &  -0.25 &  -1.00 &   4.00 & 6643.32 & 100.00  & EXACT-CMAX\\
		&   &   50&   7 &  -0.14 &  -1.00 &   7.00 & 14260.80 & 100.00  & EXACT-CMAX&  28 &  -0.14 &  -1.00 &   7.00 & 34088.00 & 100.00  & EXACT-CMAX&  49 &  -0.14 &  -1.00 &   7.00 & 17326.44 & 100.00  & EXACT-CMAX\\
		\hline 
		4 & 4 &   25&   8 &  -0.14 &  -1.00 &   7.00 & 1571.00 & 100.00  & EXACT-CMAX&  16 &  -0.14 &  -1.00 &   7.00 & 3987.20 & 100.00  & EXACT-CMAX&  24 &  -0.14 &  -1.00 &   7.00 & 6486.32 & 100.00  & EXACT-CMAX\\
		&   &   50&  17 &  -0.08 &  -1.00 &  13.00 & 2692.64 & 100.00  & EXACT-CMAX&  33 &  -0.08 &  -1.00 &  13.00 & 6146.76 & 100.00  & EXACT-CMAX&  49 &  -0.08 &  -1.00 &  13.00 & 9070.04 & 100.00  & EXACT-CMAX\\
		& 8 &   25&   4 &  1.00 &   4.00 &   4.00 & 99543.92 & 100.00  & EXACT-CMAX&  14 &  1.00 &   4.00 &   4.00 & 1052715.20 & 100.00  & EXACT-CMAX&  24 &  1.00 &   4.00 &   4.00 & 1044746.96 & 100.00  & EXACT-CMAX\\
		&   &   50&   7 &  -0.14 &  -1.00 &   7.00 & 21636.00 & 100.00  & EXACT-CMAX&  28 &  -0.14 &  -1.00 &   7.00 & 1346302.60 & 100.00  & EXACT-CMAX&  49 &  -0.14 &  -1.00 &   7.00 & 1230309.36 & 100.00  & EXACT-CMAX\\
		\hline 
		6 & 4 &   25&   8 &  -0.14 &  -1.00 &   7.00 & 1869.56 & 100.00  & EXACT-CMAX&  16 &  -0.14 &  -1.00 &   7.00 & 11790.44 & 100.00  & EXACT-CMAX&  24 &  -0.14 &  -1.00 &   7.00 & 33542.88 & 100.00  & EXACT-CMAX\\
		&   &   50&  17 &  -0.08 &  -1.00 &  13.00 & 3268.76 & 100.00  & EXACT-CMAX&  33 &  -0.08 &  -1.00 &  13.00 & 17046.20 & 100.00  & EXACT-CMAX&  49 &  -0.08 &  -1.00 &  13.00 & 44106.60 & 100.00  & EXACT-CMAX\\
		& 8 &   25&   4 &  1.00 &   4.00 &   4.00 & 102171.76 & 100.00  & EXACT-CMAX&  14 &  1.00 &   4.00 &   4.00 & 8456828.20 & 100.00  & EXACT-CMAX&  24 &  1.00 &   4.00 &   4.00 & 13720041.88 & 100.00  & EXACT-CMAX\\
		&   &   50&   7 &  -0.14 &  -1.00 &   7.00 & 50216.72 & 100.00  & EXACT-CMAX&  28   &   M   &      M   &      M   &      M   &      M &       &  49   &   M   &      M   &      M   &      M   &      M &       \\
		\hline 
		8 & 4 &   25&   8 &  1.00 &   7.00 &   7.00 & 4198.20 & 100.00  & EXACT-CMAX&  16 &  1.01 &   7.04 &   7.00 & 42511.40 & 100.00  & EXACT-CMAX&  24 &  1.00 &   7.00 &   7.00 & 139879.12 & 100.00  & EXACT-CMAX\\
		&   &   50&  17 &  -0.08 &  -1.00 &  13.00 & 6786.84 & 100.00  & EXACT-CMAX&  33 &  -0.08 &  -1.00 &  13.00 & 57931.56 & 100.00  & EXACT-CMAX&  49 &  -0.08 &  -1.00 &  13.00 & 200516.92 & 100.00  & EXACT-CMAX\\
		& 8 &   25&   4 &  1.00 &   4.00 &   4.00 & 119770.36 & 100.00  & EXACT-CMAX&  14 &  1.00 &   4.00 &   4.00 & 19407589.32 & 100.00  & EXACT-CMAX&  24   &   M   &      M   &      M   &      M   &      M &       \\
		&   &   50&   7 &  1.00 &   7.00 &   7.00 & 104718.60 & 100.00  & EXACT-CMAX&  28   &   M   &      M   &      M   &      M   &      M &       &  49   &   M   &      M   &      M   &      M   &      M &       \\
		\hline 
		10 & 4 &   25&   8 &  1.00 &   7.00 &   7.00 & 4601.08 & 100.00  & EXACT-CMAX&  16 &  1.01 &   7.04 &   7.00 & 77086.52 & 100.00  & EXACT-CMAX&  24 &  1.00 &   7.00 &   7.00 & 316493.96 & 100.00  & EXACT-CMAX\\
		&   &   50&  17 &  -0.08 &  -1.00 &  13.00 & 11749.00 & 100.00  & EXACT-CMAX&  33 &  -0.08 &  -1.00 &  13.00 & 197852.48 & 100.00  & EXACT-CMAX&  49 &  -0.08 &  -1.00 &  13.00 & 805403.16 & 100.00  & EXACT-CMAX\\
		& 8 &   25&   4 &  1.00 &   4.00 &   4.00 & 125151.48 & 100.00  & EXACT-CMAX&  14 &  1.00 &   4.00 &   4.00 & 26388878.44 & 100.00  & EXACT-CMAX&  24   &   M   &      M   &      M   &      M   &      M &       \\
		&   &   50&   7 &  1.00 &   7.00 &   7.00 & 115484.52 & 100.00  & EXACT-CMAX&  28   &   M   &      M   &      M   &      M   &      M &       &  49   &   M   &      M   &      M   &      M   &      M &       \\
		\hline
	\end{tabular}
	\caption{
		Processing times for \EXACTCMAX.
		As stated earlier, the letter 'M' denotes cancellation of the program due to lack of the memory.
}
	\label{tab:exact_cmax_increasing_cmax}
\end{table}

\begin{figure}[H]
	\centering
	\begin{subfigure}{0.49\textwidth}
		\centering
		\includegraphics[width=\textwidth]{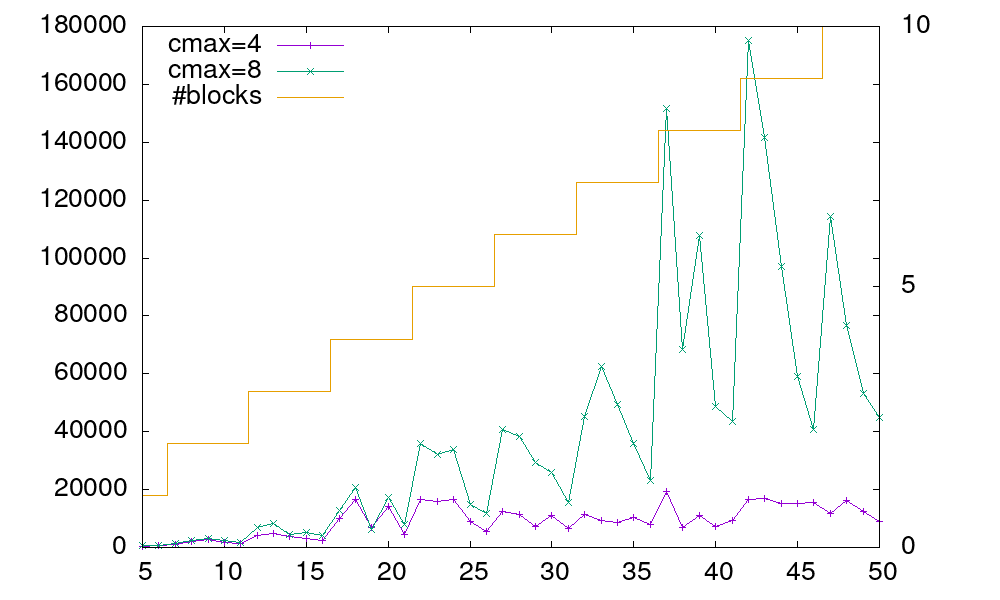}
	\end{subfigure}
	\begin{subfigure}{0.49\textwidth}
		\centering
\includegraphics[width=\textwidth]{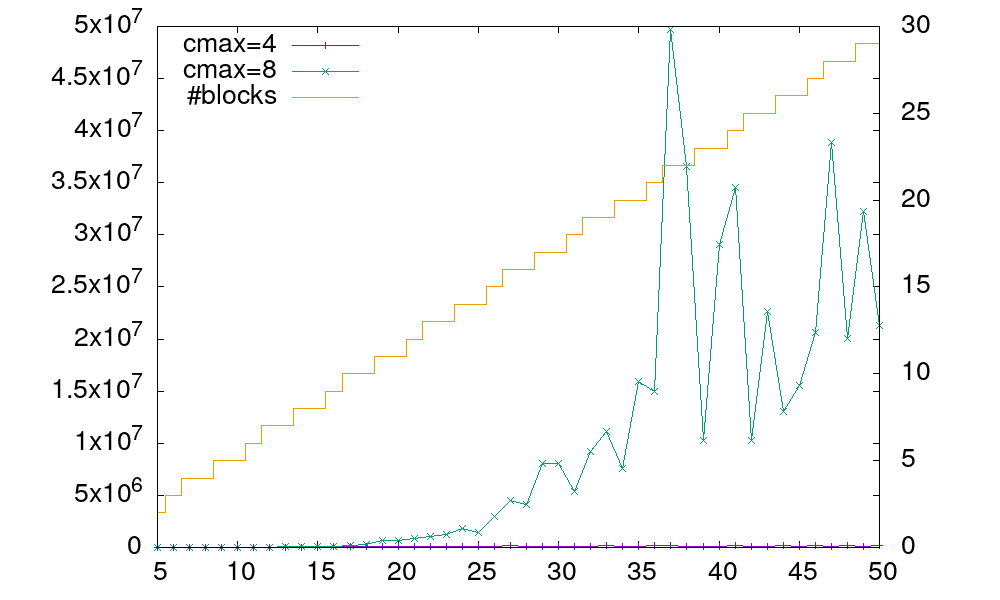}
\end{subfigure}
\caption{
	Processing time [$\mu$s] in the number of vertices for different $b$ functions: $\#blocks = b_{\min}$ and $\#blocks = b_{\max}$, respectively.
	In addition, the yellow lines show the number of blocks.
	For $b_{\min}$, the slopes are clearly correlated with changes of the block number.
}
\label{fig:different_b_functions_for_cmax}
\end{figure}

\subsection{Empirical results on \TREE}
When testing this algorithm we have observed that, in quite a few cases, the algorithm could not be applied at all due to the high memory requirements.
In particular for even $m \ge 6$ machines and $n \ge 50$ jobs, the program started to consume more than $4$GB of memory, the total amount that we assumed as a limit.

One idea to make it applicable to a larger data, is to better exploit the simplicity of the structure of block graphs and to prepare an \FPTAS specially tailored to this class.
More precisely, to exploit the fact that for identical machines in many places we do not have to carry about loads of particular machines, only about loads at all. 
This means that one could apply \EXACTCMAX to this case as well, provided that its complexity is upper bounded by the number of distinct colorings that appear during the processing.
Obviously, the maximum cardinality of a color class has size potentially up to $n$ (a good heuristic is to use $\cmaxcost$ produced by \GREEDY here instead).
However, the set of distinct colorings would be mostly empty, the total number of vectors still would be $\Osymbol(n^{poly(m)})$--polynomial one, due to the bound on the number of machines.
Hence, if one prepare data structures in a way that they are sensitive to the number of different coloring that are stored, it should be possible to apply the algorithms as well.
We have tested the proposed application of \EXACTCMAX, together with the heuristic approach, due to the fact that it is simply equivalent to running \EXACTCMAX with the bound on $\cmaxcost$ given by greedy algorithm.
As we can see in \Cref{tab:tree_vs_exactcmax1} and \Cref{tab:tree_vs_exactcmax2}, for small number of machines the savings on identifying colorings were not justified. 
This is reasonable, for small number of machines, the number of possibly equivalent (for a block graph) loads is small.
However, when the number of machines grows, also the number of equivalent loads of the machines grows, which means that the savings will grow as well.
This proves, that an approach tailored to the studied structure is more efficient, than a general approach.
\begin{table}[H]
	\newcolumntype{H}{>{\setbox0=\hbox\bgroup}c<{\egroup}@{}}
	\newcolumntype{F}{HHHrHH}
	\footnotesize
	\centering
	\begin{minipage}[t]{.5\linewidth}
		\centering
		\begin{tabular}{|c|c|F|F|F|}
			\hline
			\multicolumn{20}{|c|}{\TIMEOFCOMPUTATIONLABEL} \\
			\hline
			\hline
			$m$ & $n$  & \multicolumn{6}{c|}{$b_{\min}$}&\multicolumn{6}{c|}{$b_{\avg}$}&\multicolumn{6}{c|}{$b_{\max}$} \\
			\hline 
			4 &    5& 1.00 &   2.00 &   2.00 &  60.88 & 100.00  &   TREE& 1.00 &   2.00 &   2.00 & 103.40 & 100.00  &   TREE& 1.00 &   2.00 &   2.00 & 156.76 & 100.00  &   TREE\\
			&   10& 1.00 &   3.00 &   3.00 &  83.48 & 100.00  &   TREE& 1.00 &   3.00 &   3.00 & 590.80 & 100.00  &   TREE& 1.00 &   3.00 &   3.00 & 1442.00 & 100.00  &   TREE\\
			&   15& 1.00 &   4.20 &   4.20 & 254.28 & 100.00  &   TREE& 1.00 &   4.00 &   4.00 & 1912.36 & 100.00  &   TREE& 1.00 &   4.00 &   4.00 & 6720.64 & 100.00  &   TREE\\
			&   20& 1.00 &   5.76 &   5.76 & 464.08 & 100.00  &   TREE& 1.00 &   5.04 &   5.04 & 5296.84 & 100.00  &   TREE& 1.00 &   5.00 &   5.00 & 24402.52 & 100.00  &   TREE\\
			&   25& 1.00 &   7.00 &   7.00 & 446.88 & 100.00  &   TREE& 1.00 &   7.04 &   7.04 & 17080.68 & 100.00  &   TREE& 1.00 &   7.00 &   7.00 & 63354.72 & 100.00  &   TREE\\
			&   30& 1.00 &   8.08 &   8.08 & 1053.60 & 100.00  &   TREE& 1.00 &   8.00 &   8.00 & 32224.08 & 100.00  &   TREE& 1.00 &   8.00 &   8.00 & 148294.36 & 100.00  &   TREE\\
			&   35& 1.00 &   9.24 &   9.24 & 1694.96 & 100.00  &   TREE& 1.00 &   9.00 &   9.00 & 88948.16 & 100.00  &   TREE& 1.00 &   9.00 &   9.00 & 274267.72 & 100.00  &   TREE\\
			&   40& 1.00 &  10.40 &  10.40 & 1886.16 & 100.00  &   TREE& 1.00 &  10.04 &  10.04 & 242250.04 & 100.00  &   TREE& 1.00 &  10.00 &  10.00 & 550778.88 & 100.00  &   TREE\\
			&   45& 1.00 &  12.00 &  12.00 & 6073.72 & 100.00  &   TREE& 1.00 &  12.00 &  12.00 & 194169.36 & 100.00  &   TREE& 1.00 &  12.00 &  12.00 & 947251.20 & 100.00  &   TREE\\
			&   50& 1.00 &  13.08 &  13.08 & 6464.20 & 100.00  &   TREE& 1.00 &  13.00 &  13.00 & 713830.92 & 100.00  &   TREE& 1.00 &  13.00 &  13.00 & 1907404.12 & 100.00  &   TREE\\
			\hline 
			6 &    5& 1.00 &   1.00 &   1.00 &  99.52 & 100.00  &   TREE& 1.00 &   1.00 &   1.00 & 473.00 & 100.00  &   TREE& 1.00 &   1.00 &   1.00 & 978.48 & 100.00  &   TREE\\
			&   10& 1.00 &   2.00 &   2.00 & 404.08 & 100.00  &   TREE& 1.00 &   2.00 &   2.00 & 8072.00 & 100.00  &   TREE& 1.00 &   2.00 &   2.00 & 31996.92 & 100.00  &   TREE\\
			&   15& 1.00 &   3.00 &   3.00 & 793.28 & 100.00  &   TREE& 1.00 &   3.00 &   3.00 & 133423.04 & 100.00  &   TREE& 1.00 &   3.00 &   3.00 & 370702.52 & 100.00  &   TREE\\
			&   20& 1.00 &   4.00 &   4.00 & 3481.00 & 100.00  &   TREE& 1.00 &   4.00 &   4.00 & 1044794.72 & 100.00  &   TREE& 1.00 &   4.00 &   4.00 & 3012753.80 & 100.00  &   TREE\\
			&   25& 1.00 &   5.00 &   5.00 & 5376.52 & 100.00  &   TREE& 1.00 &   5.00 &   5.00 & 2345882.48 & 100.00  &   TREE& 1.00 &   5.00 &   5.00 & 12808966.60 & 100.00  &   TREE\\
			&   30& 1.00 &   5.56 &   5.56 & 11500.52 & 100.00  &   TREE& 1.00 &   5.00 &   5.00 & 41399041.72 & 100.00  &   TREE& 1.00 &   5.00 &   5.00 & 65858571.92 & 100.00  &   TREE\\
			&   35& 1.00 &   6.44 &   6.44 & 11521.40 & 100.00  &   TREE&   ?   &      T   &      0   &      T   &  40.00 &   TREE& 1.00 &   6.00 &   6.00 & 113559549.88 & 100.00  &   TREE\\
			&   40& 1.00 &   7.08 &   7.08 & 18891.28 & 100.00  &   TREE&   ?   &      T   &      0   &      T   &      0 &   TREE&   ?   &      T   &      0   &      T   &  52.00 &   TREE\\
			&   45& 1.00 &   8.04 &   8.04 & 48510.52 & 100.00  &   TREE&   ?   &      T   &      0   &      T   &  64.00 &   TREE&   ?   &      T   &      0   &      T   &      0 &   TREE\\
			&   50& 1.00 &   9.00 &   9.00 & 74830.52 & 100.00  &   TREE&   ?   &      T   &      0   &      T   &  16.00 &   TREE&   ?   &      T   &      0   &      T   &      0 &   TREE\\
			\hline 
		\end{tabular}
	\caption{Results of applications of Algorithm \TREE}
		\label{tab:tree_vs_exactcmax1}
	\end{minipage}%
	\begin{minipage}[t]{.5\linewidth}
		\centering
		\begin{tabular}{|c|c|F|F|F|}
			\hline
			\multicolumn{20}{|c|}{\TIMEOFCOMPUTATIONLABEL} \\
			\hline
			\hline
			$m$ & $n$ & \multicolumn{6}{c|}{$b_{\min}$}&\multicolumn{6}{c|}{$b_{\avg}$}&\multicolumn{6}{c|}{$b_{\max}$} \\
			\hline 
			4 &    5& 1.00 &    2.0 &    2.0 &  193.2 & 100.00 & EXACT-CMAX& 1.00 &    2.0 &    2.0 &  192.1 & 100.00 & EXACT-CMAX& 1.00 &    2.0 &    2.0 &  215.6 & 100.00 & EXACT-CMAX\\
			&   10& 1.00 &    3.0 &    3.0 &  527.8 & 100.00 & EXACT-CMAX& 1.00 &    3.0 &    3.0 &  844.3 & 100.00 & EXACT-CMAX& 1.00 &    3.0 &    3.0 & 1337.7 & 100.00 & EXACT-CMAX\\
			&   15& 1.08 &    4.3 &    4.0 & 1133.7 & 100.00 & EXACT-CMAX& 1.00 &    4.0 &    4.0 & 3217.5 & 100.00 & EXACT-CMAX& 1.00 &    4.0 &    4.0 & 4394.2 & 100.00 & EXACT-CMAX\\
			&   20& 1.16 &    5.8 &    5.0 & 1961.9 & 100.00 & EXACT-CMAX& 1.00 &    5.0 &    5.0 & 7771.8 & 100.00 & EXACT-CMAX& 1.00 &    5.0 &    5.0 & 21953.3 & 100.00 & EXACT-CMAX\\
			&   25& 1.00 &    7.0 &    7.0 & 2887.6 & 100.00 & EXACT-CMAX& 1.00 &    7.0 &    7.0 & 26968.2 & 100.00 & EXACT-CMAX& 1.00 &    7.0 &    7.0 & 69819.4 & 100.00 & EXACT-CMAX\\
			&   30& 1.01 &    8.1 &    8.0 & 6190.7 & 100.00 & EXACT-CMAX& 1.00 &    8.0 &    8.0 & 55994.3 & 100.00 & EXACT-CMAX& 1.00 &    8.0 &    8.0 & 139210.6 & 100.00 & EXACT-CMAX\\
			&   35& 1.02 &    9.2 &    9.0 & 15541.6 & 100.00 & EXACT-CMAX& 1.00 &    9.0 &    9.0 & 116169.9 & 100.00 & EXACT-CMAX& 1.00 &    9.0 &    9.0 & 350774.3 & 100.00 & EXACT-CMAX\\
			&   40& 1.06 &   10.6 &   10.0 & 10777.2 & 100.00 & EXACT-CMAX& 1.00 &   10.0 &   10.0 & 170780.5 & 100.00 & EXACT-CMAX& 1.00 &   10.0 &   10.0 & 716785.5 & 100.00 & EXACT-CMAX\\
			&   45& 1.00 &   12.0 &   12.0 & 21562.8 & 100.00 & EXACT-CMAX& 1.00 &   12.0 &   12.0 & 555952.9 & 100.00 & EXACT-CMAX& 1.00 &   12.0 &   12.0 & 1652327.9 & 100.00 & EXACT-CMAX\\
			&   50& 1.00 &   13.0 &   13.0 & 39325.7 & 100.00 & EXACT-CMAX& 1.00 &   13.0 &   13.0 & 939086.0 & 100.00 & EXACT-CMAX& 1.00 &   13.0 &   13.0 & 3682377.4 & 100.00 & EXACT-CMAX\\
			\hline 
			6 &    5& 1.00 &    1.0 &    1.0 &  321.3 & 100.00 & EXACT-CMAX& 1.00 &    1.0 &    1.0 &  261.3 & 100.00 & EXACT-CMAX& 1.00 &    1.0 &    1.0 &  115.4 & 100.00 & EXACT-CMAX\\
			&   10& 1.00 &    2.0 &    2.0 & 1296.1 & 100.00 & EXACT-CMAX& 1.00 &    2.0 &    2.0 & 1262.9 & 100.00 & EXACT-CMAX& 1.00 &    2.0 &    2.0 & 1647.2 & 100.00 & EXACT-CMAX\\
			&   15& 1.00 &    3.0 &    3.0 & 3158.7 & 100.00 & EXACT-CMAX& 1.00 &    3.0 &    3.0 & 21931.2 & 100.00 & EXACT-CMAX& 1.00 &    3.0 &    3.0 & 13463.4 & 100.00 & EXACT-CMAX\\
			&   20& 1.00 &    4.0 &    4.0 & 13237.8 & 100.00 & EXACT-CMAX& 1.00 &    4.0 &    4.0 & 72318.1 & 100.00 & EXACT-CMAX& 1.00 &    4.0 &    4.0 & 104876.6 & 100.00 & EXACT-CMAX\\
			&   25& 1.00 &    5.0 &    5.0 & 16034.5 & 100.00 & EXACT-CMAX& 1.00 &    5.0 &    5.0 & 222829.8 & 100.00 & EXACT-CMAX& 1.00 &    5.0 &    5.0 & 444028.8 & 100.00 & EXACT-CMAX\\
			&   30& 1.10 &    5.5 &    5.0 & 29541.5 & 100.00 & EXACT-CMAX& 1.00 &    5.0 &    5.0 & 1110095.0 & 100.00 & EXACT-CMAX& 1.00 &    5.0 &    5.0 & 1114089.7 & 100.00 & EXACT-CMAX\\
			&   35& 1.10 &    6.6 &    6.0 & 37783.8 & 100.00 & EXACT-CMAX& 1.00 &    6.0 &    6.0 & 3533159.9 & 100.00 & EXACT-CMAX& 1.00 &    6.0 &    6.0 & 2009978.0 & 100.00 & EXACT-CMAX\\
			&   40& 1.00 &    7.0 &    7.0 & 52415.4 & 100.00 & EXACT-CMAX& 1.00 &    7.0 &    7.0 & 11885259.6 & 100.00 & EXACT-CMAX& 1.00 &    7.0 &    7.0 & 7210562.3 & 100.00 & EXACT-CMAX\\
			&   45& 1.00 &    8.0 &    8.0 & 74296.4 & 100.00 & EXACT-CMAX& 1.00 &    8.0 &    8.0 & 9733264.6 & 100.00 & EXACT-CMAX& 1.00 &    8.0 &    8.0 & 20473965.4 & 100.00 & EXACT-CMAX\\
			&   50& 1.00 &    9.0 &    9.0 & 120130.2 & 100.00 & EXACT-CMAX& 1.00 &    9.0 &    9.0 & 49846133.6 & 100.00 & EXACT-CMAX& 1.00 &    9.0 &    9.0 & 80588028.0 & 100.00 & EXACT-CMAX\\
			\hline
		\end{tabular}
		\caption{\EXACTCMAX applied for a small number of machines, using greedy-upper-bound heuristic to obtain upper bound on $\cmaxcost$}
		\label{tab:tree_vs_exactcmax2}
	\end{minipage}%
	
\end{table}

We have also shown the memory exhaustion and timeouts in \Cref{tab:tree_time_landscape}.
Sadly, it seems that for $m \ge 8$ machines, the question is rather if the algorithm will crash due to excessive memory consumption, or the user will eventually terminate it due to excessive running time.
Similarly to the \EXACTCMAX case, for $b_{\min}$ in quite few cases the processing time decreases for the next neighbour.
And a similar possible explanation is that this is due to lower number of possible assignments.
If blocks have sizes close to $m$, as it should be the case for $b_{\min}$, the number of distinct colorings produced is small.
For example, assume a coloring where classes have sizes $(l_1, \ldots, l_i^*, \ldots, l_m)$, $m_i$ was assigned already a vertex from some block (some cut vertex) and the block has size $m$. 
In the end valid colorings resulting from this coloring have classes with sizes $(l_1+1, \ldots, l_i^*, \ldots, l_m+1)$.

Also, the results at the first glance might seem inconsistent, due to the fact that crashes due to memory exhaustion are intertwined with timeouts.
This is simply due to the fact that we are showing the first error that appeared. 
If there is a memory crash for any instance at all, then the test is classified as 'M' and we skip other instances.
If there is a timeout for any instance, then the tests is classified as 'T' and we skip other instances.

\begin{table}[H]
	\newcolumntype{H}{>{\setbox0=\hbox\bgroup}c<{\egroup}@{}}
	\newcolumntype{F}{HHHrHH}
	\footnotesize
	\centering
	\begin{tabular}{|c|c|c|F|F|F|}
					\hline
		\multicolumn{20}{|c|}{\TIMEOFCOMPUTATIONLABEL} \\
		\hline
		\hline
		$m$    & $n$  & $b_{\min}$ & \multicolumn{6}{c|}{$b_{\min}$}&\multicolumn{6}{c|}{$b_{\avg}$}&\multicolumn{6}{c|}{$b_{\max}$} \\		
		\hline 
		8 &  5 & 1 & 1.00 &   1.00 &   1.00 & 506.24 & 100.00  &   TREE& 1.00 &   1.00 &   1.00 & 2343.44 & 100.00  &   TREE& 1.00 &   1.00 &   1.00 & 3993.24 & 100.00  &   TREE\\
		&  10& 2 &1.00 &   2.00 &   2.00 & 12848.76 & 100.00  &   TREE& 1.00 &   2.00 &   2.00 & 189314.12 & 100.00  &   TREE& 1.00 &   2.00 &   2.00 & 351584.56 & 100.00  &   TREE\\
		&  15& 2 & 1.00 &   2.00 &   2.00 & 1401.36 & 100.00  &   TREE& 1.00 &   2.00 &   2.00 & 4039424.84 & 100.00  &   TREE& 1.00 &   2.00 &   2.00 & 9041523.04 & 100.00  &   TREE\\
		&  20& 3 &1.00 &   3.00 &   3.00 & 20674.32 & 100.00  &   TREE&   ?   &      T   &      0   &      T   &  80.00 &   TREE& 1.00 &   3.00 &   3.00 & 173041589.08 & 100.00  &   TREE\\
		&  25& 4 & 1.00 &   4.00 &   4.00 & 218526.00 & 100.00  &   TREE&   M   &      M   &      M   &      M   &      M &       &   ?   &      T   &      0   &      T   &   4.00 &   TREE\\
		&  30& 5 &1.00 &   4.00 &   4.00 & 1147264.80 & 100.00  &   TREE&   M   &      M   &      M   &      M   &      M &       &   M   &      M   &      M   &      M   &      M &       \\
		&  35& 5 &1.00 &   5.00 &   5.00 & 108673.48 & 100.00  &   TREE&   ?   &      T   &      0   &      T   &      0 &   TREE&   ?   &      T   &      0   &      T   &      0 &   TREE\\
		&  40& 6 &1.00 &   5.80 &   5.80 & 572527.00 & 100.00  &   TREE&   M   &      M   &      M   &      M   &      M &       &   ?   &      T   &      0   &      T   &      0 &   TREE\\
		&  45& 7 &1.00 &   6.00 &   6.00 & 5720324.32 & 100.00  &   TREE&   M   &      M   &      M   &      M   &      M &       &   M   &      M   &      M   &      M   &      M &       \\
		&  50& 7 &1.00 &   7.00 &   7.00 & 151566.84 & 100.00  &   TREE&   ?   &      T   &      0   &      T   &      0 &   TREE&   M   &      M   &      M   &      M   &      M &       \\
		\hline
	\end{tabular}
	\caption{Processing times of \TREE when applied to $P|G=\blockgraph, p_j=1|\cmaxcost$ with $\varepsilon = 0$ (exact algorithm).}
	\label{tab:tree_time_landscape}
\end{table}

As a final note, we observe in \Cref{tab:treewidth_running_time_vs_epsilon} that the running times are exponential in $1/\varepsilon$, as expected.
\begin{table}[H]
	\newcolumntype{H}{>{\setbox0=\hbox\bgroup}c<{\egroup}@{}}
	\newcolumntype{F}{HHHrHH}
	\footnotesize
	\centering
	\begin{tabular}{|c|c|c|F|F|F|}
		\hline
		\multicolumn{20}{|c|}{\TIMEOFCOMPUTATIONLABEL} \\
		\hline
		\hline
		$\varepsilon$ &  $m$ & $n$                     & \multicolumn{6}{c|}{$b_{\min}$}&\multicolumn{6}{c|}{$b_{\avg}$}&\multicolumn{6}{c|}{$b_{\max}$} \\
		\hline 
		0.2 & 4 &   25& 0.30 &  37.04 & 123.76 & 2315540.16 & 100.00  &   TREE&   ?   &      T   &      0   &      T   &  96.00 &   TREE&   M   &      M   &      M   &      M   &      M &       \\
		\hline 
		0.4 & 4 &   25& 0.62 &  37.32 &  60.28 & 1796568.52 & 100.00  &   TREE& 0.59 &  35.72 &  60.16 & 71583312.68 & 100.00  &   TREE&   ?   &      T   &      0   &      T   &  92.00 &   TREE\\
		\hline 
		0.6 & 4 &   25& 0.95 &  37.36 &  39.32 & 935911.48 & 100.00  &   TREE& 0.91 &  36.00 &  39.64 & 27690947.56 & 100.00  &   TREE&   ?   &      T   &      0   &      T   &  96.00 &   TREE\\
		\hline 
		0.8 & 4 &   25& 1.33 &  38.00 &  28.56 & 599572.00 & 100.00  &   TREE& 1.28 &  36.16 &  28.36 & 5458285.48 & 100.00  &   TREE& 1.33 &  37.36 &  28.00 & 22892148.92 & 100.00  &   TREE\\
		\hline 
		1.0 & 4 &   25& 1.71 &  38.16 &  22.28 & 312753.60 & 100.00  &   TREE& 1.63 &  36.40 &  22.36 & 2346230.44 & 100.00  &   TREE& 1.67 &  37.52 &  22.44 & 8741508.80 & 100.00  &   TREE\\
		\hline
	\end{tabular}
\caption{The running times for various $\varepsilon$.}
\label{tab:treewidth_running_time_vs_epsilon}
\end{table}

\subsection{Empirical results on PTAS for identical machines and unit tasks}

As indicated in the proof of \Cref{thm:ptas_identical_unit_block}, PTAS for identical machines and unit tasks consists of $3$ subprocedures:
\begin{itemize}
    \item the first one used is \EXACTCMAX, consisting of \Cref{algorithm:composing-subcolorings-for-subblockgraphs,algorithm:merger-of-earlier-cliques-and-this-clique,algorithm:colorings-of-block-without-parent-to-colorings-with-parent,algorithm:main-loop}, used to compute cases with small $\cmaxcost$; 
    more precisely, to either prove that there is no schedule with $\cmaxcost \le \frac{2}{\varepsilon}$ or to find one;
    \item when the number of machines is relatively small ($m \le \frac{2}{\varepsilon} + 1$) the \TREE, a simplified version of \Cref{alg:arbitrary_tw}, is used;
    \item and in the other cases the already discussed \Cref{alg:2apx_identical_block} \GREEDY is used -- when \EXACTCMAX returns \NO and $m > \frac{\varepsilon}{2} + 1$.
\end{itemize}
Note that, although a proper choice of designated areas for the algorithms guarantees both the approximation ratio and running times required for PTAS, all three algorithms have widely different running times.
In particular, \GREEDY does not refer to the approximation ratio at all and finishes in linear time, whereas the two others have an exponential dependency on $\varepsilon$, thus making them impractical for medium-size $\cmaxcost$ (in case of \EXACTCMAX) or number of machines $m$ (in case of \TREE).
This means that it is a good strategy to run \GREEDY and check if the returned schedule achieves the desired approximation ratio.

\begin{table}[H]
	\newcolumntype{H}{>{\setbox0=\hbox\bgroup}c<{\egroup}@{}}
	\newcolumntype{F}{HHHHrHr}
	\footnotesize
	\centering
	\begin{tabular}{|c|c|c|F|F|F|}
				\hline
		\multicolumn{24}{|c|}{\TIMEOFCOMPUTATIONLABEL, Participation of algorithms} \\
		\hline
		\hline
		$\varepsilon$  & $m$ & $n$  & \multicolumn{7}{c|}{$b_{\min}$}&\multicolumn{7}{c|}{$b_{\avg}$}&\multicolumn{7}{c|}{$b_{\max}$} \\		
		\hline 
		0.25 & 4 &   10&   3 &  1.00 &   3.00 &   3.00 & 751.72 & 100.00  & C: 100  &   6 &  1.00 &   3.00 &   3.00 & 2583.84 & 100.00  & C: 100  &   9 &  1.00 &   3.00 &   3.00 & 11035.24 & 100.00  & C: 100  \\
		&   &   20&   7 &  1.15 &   5.76 &   5.00 & 2928.64 & 100.00  & C: 100  &  13 &  1.01 &   5.04 &   5.00 & 25312.72 & 100.00  & C: 100  &  19 &  1.00 &   5.00 &   5.00 & 88463.20 & 100.00  & C: 100  \\
		&   &   30&  10 &  1.00 &   8.08 &   8.08 & 5015.52 & 100.00  & C: 92  C+T: \quad8  &  19 &  1.00 &   8.00 &   8.00 & 64434.80 & 100.00  & C: 100  &  29 &  1.00 &   8.00 &   8.00 & 153181.96 & 100.00  & C: 100  \\
		&   &   40&  13 &  1.00 &  10.40 &  10.40 & 6324.96 & 100.00  & C+T: 100  &  26 &  1.00 &  10.04 &  10.04 & 299933.88 & 100.00  & C+T: 100  &  39 &  1.00 &  10.00 &  10.00 & 734093.92 & 100.00  & C+T: 100  \\
		&   &   50&  17 &  1.00 &  13.08 &  13.08 & 11605.04 & 100.00  & C+T: 100  &  33 &  1.00 &  13.00 &  13.00 & 709405.36 & 100.00  & C+T: 100  &  49 &  1.00 &  13.00 &  13.00 & 2109583.76 & 100.00  & C+T: 100  \\
		\hline 
		& 6 &   10&   2 &  1.00 &   2.00 &   2.00 & 2263.32 & 100.00  & C: 100  &   5 &  1.00 &   2.00 &   2.00 & 8977.00 & 100.00  & C: 100  &   9 &  1.00 &   2.00 &   2.00 & 21894.80 & 100.00  & C: 100  \\
		&   &   20&   4 &  1.00 &   4.00 &   4.00 & 11377.40 & 100.00  & C: 100  &  11 &  1.00 &   4.00 &   4.00 & 657704.04 & 100.00  & C: 100  &  19 &  1.00 &   4.00 &   4.00 & 1867408.52 & 100.00  & C: 100  \\
		&   &   30&   6 &  1.11 &   5.56 &   5.00 & 26017.20 & 100.00  & C: 100  &  17 &  1.00 &   5.00 &   5.00 & 8150121.76 & 100.00  & C: 100  &  29 &  1.00 &   5.00 &   5.00 & 12994871.92 & 100.00  & C: 100  \\
		&   &   40&   8 &  1.01 &   7.08 &   7.00 & 46895.52 & 100.00  & C: 100  &  23 &  1.00 &   7.00 &   7.00 & 29623956.08 & 100.00  & C: 100  &  39 &  1.00 &   7.00 &   7.00 & 22302443.92 & 100.00  & C: 100  \\
		&   &   50&  10 &  1.00 &   9.00 &   9.00 & 122690.36 & 100.00  & C+T: 100  &  29   &   M   &      M   &      M   &      M   &      M &       &  49   &   M   &      M   &      M   &      M   &      M &       \\
		\hline 
		& 8 &   10&   2 &  1.00 &   2.00 &   2.00 & 9498.00 & 100.00  & C: 100  &   5 &  1.00 &   2.00 &   2.00 & 28062.76 & 100.00  & C: 100  &   9 &  1.00 &   2.00 &   2.00 & 40711.28 & 100.00  & C: 100  \\
		&   &   20&   3 &  1.00 &   3.00 &   3.00 & 23902.68 & 100.00  & C: 100  &  11 &  1.00 &   3.00 &   3.00 & 3520367.32 & 100.00  & C: 100  &  19 &  1.00 &   3.00 &   3.00 & 8147424.84 & 100.00  & C: 100  \\
		&   &   30&   5 &  1.00 &   4.00 &   4.00 & 400850.40 & 100.00  & C: 100  &  17   &   M   &      M   &      M   &      M   &      M &       &  29   &   M   &      M   &      M   &      M   &      M &       \\
		&   &   40&   6 &  1.16 &   5.80 &   5.00 & 245312.24 & 100.00  & C: 100  &  22   &   M   &      M   &      M   &      M   &      M &       &  39   &   M   &      M   &      M   &      M   &      M &       \\
		&   &   50&   7 &  1.00 &   7.00 &   7.00 & 106114.36 & 100.00  & C: 100  &  28   &   M   &      M   &      M   &      M   &      M &       &  49   &   M   &      M   &      M   &      M   &      M &       \\
		\hline 
		0.50 & 4 &   10&   3 &  1.00 &   3.00 &   3.00 & 578.56 & 100.00  & C: 100  &   6 &  1.00 &   3.00 &   3.00 & 1750.96 & 100.00  & C: 100  &   9 &  1.00 &   3.00 &   3.00 & 2846.32 & 100.00  & C: 100  \\
		&   &   20&   7 &  1.00 &   5.76 &   5.76 & 3433.68 & 100.00  & C+T: 100  &  13 &  1.00 &   5.04 &   5.04 & 8304.20 & 100.00  & C+T: 100  &  19 &  1.00 &   5.00 &   5.00 & 29508.32 & 100.00  & C+T: 100  \\
		&   &   30&  10 &  1.00 &   8.08 &   8.08 & 3849.00 & 100.00  & C+T: 100  &  19 &  1.00 &   8.00 &   8.00 & 39502.52 & 100.00  & C+T: 100  &  29 &  1.00 &   8.00 &   8.00 & 156506.68 & 100.00  & C+T: 100  \\
		&   &   40&  13 &  1.00 &  10.40 &  10.40 & 4994.00 & 100.00  & C+T: 100  &  26 &  1.00 &  10.04 &  10.04 & 246833.64 & 100.00  & C+T: 100  &  39 &  1.00 &  10.00 &  10.00 & 568246.96 & 100.00  & C+T: 100  \\
		&   &   50&  17 &  1.00 &  13.08 &  13.08 & 9212.88 & 100.00  & C+T: 100  &  33 &  1.00 &  13.00 &  13.00 & 639152.08 & 100.00  & C+T: 100  &  49 &  1.00 &  13.00 &  13.00 & 1931143.76 & 100.00  & C+T: 100  \\
		\hline 
		& 6 &   10&   2 &  1.00 &   2.00 &   2.00 & 1918.76 & 100.00  & C: 100  &   5 &  1.00 &   2.00 &   2.00 & 8114.76 & 100.00  & C: 100  &   9 &  1.00 &   2.00 &   2.00 & 20179.72 & 100.00  & C: 100  \\
		&   &   20&   4 &  1.00 &   4.00 &   4.00 & 9270.08 & 100.00  & C: 100  &  11 &  1.00 &   4.00 &   4.00 & 77613.24 & 100.00  & C: 100  &  19 &  1.00 &   4.00 &   4.00 & 94206.76 & 100.00  & C: 100  \\
		&   &   30&   6 &  1.18 &   5.92 &   5.00 & 7486.96 & 100.00  & C+G: 100  &  17 &  1.10 &   5.52 &   5.00 & 81331.88 & 100.00  & C+G: 100  &  29 &  1.08 &   5.40 &   5.00 & 81895.68 & 100.00  & C+G: 100  \\
		&   &   40&   8 &  1.11 &   7.80 &   7.00 & 10821.96 & 100.00  & C+G: 100  &  23 &  1.00 &   7.00 &   7.00 & 125276.20 & 100.00  & C+G: 100  &  39 &  1.00 &   7.00 &   7.00 & 106815.44 & 100.00  & C+G: 100  \\
		&   &   50&  10 &  1.06 &   9.52 &   9.00 & 9051.60 & 100.00  & C+G: 100  &  29 &  1.00 &   9.00 &   9.00 & 157708.68 & 100.00  & C+G: 100  &  49 &  1.00 &   9.00 &   9.00 & 128766.84 & 100.00  & C+G: 100  \\
		\hline 
		& 8 &   10&   2 &  1.00 &   2.00 &   2.00 & 10118.88 & 100.00  & C: 100  &   5 &  1.00 &   2.00 &   2.00 & 22013.08 & 100.00  & C: 100  &   9 &  1.00 &   2.00 &   2.00 & 24548.44 & 100.00  & C: 100  \\
		&   &   20&   3 &  1.00 &   3.00 &   3.00 & 21817.52 & 100.00  & C: 100  &  11 &  1.00 &   3.00 &   3.00 & 802401.04 & 100.00  & C: 100  &  19 &  1.00 &   3.00 &   3.00 & 650324.44 & 100.00  & C: 100  \\
		&   &   30&   5 &  1.00 &   4.00 &   4.00 & 120187.88 & 100.00  & C: 100  &  17 &  1.00 &   4.00 &   4.00 & 884847.56 & 100.00  & C: 100  &  29 &  1.00 &   4.00 &   4.00 & 908802.56 & 100.00  & C: 100  \\
		&   &   40&   6 &  1.19 &   5.96 &   5.00 & 35854.08 & 100.00  & C+G: 100  &  22 &  1.07 &   5.36 &   5.00 & 2560255.04 & 100.00  & C+G: 100  &  39 &  1.06 &   5.28 &   5.00 & 1183361.00 & 100.00  & C+G: 100  \\
		&   &   50&   7 &  1.00 &   7.00 &   7.00 & 27263.28 & 100.00  & C+G: 100  &  28 &  1.00 &   7.00 &   7.00 & 1347918.72 & 100.00  & C+G: 100  &  49 &  1.00 &   7.00 &   7.00 & 1230175.36 & 100.00  & C+G: 100  \\
		\hline 
		0.75 & 4 &   10&   3 &  1.00 &   3.00 &   3.00 & 363.20 & 100.00  & C+G: 100  &   6 &  1.00 &   3.00 &   3.00 & 470.36 & 100.00  & C+G: 100  &   9 &  1.00 &   3.00 &   3.00 & 391.64 & 100.00  & C+G: 100  \\
		&   &   20&   7 &  1.18 &   5.92 &   5.00 & 692.20 & 100.00  & C+G: 100  &  13 &  1.10 &   5.52 &   5.00 & 740.40 & 100.00  & C+G: 100  &  19 &  1.08 &   5.40 &   5.00 & 653.84 & 100.00  & C+G: 100  \\
		&   &   30&  10 &  1.09 &   8.68 &   8.00 & 1012.00 & 100.00  & C+G: 100  &  19 &  1.01 &   8.08 &   8.00 & 1131.32 & 100.00  & C+G: 100  &  29 &  1.00 &   8.00 &   8.00 & 833.44 & 100.00  & C+G: 100  \\
		&   &   40&  13 &  1.17 &  11.68 &  10.00 & 1293.32 & 100.00  & C+G: 100  &  26 &  1.05 &  10.52 &  10.00 & 1382.92 & 100.00  & C+G: 100  &  39 &  1.04 &  10.40 &  10.00 & 1114.68 & 100.00  & C+G: 100  \\
		&   &   50&  17 &  1.08 &  14.04 &  13.00 & 1686.68 & 100.00  & C+G: 100  &  33 &  1.01 &  13.12 &  13.00 & 1716.36 & 100.00  & C+G: 100  &  49 &  1.00 &  13.04 &  13.00 & 1384.52 & 100.00  & C+G: 100  \\
		\hline 
		& 6 &   10&   2 &  1.00 &   2.00 &   2.00 & 1566.84 & 100.00  & C: 100  &   5 &  1.00 &   2.00 &   2.00 & 2579.20 & 100.00  & C: 100  &   9 &  1.00 &   2.00 &   2.00 & 1807.96 & 100.00  & C: 100  \\
		&   &   20&   4 &  1.00 &   4.00 &   4.00 & 2457.76 & 100.00  & C+G: 100  &  11 &  1.00 &   4.00 &   4.00 & 4664.52 & 100.00  & C+G: 100  &  19 &  1.00 &   4.00 &   4.00 & 2234.68 & 100.00  & C+G: 100  \\
		&   &   30&   6 &  1.18 &   5.92 &   5.00 & 4502.00 & 100.00  & C+G: 100  &  17 &  1.10 &   5.52 &   5.00 & 5777.28 & 100.00  & C+G: 100  &  29 &  1.08 &   5.40 &   5.00 & 4360.28 & 100.00  & C+G: 100  \\
		&   &   40&   8 &  1.11 &   7.80 &   7.00 & 3600.72 & 100.00  & C+G: 100  &  23 &  1.00 &   7.00 &   7.00 & 7919.16 & 100.00  & C+G: 100  &  39 &  1.00 &   7.00 &   7.00 & 3450.28 & 100.00  & C+G: 100  \\
		&   &   50&  10 &  1.06 &   9.52 &   9.00 & 4555.96 & 100.00  & C+G: 100  &  29 &  1.00 &   9.00 &   9.00 & 9303.60 & 100.00  & C+G: 100  &  49 &  1.00 &   9.00 &   9.00 & 4331.08 & 100.00  & C+G: 100  \\
		\hline 
		& 8 &   10&   2 &  1.00 &   2.00 &   2.00 & 6936.44 & 100.00  & C: 100  &   5 &  1.00 &   2.00 &   2.00 & 7883.68 & 100.00  & C: 100  &   9 &  1.00 &   2.00 &   2.00 & 3855.36 & 100.00  & C: 100  \\
		&   &   20&   3 &  1.00 &   3.00 &   3.00 & 12221.92 & 100.00  & C+G: 100  &  11 &  1.00 &   3.00 &   3.00 & 21752.96 & 100.00  & C+G: 100  &  19 &  1.00 &   3.00 &   3.00 & 9362.08 & 100.00  & C+G: 100  \\
		&   &   30&   5 &  1.05 &   4.20 &   4.00 & 13313.28 & 100.00  & C+G: 100  &  17 &  1.00 &   4.00 &   4.00 & 18278.64 & 100.00  & C+G: 100  &  29 &  1.00 &   4.00 &   4.00 & 11103.20 & 100.00  & C+G: 100  \\
		&   &   40&   6 &  1.19 &   5.96 &   5.00 & 7080.28 & 100.00  & C+G: 100  &  22 &  1.07 &   5.36 &   5.00 & 30312.92 & 100.00  & C+G: 100  &  39 &  1.06 &   5.28 &   5.00 & 8666.56 & 100.00  & C+G: 100  \\
		&   &   50&   7 &  1.00 &   7.00 &   7.00 & 14541.88 & 100.00  & C+G: 100  &  28 &  1.00 &   7.00 &   7.00 & 34259.92 & 100.00  & C+G: 100  &  49 &  1.00 &   7.00 &   7.00 & 17344.96 & 100.00  & C+G: 100  \\
		\hline 
		1.00 & 4 &   10&   3 &  1.00 &   3.00 &   3.00 & 362.52 & 100.00  & C+G: 100  &   6 &  1.00 &   3.00 &   3.00 & 477.24 & 100.00  & C+G: 100  &   9 &  1.00 &   3.00 &   3.00 & 387.68 & 100.00  & C+G: 100  \\
		&   &   20&   7 &  1.18 &   5.92 &   5.00 & 697.04 & 100.00  & C+G: 100  &  13 &  1.10 &   5.52 &   5.00 & 752.00 & 100.00  & C+G: 100  &  19 &  1.08 &   5.40 &   5.00 & 642.04 & 100.00  & C+G: 100  \\
		&   &   30&  10 &  1.09 &   8.68 &   8.00 & 1014.60 & 100.00  & C+G: 100  &  19 &  1.01 &   8.08 &   8.00 & 1170.04 & 100.00  & C+G: 100  &  29 &  1.00 &   8.00 &   8.00 & 807.64 & 100.00  & C+G: 100  \\
		&   &   40&  13 &  1.17 &  11.68 &  10.00 & 1285.68 & 100.00  & C+G: 100  &  26 &  1.05 &  10.52 &  10.00 & 1371.24 & 100.00  & C+G: 100  &  39 &  1.04 &  10.40 &  10.00 & 1123.52 & 100.00  & C+G: 100  \\
		&   &   50&  17 &  1.08 &  14.04 &  13.00 & 1679.68 & 100.00  & C+G: 100  &  33 &  1.01 &  13.12 &  13.00 & 1721.44 & 100.00  & C+G: 100  &  49 &  1.00 &  13.04 &  13.00 & 1380.00 & 100.00  & C+G: 100  \\
		\hline 
		& 6 &   10&   2 &  1.00 &   2.00 &   2.00 & 1560.36 & 100.00  & C: 100  &   5 &  1.00 &   2.00 &   2.00 & 2590.32 & 100.00  & C: 100  &   9 &  1.00 &   2.00 &   2.00 & 2098.36 & 100.00  & C: 100  \\
		&   &   20&   4 &  1.00 &   4.00 &   4.00 & 2618.64 & 100.00  & C+G: 100  &  11 &  1.00 &   4.00 &   4.00 & 5645.76 & 100.00  & C+G: 100  &  19 &  1.00 &   4.00 &   4.00 & 2861.28 & 100.00  & C+G: 100  \\
		&   &   30&   6 &  1.18 &   5.92 &   5.00 & 4673.00 & 100.00  & C+G: 100  &  17 &  1.10 &   5.52 &   5.00 & 6701.32 & 100.00  & C+G: 100  &  29 &  1.08 &   5.40 &   5.00 & 4025.40 & 100.00  & C+G: 100  \\
		&   &   40&   8 &  1.11 &   7.80 &   7.00 & 5763.88 & 100.00  & C+G: 100  &  23 &  1.00 &   7.00 &   7.00 & 7341.36 & 100.00  & C+G: 100  &  39 &  1.00 &   7.00 &   7.00 & 4389.72 & 100.00  & C+G: 100  \\
		&   &   50&  10 &  1.06 &   9.52 &   9.00 & 4778.28 & 100.00  & C+G: 100  &  29 &  1.00 &   9.00 &   9.00 & 9292.96 & 100.00  & C+G: 100  &  49 &  1.00 &   9.00 &   9.00 & 4324.64 & 100.00  & C+G: 100  \\
		\hline 
		& 8 &   10&   2 &  1.00 &   2.00 &   2.00 & 8520.16 & 100.00  & C: 100  &   5 &  1.00 &   2.00 &   2.00 & 9690.32 & 100.00  & C: 100  &   9 &  1.00 &   2.00 &   2.00 & 5182.36 & 100.00  & C: 100  \\
		&   &   20&   3 &  1.00 &   3.00 &   3.00 & 7363.08 & 100.00  & C+G: 100  &  11 &  1.00 &   3.00 &   3.00 & 20959.00 & 100.00  & C+G: 100  &  19 &  1.00 &   3.00 &   3.00 & 6295.80 & 100.00  & C+G: 100  \\
		&   &   30&   5 &  1.05 &   4.20 &   4.00 & 13229.32 & 100.00  & C+G: 100  &  17 &  1.00 &   4.00 &   4.00 & 17973.52 & 100.00  & C+G: 100  &  29 &  1.00 &   4.00 &   4.00 & 7362.28 & 100.00  & C+G: 100  \\
		&   &   40&   6 &  1.19 &   5.96 &   5.00 & 10187.92 & 100.00  & C+G: 100  &  22 &  1.07 &   5.36 &   5.00 & 26987.60 & 100.00  & C+G: 100  &  39 &  1.06 &   5.28 &   5.00 & 8656.16 & 100.00  & C+G: 100  \\
		&   &   50&   7 &  1.00 &   7.00 &   7.00 & 14291.24 & 100.00  & C+G: 100  &  28 &  1.00 &   7.00 &   7.00 & 33998.72 & 100.00  & C+G: 100  &  49 &  1.00 &   7.00 &   7.00 & 17385.04 & 100.00  & C+G: 100  \\
		\hline
	\end{tabular}

	\caption{
		Results on \PTAS. 
		For a given $\varepsilon, m, n$ and $b$ function the first value is the computation time.
		The second value is \% of instances solved by particular combination of algorithms.
		A lone letter $C$ means \EXACTCMAX, $C+T$ means by \EXACTCMAX and \TREE (applied due to the \NO answer from \EXACTCMAX); $C+G$ means \EXACTCMAX and \GREEDY. 
	}
	\label{tab:ptas}
\end{table}

Let us notice a few facts about the results presented in \Cref{tab:ptas}.
There are significant differences in running times, due to the mentioned peculiarities of the algorithms, but also due to invocation of distinct subalgorithms due to distinct combinations of $\varepsilon$ and $m$.
Our tests show that another optimization would be reasonable.
We should allow to return an early \NO answer from \EXACTCMAX: if for some part of block graph there is no schedule with $\cmaxcost$ smaller than the bound, there cannot be any such a schedule for the whole graph.
Hence, in fact, for a fixed $m$ and $\varepsilon$, the 'proper' behavior of running time function should be to be limited by some big constant value. 

Overall, since \GREEDY is very fast and has quite good approximation ratio, a decent heuristic in practical setting would be to run it along others and return its result if the other computations exceed predefined limits.
Or, even run it and compare the $\cmaxcost$ of the constructed schedule with the natural lower bound.
It seems that even for small $\varepsilon$, and big enough data, there is a good chance that the solution will meet the requirements.
 
\subsection{Empirical results on \Cref{alg:uniform_flow_network} -- \FLOW}

In this subsection we present empirical results of scheduling achieved by \Cref{alg:uniform_flow_network} -- discussed in \Cref{thm:uniform_block_unit_cut}, the exact algorithm for $Q|G = \blockgraph, \cutvertices(G) \le k, p_j = 1|\cmaxcost$ problem.
Our implementation (called \FLOW) uses Boost graph library to solve instances of maximum flow, in particular the Boykov-Kolmogorov algorithm from this library.
The algorithm \FLOW was executed for $m \in \{4,6,8\}$ uniform machines for $n \in \{5, \ldots, 50\}$ jobs.
\begin{table}[H]
	\newcolumntype{H}{>{\setbox0=\hbox\bgroup}c<{\egroup}@{}}
	\newcolumntype{F}{HHHrHH}
	\footnotesize
	\centering
	\begin{minipage}{.5\linewidth}
		\begin{tabular}{|c|c|F|F|F|}
					\hline
			\multicolumn{20}{|c|}{\TIMEOFCOMPUTATIONLABEL} \\
			\hline
			\hline
			$m$ & $n$  & \multicolumn{6}{c|}{$b_{\min}$}&\multicolumn{6}{c|}{$b_{\avg}$}&\multicolumn{6}{c|}{$b_{\max}$} \\
			\hline 
			4 &    5& 2.12 &    0.7 &    0.3 &  220.9 & 100.00 &   FLOW& 1.28 &    0.4 &    0.3 &  236.2 & 100.00 &   FLOW& 1.28 &    0.4 &    0.3 &  267.3 & 100.00 &   FLOW\\
			&   10& 3.14 &    2.0 &    0.6 & 1661.5 & 100.00 &   FLOW& 1.42 &    0.9 &    0.6 & 1018.8 & 100.00 &   FLOW& 1.28 &    0.8 &    0.6 & 1595.0 & 100.00 &   FLOW\\
			&   15& 2.18 &    2.0 &    0.9 & 8472.5 & 100.00 &   FLOW& 1.54 &    1.4 &    0.9 & 9345.1 & 100.00 &   FLOW& 1.07 &    1.0 &    0.9 & 15198.4 & 100.00 &   FLOW\\
			&   20& 2.40 &    3.0 &    1.2 & 30493.7 & 100.00 &   FLOW& 1.45 &    1.8 &    1.2 & 74451.4 & 100.00 &   FLOW& 1.12 &    1.4 &    1.2 & 497615.8 & 100.00 &   FLOW\\
			&   25& 2.69 &    4.2 &    1.6 & 105894.0 & 100.00 &   FLOW& 1.44 &    2.2 &    1.6 & 1253187.3 & 100.00 &   FLOW& 1.02 &    1.6 &    1.6 & 1780654.2 & 100.00 &   FLOW\\
			&   30& 2.39 &    4.5 &    1.9 & 537458.5 & 100.00 &   FLOW& 1.37 &    2.6 &    1.9 & 3270030.4 & 100.00 &   FLOW& 1.07 &    2.0 &    1.9 & 20978525.3 & 100.00 &   FLOW\\
			&   35& 2.41 &    5.3 &    2.2 & 2668535.2 & 100.00 &   FLOW& 1.24 &    2.7 &    2.2 & 28763959.4 & 100.00 &   FLOW&   T & $\infty$ &      0 &      T & 36.00 &   FLOW\\
			&   40& 2.61 &    6.5 &    2.5 & 5677564.4 & 100.00 &   FLOW&   T & $\infty$ &      0 &      T & 20.00 &   FLOW&   T & $\infty$ &      0 &      T & 8.00 &   FLOW\\
			&   45& 2.57 &    7.2 &    2.8 & 29771233.5 & 100.00 &   FLOW&   T & $\infty$ &      0 &      T & 20.00 &   FLOW&   T & $\infty$ &      0 &      T & 0 &   FLOW\\
			&   50& 2.61 &    8.2 &    3.1 & 129191637.0 & 100.00 &   FLOW&   T & $\infty$ &      0 &      T & 8.00 &   FLOW&   T & $\infty$ &      0 &      T & 0 &   FLOW\\
			\hline 
			6 &    5& 4.40 &    1.0 &    0.2 &  240.8 & 100.00 &   FLOW& 1.76 &    0.4 &    0.2 &  370.9 & 100.00 &   FLOW& 1.76 &    0.4 &    0.2 &  749.4 & 100.00 &   FLOW\\
			&   10& 2.20 &    1.0 &    0.5 & 1288.7 & 100.00 &   FLOW& 1.85 &    0.8 &    0.5 & 4275.3 & 100.00 &   FLOW& 1.32 &    0.6 &    0.5 & 11160.1 & 100.00 &   FLOW\\
			&   15& 2.93 &    2.0 &    0.7 & 6818.3 & 100.00 &   FLOW& 1.35 &    0.9 &    0.7 & 19892.1 & 100.00 &   FLOW& 1.17 &    0.8 &    0.7 & 423167.0 & 100.00 &   FLOW\\
			&   20& 2.68 &    2.4 &    0.9 & 86251.3 & 100.00 &   FLOW& 1.19 &    1.1 &    0.9 & 1703836.1 & 100.00 &   FLOW& 1.10 &    1.0 &    0.9 & 8450726.2 & 100.00 &   FLOW\\
			&   25& 2.78 &    3.2 &    1.1 & 452027.3 & 100.00 &   FLOW& 1.28 &    1.5 &    1.1 & 12484370.0 & 100.00 &   FLOW&   T & $\infty$ &      0 &      T & 60.00 &   FLOW\\
			&   30& 2.82 &    3.8 &    1.4 & 1160503.8 & 100.00 &   FLOW&   T & $\infty$ &      0 &      T & 60.00 &   FLOW&   T & $\infty$ &      0 &      T & 12.00 &   FLOW\\
			&   35& 2.87 &    4.6 &    1.6 & 4404040.2 & 100.00 &   FLOW&   T & $\infty$ &      0 &      T & 12.00 &   FLOW&   T & $\infty$ &      0 &      T & 0 &   FLOW\\
			&   40& 2.82 &    5.1 &    1.8 & 38533308.8 & 100.00 &   FLOW&   T & $\infty$ &      0 &      T & 0 &   FLOW&   T & $\infty$ &      0 &      T & 0 &   FLOW\\
			&   45&   T & $\infty$ &      0 &      T & 68.00 &   FLOW&   T & $\infty$ &      0 &      T & 0 &   FLOW&   T & $\infty$ &      0 &      T & 0 &   FLOW\\
			&   50&   T & $\infty$ &      0 &      T & 44.00 &   FLOW&   T & $\infty$ &      0 &      T & 0 &   FLOW&   T & $\infty$ &      0 &      T & 0 &   FLOW\\
			\hline 
			8 &    5& 1.16 &    0.2 &    0.2 &  175.7 & 100.00 &   FLOW& 1.16 &    0.2 &    0.2 &  205.2 & 100.00 &   FLOW& 1.16 &    0.2 &    0.2 &  311.7 & 100.00 &   FLOW\\
			&   10& 2.49 &    0.9 &    0.3 & 4514.8 & 100.00 &   FLOW& 1.28 &    0.4 &    0.3 & 11128.4 & 100.00 &   FLOW& 1.16 &    0.4 &    0.3 & 18260.3 & 100.00 &   FLOW\\
			&   15& 3.87 &    2.0 &    0.5 & 12912.4 & 100.00 &   FLOW& 1.35 &    0.7 &    0.5 & 210154.3 & 100.00 &   FLOW& 1.16 &    0.6 &    0.5 & 7056851.2 & 100.00 &   FLOW\\
			&   20& 2.87 &    2.0 &    0.7 & 78260.5 & 100.00 &   FLOW& 1.24 &    0.9 &    0.7 & 11081604.3 & 100.00 &   FLOW&   T & $\infty$ &      0 &      T & 36.00 &   FLOW\\
			&   25& 2.53 &    2.2 &    0.9 & 592010.4 & 100.00 &   FLOW& 1.16 &    1.0 &    0.9 & 45627315.2 & 100.00 &   FLOW&   T & $\infty$ &      0 &      T & 12.00 &   FLOW\\
			&   30& 2.40 &    2.5 &    1.0 & 2059662.1 & 100.00 &   FLOW&   T & $\infty$ &      0 &      T & 0 &   FLOW&   T & $\infty$ &      0 &      T & 0 &   FLOW\\
			&   35& 2.59 &    3.1 &    1.2 & 4460162.3 & 100.00 &   FLOW&   T & $\infty$ &      0 &      T & 0 &   FLOW&   T & $\infty$ &      0 &      T & 0 &   FLOW\\
			&   40& 2.52 &    3.5 &    1.4 & 14074302.0 & 100.00 &   FLOW&   T & $\infty$ &      0 &      T & 0 &   FLOW&   T & $\infty$ &      0 &      T & 0 &   FLOW\\
			&   45& 2.37 &    3.7 &    1.6 & 113356947.0 & 100.00 &   FLOW&   T & $\infty$ &      0 &      T & 0 &   FLOW&   T & $\infty$ &      0 &      T & 0 &   FLOW\\
			&   50& 2.83 &    4.9 &    1.7 & 220321546.7 & 100.00 &   FLOW&   T & $\infty$ &      0 &      T & 0 &   FLOW&   T & $\infty$ &      0 &      T & 0 &   FLOW\\
			\hline
		\end{tabular}
		\caption{\FLOW applied using Boykov-Kolmogorov algorithm.}
		\label{tab:bk5}
	\end{minipage}%
	\begin{minipage}{.5\linewidth}
		\begin{tabular}{|c|c|F|F|F|}
			\hline
			\multicolumn{20}{|c|}{\TIMEOFCOMPUTATIONLABEL} \\
			\hline
			\hline
			$m$ & $n$  & \multicolumn{6}{c|}{$b_{\min}$}&\multicolumn{6}{c|}{$b_{\avg}$}&\multicolumn{6}{c|}{$b_{\max}$} \\
			\hline 
			4 &    5& 2.25 &    0.2 &    0.1 &  661.3 & 100.00 &   FLOW& 1.31 &    0.1 &    0.1 &  856.8 & 100.00 &   FLOW& 1.31 &    0.1 &    0.1 & 1014.4 & 100.00 &   FLOW\\
			&   10& 3.38 &    0.5 &    0.1 & 4730.9 & 100.00 &   FLOW& 1.44 &    0.2 &    0.1 & 4711.9 & 100.00 &   FLOW& 1.20 &    0.2 &    0.1 & 11629.0 & 100.00 &   FLOW\\
			&   15& 2.35 &    0.5 &    0.2 & 16859.8 & 100.00 &   FLOW& 1.65 &    0.4 &    0.2 & 21350.0 & 100.00 &   FLOW& 1.10 &    0.2 &    0.2 & 32946.3 & 100.00 &   FLOW\\
			&   20& 2.59 &    0.8 &    0.3 & 59660.1 & 100.00 &   FLOW& 1.50 &    0.4 &    0.3 & 268096.4 & 100.00 &   FLOW& 1.05 &    0.3 &    0.3 & 2535209.9 & 100.00 &   FLOW\\
			&   25& 2.90 &    1.0 &    0.4 & 305569.8 & 100.00 &   FLOW& 1.53 &    0.6 &    0.4 & 3738761.0 & 100.00 &   FLOW& 1.05 &    0.4 &    0.4 & 12465548.0 & 100.00 &   FLOW\\
			&   30& 2.58 &    1.1 &    0.4 & 1117201.6 & 100.00 &   FLOW& 1.46 &    0.6 &    0.4 & 8011979.8 & 100.00 &   FLOW& 1.10 &    0.5 &    0.4 & 41983519.2 & 100.00 &   FLOW\\
			&   35& 2.60 &    1.3 &    0.5 & 5728300.3 & 100.00 &   FLOW& 1.31 &    0.7 &    0.5 & 77620135.1 & 100.00 &   FLOW&   T & $\infty$ &      0 &      T & 36.00 &   FLOW\\
			&   40& 2.81 &    1.6 &    0.6 & 11869819.1 & 100.00 &   FLOW&   T & $\infty$ &      0 &      T & 20.00 &   FLOW&   T & $\infty$ &      0 &      T & 0 &   FLOW\\
			&   45& 2.78 &    1.8 &    0.7 & 59703938.2 & 100.00 &   FLOW&   T & $\infty$ &      0 &      T & 8.00 &   FLOW&   T & $\infty$ &      0 &      T & 0 &   FLOW\\
			&   50&   T & $\infty$ &      0 &      T & 12.00 &   FLOW&   T & $\infty$ &      0 &      T & 0 &   FLOW&   T & $\infty$ &      0 &      T & 0 &   FLOW\\
			\hline 
			6 &    5& 4.90 &    0.2 &    0.1 &  852.3 & 100.00 &   FLOW& 1.57 &    0.1 &    0.1 & 1419.2 & 100.00 &   FLOW& 1.57 &    0.1 &    0.1 & 2865.9 & 100.00 &   FLOW\\
			&   10& 2.45 &    0.2 &    0.1 & 6893.9 & 100.00 &   FLOW& 1.99 &    0.2 &    0.1 & 17346.9 & 100.00 &   FLOW& 1.28 &    0.1 &    0.1 & 58074.2 & 100.00 &   FLOW\\
			&   15& 3.27 &    0.5 &    0.2 & 34862.8 & 100.00 &   FLOW& 1.48 &    0.2 &    0.2 & 63056.8 & 100.00 &   FLOW& 1.24 &    0.2 &    0.2 & 1087509.7 & 100.00 &   FLOW\\
			&   20& 2.99 &    0.6 &    0.2 & 287774.8 & 100.00 &   FLOW& 1.29 &    0.3 &    0.2 & 5932357.8 & 100.00 &   FLOW& 1.17 &    0.2 &    0.2 & 29078224.9 & 100.00 &   FLOW\\
			&   25& 3.10 &    0.8 &    0.3 & 1219885.9 & 100.00 &   FLOW& 1.37 &    0.4 &    0.3 & 58728861.5 & 100.00 &   FLOW&   T & $\infty$ &      0 &      T & 12.00 &   FLOW\\
			&   30& 3.14 &    1.0 &    0.3 & 3030384.4 & 100.00 &   FLOW&   T & $\infty$ &      0 &      T & 4.00 &   FLOW&   T & $\infty$ &      0 &      T & 0 &   FLOW\\
			&   35& 3.19 &    1.1 &    0.4 & 8554763.3 & 100.00 &   FLOW&   T & $\infty$ &      0 &      T & 0 &   FLOW&   T & $\infty$ &      0 &      T & 0 &   FLOW\\
			&   40& 3.14 &    1.3 &    0.4 & 123338982.2 & 100.00 &   FLOW&   T & $\infty$ &      0 &      T & 0 &   FLOW&   T & $\infty$ &      0 &      T & 0 &   FLOW\\
			&   45&   T & $\infty$ &      0 &      T & 60.00 &   FLOW&   T & $\infty$ &      0 &      T & 0 &   FLOW&   T & $\infty$ &      0 &      T & 0 &   FLOW\\
			&   50&   T & $\infty$ &      0 &      T & 4.00 &   FLOW&   T & $\infty$ &      0 &      T & 0 &   FLOW&   T & $\infty$ &      0 &      T & 0 &   FLOW\\
			\hline 
			8 &    5& 1.21 &    0.0 &    0.0 &  778.1 & 100.00 &   FLOW& 1.21 &    0.0 &    0.0 & 1156.6 & 100.00 &   FLOW& 1.21 &    0.0 &    0.0 & 2314.1 & 100.00 &   FLOW\\
			&   10& 2.88 &    0.2 &    0.1 & 16121.8 & 100.00 &   FLOW& 1.52 &    0.1 &    0.1 & 33627.0 & 100.00 &   FLOW& 1.21 &    0.1 &    0.1 & 118551.4 & 100.00 &   FLOW\\
			&   15& 4.23 &    0.5 &    0.1 & 33455.3 & 100.00 &   FLOW& 1.46 &    0.2 &    0.1 & 417418.0 & 100.00 &   FLOW& 1.21 &    0.1 &    0.1 & 12229761.4 & 100.00 &   FLOW\\
			&   20& 3.18 &    0.5 &    0.2 & 184903.6 & 100.00 &   FLOW& 1.32 &    0.2 &    0.2 & 34403618.2 & 100.00 &   FLOW&   T & $\infty$ &      0 &      T & 12.00 &   FLOW\\
			&   25& 2.84 &    0.6 &    0.2 & 1842825.8 & 100.00 &   FLOW&   T & $\infty$ &      0 &      T & 40.00 &   FLOW&   T & $\infty$ &      0 &      T & 0 &   FLOW\\
			&   30& 2.68 &    0.6 &    0.2 & 6570201.9 & 100.00 &   FLOW&   T & $\infty$ &      0 &      T & 0 &   FLOW&   T & $\infty$ &      0 &      T & 0 &   FLOW\\
			&   35& 2.84 &    0.8 &    0.3 & 14304703.9 & 100.00 &   FLOW&   T & $\infty$ &      0 &      T & 0 &   FLOW&   T & $\infty$ &      0 &      T & 0 &   FLOW\\
			&   40& 2.80 &    0.9 &    0.3 & 36371260.0 & 100.00 &   FLOW&   T & $\infty$ &      0 &      T & 0 &   FLOW&   T & $\infty$ &      0 &      T & 0 &   FLOW\\
			&   45&   T & $\infty$ &      0 &      T & 40.00 &   FLOW&   T & $\infty$ &      0 &      T & 0 &   FLOW&   T & $\infty$ &      0 &      T & 0 &   FLOW\\
			&   50&   T & $\infty$ &      0 &      T & 24.00 &   FLOW&   T & $\infty$ &      0 &      T & 0 &   FLOW&   T & $\infty$ &      0 &      T & 0 &   FLOW\\
			\hline
		\end{tabular}
		\caption{\FLOW applied using Boykov-Kolmogorov algorithm.}
		\label{tab:bk25}
	\end{minipage}%
\end{table}

Since the computations turned out time-consuming, we gave up examining graphs created for $n > 50$ vertices.  
Moreover, also due to time restrictions, we applied the algorithm for only a few sequences of speeds.
Namely, the speeds of machines generated for $s \in [1, 5]$, are $5,5,5,1$ for $4$ machines; $5,5,5,5,5,1,1$ for $6$ machines; and $5,5,5,5,5,2,1,1$ for $8$ machines.
And the speeds of the machines generated for $s \in [1, 25]$, are $23, 21, 21, 4$ for $f$ machines; $25, 23, 21, 21, 4, 4$ for $6$ machines; $25, 23, 23, 21, 21, 6, 4, 4$ for $8$ machines.

Due to the report in \cite{GoldbergHKTW11} that Boykov-Kolmogorov behaves badly for some instances of general flow problems (the algorithm is dedicated for applications in computer vision), we have checked also another flow algorithm presented in \cite{goldberg85}, which is also implemented in Boost.
As a side note, there exists an algorithm which has better asymptotic worst-case time complexity ($\Osymbol(mn)$) -- Orlin's Algorithm.
However, the library does not contain an implementation of it; moreover to the best of our knowledge the Orlin's Algorithm has very high constants time-constant, which implies its impracticality.

\begin{table}[H]
	\newcolumntype{H}{>{\setbox0=\hbox\bgroup}c<{\egroup}@{}}
	\newcolumntype{F}{HHHrHH}
	\footnotesize
	\centering
	\begin{tabular}{|c|c|F|F|F|}
		\hline
		\multicolumn{20}{|c|}{\TIMEOFCOMPUTATIONLABEL} \\
		\hline
		\hline
		
	    $m$	& $n$  & \multicolumn{6}{c|}{$b_{\min}$}&\multicolumn{6}{c|}{$b_{\avg}$}&\multicolumn{6}{c|}{$b_{\max}$} \\
		\hline 
		4 &    5& 2.25 &    0.2 &    0.1 &  807.6 & 100.00 &   FLOW& 1.31 &    0.1 &    0.1 & 1062.4 & 100.00 &   FLOW& 1.31 &    0.1 &    0.1 & 1261.4 & 100.00 &   FLOW\\
		&   10& 3.38 &    0.5 &    0.1 & 4688.3 & 100.00 &   FLOW& 1.44 &    0.2 &    0.1 & 4394.1 & 100.00 &   FLOW& 1.20 &    0.2 &    0.1 & 9023.2 & 100.00 &   FLOW\\
		&   15& 2.35 &    0.5 &    0.2 & 20882.7 & 100.00 &   FLOW& 1.65 &    0.4 &    0.2 & 26550.2 & 100.00 &   FLOW& 1.10 &    0.2 &    0.2 & 40292.6 & 100.00 &   FLOW\\
		&   20& 2.59 &    0.8 &    0.3 & 76475.4 & 100.00 &   FLOW& 1.50 &    0.4 &    0.3 & 327202.4 & 100.00 &   FLOW& 1.05 &    0.3 &    0.3 & 3138793.6 & 100.00 &   FLOW\\
		&   25& 2.90 &    1.0 &    0.4 & 386726.8 & 100.00 &   FLOW& 1.53 &    0.6 &    0.4 & 4433535.8 & 100.00 &   FLOW& 1.05 &    0.4 &    0.4 & 15197871.9 & 100.00 &   FLOW\\
		&   30& 2.58 &    1.1 &    0.4 & 1383803.2 & 100.00 &   FLOW& 1.46 &    0.6 &    0.4 & 9431244.6 & 100.00 &   FLOW& 1.10 &    0.5 &    0.4 & 45093214.0 & 100.00 &   FLOW\\
		&   35& 2.60 &    1.3 &    0.5 & 6986467.0 & 100.00 &   FLOW& 1.31 &    0.7 &    0.5 & 89669706.8 & 100.00 &   FLOW&   ? & $\infty$ &      0 &      T & 16.00 &   FLOW\\
		&   40& 2.81 &    1.6 &    0.6 & 14539747.6 & 100.00 &   FLOW&   ? & $\infty$ &      0 &      T & 20.00 &   FLOW&   ? & $\infty$ &      0 &      T & 0 &   FLOW\\
		&   45& 2.78 &    1.8 &    0.7 & 72666190.0 & 100.00 &   FLOW&   ? & $\infty$ &      0 &      T & 8.00 &   FLOW&   ? & $\infty$ &      0 &      T & 0 &   FLOW\\
		&   50&   ? & $\infty$ &      0 &      T & 12.00 &   FLOW&   ? & $\infty$ &      0 &      T & 0 &   FLOW&   ? & $\infty$ &      0 &     T & 0 &   FLOW\\
		\hline 
		6 &    5& 4.90 &    0.2 &    0.1 & 1001.3 & 100.00 &   FLOW& 1.57 &    0.1 &    0.1 & 1953.4 & 100.00 &   FLOW& 1.57 &    0.1 &    0.1 & 6692.3 & 100.00 &   FLOW\\
		&   10& 2.45 &    0.2 &    0.1 & 8131.9 & 100.00 &   FLOW& 1.99 &    0.2 &    0.1 & 25443.9 & 100.00 &   FLOW& 1.28 &    0.1 &    0.1 & 73486.8 & 100.00 &   FLOW\\
		&   15& 3.27 &    0.5 &    0.2 & 38638.9 & 100.00 &   FLOW& 1.48 &    0.2 &    0.2 & 77664.9 & 100.00 &   FLOW& 1.24 &    0.2 &    0.2 & 1345379.4 & 100.00 &   FLOW\\
		&   20& 2.99 &    0.6 &    0.2 & 338600.4 & 100.00 &   FLOW& 1.29 &    0.3 &    0.2 & 7253648.3 & 100.00 &   FLOW& 1.17 &    0.2 &    0.2 & 32377854.2 & 100.00 &   FLOW\\
		&   25& 3.10 &    0.8 &    0.3 & 1434680.4 & 100.00 &   FLOW& 1.37 &    0.4 &    0.3 & 70249216.6 & 100.00 &   FLOW&   ? & $\infty$ &      0 &      T & 12.00 &   FLOW\\
		&   30& 3.14 &    1.0 &    0.3 & 3544487.5 & 100.00 &   FLOW&   ? & $\infty$ &      0 &      T & 4.00 &   FLOW&   ? & $\infty$ &      0 &      T & 0 &   FLOW\\
		&   35& 3.19 &    1.1 &    0.4 & 9976071.1 & 100.00 &   FLOW&   ? & $\infty$ &      0 &      T & 0 &   FLOW&   ? & $\infty$ &      0 &      T & 0 &   FLOW\\
		&   40&   ? & $\infty$ &      0 &      T & 76.00 &   FLOW&   ? & $\infty$ &      0 &      T & 0 &   FLOW&   ? & $\infty$ &      0 &      T & 0 &   FLOW\\
		&   45&   ? & $\infty$ &      0 &      T & 60.00 &   FLOW&   ? & $\infty$ &      0 &      T & 0 &   FLOW&   ? & $\infty$ &      0 &      T & 0 &   FLOW\\
		&   50&   ? & $\infty$ &      0 &      T & 4.00 &   FLOW&   ? & $\infty$ &      0 &      T & 0 &   FLOW&   ? & $\infty$ &      0 &      T & 0 &   FLOW\\
		\hline 
		8 &    5& 1.21 &    0.0 &    0.0 &  890.9 & 100.00 &   FLOW& 1.21 &    0.0 &    0.0 & 1409.0 & 100.00 &   FLOW& 1.21 &    0.0 &    0.0 & 3079.2 & 100.00 &   FLOW\\
		&   10& 2.88 &    0.2 &    0.1 & 15368.8 & 100.00 &   FLOW& 1.52 &    0.1 &    0.1 & 40684.4 & 100.00 &   FLOW& 1.21 &    0.1 &    0.1 & 154451.8 & 100.00 &   FLOW\\
		&   15& 4.23 &    0.5 &    0.1 & 22989.2 & 100.00 &   FLOW& 1.46 &    0.2 &    0.1 & 502285.0 & 100.00 &   FLOW& 1.21 &    0.1 &    0.1 & 14915309.3 & 100.00 &   FLOW\\
		&   20& 3.18 &    0.5 &    0.2 & 202640.6 & 100.00 &   FLOW& 1.32 &    0.2 &    0.2 & 41045127.0 & 100.00 &   FLOW&   ? & $\infty$ &      0 &      T & 12.00 &   FLOW\\
		&   25& 2.84 &    0.6 &    0.2 & 2081356.1 & 100.00 &   FLOW&   ? & $\infty$ &      0 &      T & 40.00 &   FLOW&   ? & $\infty$ &      0 &      T & 0 &   FLOW\\
		&   30& 2.68 &    0.6 &    0.2 & 7415713.6 & 100.00 &   FLOW&   ? & $\infty$ &      0 &      T & 0 &   FLOW&   ? & $\infty$ &      0 &      T & 0 &   FLOW\\
		&   35& 2.84 &    0.8 &    0.3 & 16083488.5 & 100.00 &   FLOW&   ? & $\infty$ &      0 &      T & 0 &   FLOW&   ? & $\infty$ &      0 &      T & 0 &   FLOW\\
		&   40& 2.80 &    0.9 &    0.3 & 40060636.0 & 100.00 &   FLOW&   ? & $\infty$ &      0 &      T & 0 &   FLOW&   ? & $\infty$ &      0 &      T & 0 &   FLOW\\
		&   45&   ? & $\infty$ &      0 &      T & 40.00 &   FLOW&   ? & $\infty$ &      0 &      T & 0 &   FLOW&   ? & $\infty$ &      0 &      T & 0 &   FLOW\\
		&   50&   ? & $\infty$ &      0 &      T & 24.00 &   FLOW&   ? & $\infty$ &      0 &      T & 0 &   FLOW&   ? & $\infty$ &      0 &      T & 0 &   FLOW\\
		\hline
	\end{tabular}
	\caption{An application of \FLOW for the generated from $[1,25]$ speeds and Goldberg push-relabel algorithm}
	\label{tab:goldberg}
\end{table}
The first observation is that, the computations times for instances with $b_{\max}$ are greater than with $b_{\avg}$, which are greater than $b_{\min}$.
This is expected, more blocks are corresponding to  bigger numbers of cut-vertices, on average, hence, also numbers of assignments are greater, on average. 

The next observation is that there is a significant difference between running times for instances with speeds in $[1,5]$ and with speeds in $[1,25]$.
To analyze this, we have checked what are the average number of $\cmaxcost$ guesses.
We have analyzed this on an example with $b=b_{\min}=13$, $n=40$, $m=4$. 
For $s \in [1,5]$, the average number of guesses was $7.6$.
For $s \in [1,25]$, the average number of guesses was $16.48$.
After rescaling the processing times for $s \in [1,5]$ by $\frac{16.48}{7.6}$, the average difference between the rescaled time and time for $s \in [1,25]$ was less than $10\%$.
The different number of guesses can be explained by the fact that the maximum size of pool of guesses has size bounded between $n$ (all the machines have equal speeds) and $nm$ (all the machines have different speeds). 
This means that, to make the algorithm more practical, a more approximation oriented approach shall be considered.
For example, if the ratio between $LB$ (lower bound on optimal $\cmaxcost$) and $UB$ (upper bound on optimal $\cmaxcost$) is less than some, given by an user, $1+\varepsilon$, then perhaps the calculations can be ended a bit earlier.

\begin{table}[H]
	\newcolumntype{H}{>{\setbox0=\hbox\bgroup}c<{\egroup}@{}}
	\newcolumntype{F}{HHHrHH}
	\footnotesize
	\centering
	\begin{tabular}{|c|c|r|r|r|r|r|r|}
		\hline
		$m$   & speeds & $b_{\min}$ & \TIMEOFCOMPUTATIONLABEL & $\#Cut\textrm{-}Vertices$ &  $\#Assignments$ & $\#Flows$ tested                  & $\frac{\#\YES}{\#\NO}$ \\
		\hline
		$4$ & $5,5,5,1$      & $13$        &  $6795740.1$    &   $8.4$               & $2605936.8$                   &   $ 39373.6$              & $1.33$ \\
		$6$ & $5,5,5,5,1,1$      & $8$         &  $42390182.0$    &  $6.3$               & $1764089.8$                    &  $238880.8$          & $0.7$\\
		$8$ & $5,5,5,5,5,2,1,1$            &$6$         &  $26534304.7$    &  $4.9$               & $ 320237.0$                    &  $116314.0$ & $1.59$ \\
		\hline 
	\end{tabular}

	\caption{
		A more detailed statistics for a sample data $n=40$, $b_{\min}$. 
		The statistics were calculated for $100$ instances.
	}
	\label{tab:details_for_flows}
\end{table}
For $b_{\avg}$ and $b_{\max}$, when we increase number of machines, the processing times also increase; cf. \Cref{tab:bk5} and \Cref{tab:bk25}.
However, for $b_{\min}$, it is not the case, see \Cref{tab:details_for_flows}.
Observe that for $b_{\min}$ increasing the number of machines, decreases the number of blocks that has to be formed.
Hence the algorithm tries up to $4^{8.4}$, $6^{6.3}$, $8^{4.9}$ distinct assignments, respectively; $4^{8.4} >  6^{6.3} > 8^{4.9}$.
Another fact is that when we increase the number of machines, with constant bound on speeds and fixed number of vertices, the size of the set of $\cmaxcost$ candidates becomes stable; however, the \YES answer becomes more likely.
By \YES answer, we mean that for a given guess of $\cmaxcost$, there exists a schedule.
There is fundamental difference between \YES and \NO answer with respect to computations times.
To answer \YES, it is enough to find one valid assignment of cut-vertices and a schedule consistent with the assignment; to answer \NO, we have to check all the assignments. 
Also due to this, for $b_{\min}$, the number of assignments checked should decrease when the number of machines increase.

However, the assignments of cut vertices are verified with respect to guessed $\cmaxcost$ and incompatibility relation and only valid ones are used.
Hence, by increasing the number of machines, the number of \emph{valid} assignments should increase. 
On the other hand, due to more more prevalent \YES answers, the number of valid assignments checked could become lower.
Interestingly, the maximum average number of valid assignments checked was for $m=6$, which also corresponds to maximum average procesing time, and to minimum ratio of \YES to \NO answers.
We think that a deeper theoretical analysis and more measurements are needed to characterize the behaviour of the algorithm.
In particular an analysis how much more costly in time terms is \NO answer comparing to \YES answer.
If a difference is huge, then perhaps instead of bisection it would be better to from some point iterate over candidates for $\cmaxcost$ starting from the biggest one.

In our example, for $m=4$, $m=6$, $m=8$, on average $1.5\%$, $13.5\%$, $36.3\%$ of assignments were valid, respectively.
Hence, a proposed heuristics is to use greedy assignment / branch-and-bound algorithm to filter out more assignment of cut-vertices.
A assignment of even some cut-vertices gives a lower bound on the length of any schedule consistent with this assignment.
Hence, if it not better than the best $\cmaxcost$ found so far, there is no point in even considering this assignment any further.
Moreover, the assignment has to be consistent with the incompatibility relation.
There is not point to continue generating the assignments if a part of the assignment is not compatible with the relation. 

Finally, for the measured instances, Boykov-Kolmogorov was consistently faster, by between 10\% and 20\%, cf. \Cref{tab:bk25} and \Cref{tab:goldberg}

\section{Open problems}

Although we have explored the status of scheduling problems with conflict block graphs to some extent, there are still some related problems remaining to be solved.
For example, we are interested in decreasing an approximation factor of an algorithm for $P|G = \blockgraph|\cmaxcost$ or proving some inapproximability results.
It is also interesting to know whether there exists a PTAS for $Q|G = \blockgraph, p_j = 1|\cmaxcost$ or any constant factor approximation algorithm for $Q|G = \blockgraph|\cmaxcost$.
Other results in the area that is the subject of this paper would be also very welcomed.

\bibliographystyle{apalike}
\bibliography{blocks}

\end{document}